\documentclass[acmsmall]{acmart}
\setcopyright{none}
\renewcommand\footnotetextcopyrightpermission[1]{}

\usepackage{amsmath}
\usepackage{amssymb}
\usepackage{semantic}
\usepackage{multicol}
\usepackage{esvect}
\usepackage{float}
\usepackage{environ}
\usepackage{centernot}
\usepackage{scalerel}
\usepackage{xifthen}
\usepackage{courier}
\usepackage{listings}
\usepackage{placeins}
\usepackage{xassoccnt}
\usepackage{stmaryrd}
\usepackage{mathtools}
\usepackage{enumitem}
\usepackage{xspace}
\usepackage[utf8]{inputenc}

\newcommand{\llbar}{\{\kern-0.5ex|}
\newcommand{\rrbar}{|\kern-0.5ex\}}

\begingroup
\catcode`\|=\active
\gdef|{\mid}
\gdef\|{\mid}
\endgroup
\AtBeginDocument{\mathcode`|="8000 }

\renewcommand{\>}{\rangle}
\renewcommand{\<}{\langle}

\NewEnviron{abstractsyntax}{%
  \vspace{-1.5em}
  \begin{align}
  \begin{split}
    \BODY
  \end{split}
  \end{align}
}

\NewEnviron{abstractsyntax*}{%
  \vspace{-1.5em}
  \begin{align*}
    \BODY
  \end{align*}
}

\newcommand{\trans}[1][]{\xrightarrow{#1\phantom{}}}
\newcommand{\nottrans}[1][]{\centernot{\trans[#1]}}

\newcommand{\rulestable}[5][]{
    \bgroup
    \newcommand{\tableCols}{ll}
    \ifthenelse{\isempty{#1}}{
        \renewcommand{\tableCols}{ll}
    }{
        \renewcommand{\tableCols}{#1}
    }
    \renewcommand{\arraystretch}{2.5}
    \ifthenelse{\isempty{#4}}{
    }{
        \renewcommand\arraystretch{#4}
    }
    \begin{table}[htpb]
        \centering
        \begin{tabular}{\tableCols}
            #5\vspace{3mm}
            \\ 
        \end{tabular}
        \vspace{.1cm}
        \caption{#2}
        \label{#3}
    \end{table}
    \egroup
}
\renewcommand{\sf}[1]{\textsf{{#1}}}

\newcommand{\rt}[2][]{%
\ifthenelse{\isempty{#1}}{\textsc{(#2)}}{($\textsc{#2}_\textsc{#1}$)}
}

\newcommand{\eqdef}{\overset{\underset{\mathrm{def}}{}}{=}}

\newcommand{\lin}{\textsf{lin}}
\newcommand{\dom}[1]{\text{dom}(#1)}

\renewcommand{\U}{\mathcal{U}}
\newcommand{\W}{\mathcal{W}}

\newcommand{\usage}{\mathcal W}
\newcommand{\envf}{\ensuremath{env_F}\xspace}

\newcommand\usageterm{\textsf{usage}}
\newcommand*{\der}[2]{{#1}.{#2}}

\newcommand{\vddash}{\vdash_{\vv{D}}}
\newcommand{\vdhash}{\vdash^{h}}
\newcommand{\vhdash}{\vdash^{h}}
\newcommand{\vdhpash}{\vdash^{h'}}
\newcommand{\vhpdash}{\vdash^{h'}}

\newcommand{\topclass}{\top_{C}}
\newcommand{\topusage}{\top_{\U}}
\newcommand{\toptype}{\topclass\<\bot\>[\topusage]}

\newcommand{\err}{\longrightarrow_{\text{err}}}
\newcommand{\noterr}{{\centernot\longrightarrow}_{\text{err}}}

\newcommand{\substitute}[2]{\ensuremath{\{#1/#2\}}}
\newcommand*\changeval[3]{{#1}\{#3 / #2 \}}

\newcommand{\sif}[3]{\sf{if}\ (#1)\ \{#2\}\ \sf{else}\ \{#3\}}
\newcommand{\switchp}[1]{\ensuremath{\sf{switch}_{x.m}\ (#1)\ \{l_i : e_i\}_{l_i \in L}}}
\newcommand{\switchf}[1]{\ensuremath{\sf{switch}_{f.m}\ (#1)\ \{l_i : e_i\}_{l_i \in L}}}

\newcommand{\objects}[1]{\sf{objects}(#1)}
\newcommand{\return}[1]{\sf{return}\{#1\}}

\SetEnumitemKey{ncases}{itemindent=!,before=}

\newenvironment{subproof}
  {\def\proofname{\hspace{-1ex}}%
   \setitemize{leftmargin=2.5em}
   \proof}
  {\endproof}

\newcommand{\tagthm}[1]{
\setcounter{theorem}{#1 - 1}
}

\newcommand{\lambdaeq}[1]{\ensuremath{=^{#1}}}

\newcommand{\mungo}{\textsf{Mungo}\xspace}
\newcommand{\stmungo}{\textsf{StMungo}\xspace}

\newcommand{\norm}[1]{\lvert #1 \rvert}

\lstdefinelanguage{mungo}{
  	morecomment=[l][\itshape]{//},
    keywords = {enum, class, void, bool, switch, if, end, new, true, false, continue, loop, unit, return, infer}
}
\title{Behavioural Types for Memory and Method Safety in a Core Object-Oriented Language}

\author{Adrian Francalanza}
\affiliation{
  \institution{University of Malta}
}
\email{adrian.francalanza@um.edu.mt}

\author{António Ravara}
\affiliation{
  \institution{Universidade Nova de Lisboa}
}
\email{aravara@fct.unl.pt}

\author{Hans H\"{u}ttel}
\affiliation{
  \institution{Aalborg University}
}
\email{hans@cs.aau.dk}

\author{Iaroslav Golovanov}
\affiliation{
  \institution{Aalborg University}
}
\email{igolov14@student.aau.dk}

\author{Mario Bravetti}
\affiliation{
  \institution{Universitá di Bologna}
}
\email{bravetti@cs.unibo.it}

\author{Mathias Steen Jakobsen}
\affiliation{
  \institution{Aalborg University}
}
\email{msja15@student.aau.dk}

\author{Mikkel Klinke Kettunen}
\affiliation{
  \institution{Aalborg University}
}
\email{mkettu16@student.aau.dk}

\begin{abstract}
  We present a type-based analysis ensuring memory safety and object
  protocol completion in the Java-like language \mungo. Objects are
  annotated with usages, typestates-like specifications of the
  admissible sequences of method calls.  The analysis entwines
  usage checking, controlling the order in which methods are called,
  with a static check determining whether references may contain null
  values. The analysis prevents null pointer dereferencing
  and memory leaks and ensures that the intended usage protocol of
  every object is respected and completed. The type system
  has been implemented in the form of a type checker.
\end{abstract}

\begin{document}
\maketitle

\thispagestyle{empty}

\section{Introduction}
\label{sec:introduction}

The notion of references is central to object-oriented programming.
This also means that object-oriented programs are particularly prone to the
problem of \emph{null-dereferencing}~\cite{Hoare}: 
a recent survey 
~\cite[Table 1.1]{SunshineAldrich} analysing questions
posted to StackOverflow referring to 
\textsf{java.lang}
exception types  notes that, as of 1 November 2013, the most common
exception was precisely null-dereferencing.

Existing approaches for preventing null-dereferencing require
annotations, \emph{e.g.,} in the form of pre-conditions or type qualifiers, together with auxiliary reasoning methods.
\emph{E.g.,} F\"{a}hndrich and Leino \cite{Fahndrich2003} use type qualifiers
and a data flow analysis to determine if fields are used 
safely, while Hubert \emph{et al.}
rely on a constraint-based flow analysis~\cite{hubert2008}. 
These approaches all rely on programmer intervention which can be viewed as
a limitation, as the absence of null-dereferencing does not come
``for free'' in well-typed programs.
Static analysis tools like the
Checker framework~\cite{DietlDEMS11} can be used to
check the code once it is in a stable state.
However, both type qualifiers and static analyses suffer from a common
problem: 
to be compositional they are either rather restrictive or
demands a lot of explicit checks in the code, resulting in
a ``defensive programming'' style.

By including the analysis as part of the type system, one obviates
the need for additional annotations or auxiliary reasoning mechanisms. 
The Eiffel type system~\cite{DBLP:journals/cacm/Meyer17} now distinguishes
between attached and detachable types: variables of an attached type
can never be assigned to a void value, which is only allowed for
variables of detachable type.
However, it is the execution of a method body that typically
changes the program state, causing object fields to become
nullified. 
The interplay between null-dereferencing and the
\emph{order} in which methods of an object are invoked is therefore
important. 
A recent manifestation of this is the bug found found in Jedis~\cite{jedis1}, 
a Redis~\cite{redis1} Java client, 
where a \textsf{close} method could be called even after a
socket had timed out \cite[Issue 1747]{jedis1}.

One can therefore see 
an object as following a \emph{protocol} describing the
\emph{admissible sequences of method invocations}. 
The intended protocol can be expressed as a \emph{behavioural type}, which allows for a rich
type-based approach that can be used to ensure memory safety via static type
checking~\cite{HuttelEtal,AnconaEtal}. 
%
%
In this paper we present (to our knowledge) the first behavioural
type system for a realistic object-oriented language that rules out
null-dereferencing and memory leaks as a by-product of a
safety property, \emph{i.e.,} protocol fidelity, and of a (weak) liveness
property, \emph{i.e.,} object protocol completion for terminated
programs. 
Protocol fidelity is an expected property in behavioural type systems,  
but it does not usually guarantee \emph{protocol completion} for mainstream-like languages, 
nor does it usually guarantee memory safety, \emph{i.e.,} no \emph{null-dereferencing} or \emph{memory leaks}.


There are two main approaches to behavioural type systems. 
The notion of \emph{typestates} has its origin in the work by Strom and Yemini~\cite{StromYemini}; 
the idea is to annotate the type of an object with information pertaining to its current state. 
Later work includes that of Vault~\cite{DBLP:conf/pldi/FahndrichD02}, 
Fugue~\cite{DBLP:conf/ecoop/DeLineF04} and 
Plaid~\cite{plaid,DBLP:conf/oopsla/SunshineSNA11}. 
Garcia \emph{et al.}~\cite{GarciaEtal} describe a gradual typestate system that
integrates access permissions with gradual types~\cite{DBLP:conf/ecoop/SiekT07} to control aliasing in a more robust
way.
%

There appears to have been little work on memory safety; the closest is work by
Suenaga \emph{et al.}~\cite{Suenaga:2012:TSR:2398857.2384618,Suenaga:2016}
studies behavioural type systems for memory deallocation in an imperative language without objects.

The second approach taken is that of \emph{session types}~\cite{Honda}. 
This 
originated in a $\pi$-calculus setting where the session type of a channel is a protocol that
describes the sequence of communications that the channel must follow. 
Channels are \emph{linear}: they are used exactly once
with the given protocol, and evolves as a result of each communication that it is involved in.

The work by Gay \emph{et al.} \cite{GayEtal} is a meeting of the two
approaches; the type of a class $C$ is annotated with a \emph{usage}
denoting a session that any instance of $C$ must follow. 
The type system ensures a \emph{fidelity} property, where the actual sequences of
method calls to any object will follow those specified by its usage.
The \mungo/\stmungo tool suite~\cite{KouzapasEtal} is based
on this work for a subset of Java. 
Objects are used linearly: an object reference can be written to
and subsequently read at most once. 
This controls aliasing and permits compositional reasoning via the type system,
 making it possible to statically verify properties such 
object usage fidelity.   
However, the \mungo type system does \emph{not} provide
guarantees ruling out null-dereferencing and memory leaks.
In particular, three unpleasant behaviours are still allowed: $(i)$ null-assignment to a
field/parameter containing an object 
with an incomplete typestate;
$(ii)$ null-dereferencing (even \textsf{null.m()}); $(iii)$ not returning
 an object 
with incomplete protocol from a method call.

In this work we present a new behavioural type system for \mungo that
handles these outstanding issues. 
Following preliminary work by Campos, Ravara and
Vasconcelos~\cite{DBLP:journals/corr/abs-1110-4157,DBLP:conf/sac/VasconcelosR17}, our type system guarantees
that well-typed programs never give rise to null-dereferencing,
memory leaks, or incorrect method calls.
%
The system has been implemented in the form of a type checker.
\footnote{\texttt{\url{https://github.com/MungoTypesystem/Mungo-Typechecker}}}
This is the first “pure” type-checking system for \mungo. 
Previous systems did not 
support typestates for method parameters and return types. 
Moreover, the safety guarantees of our system ensure both memory safety (a novelty) and protocol fidelity, while enforcing a weak liveness property for methods (another novelty), \emph{i.e.,} protocol completion for terminated programs. 
The approach and techniques developed for this work may easily be adaptable to other languages and systems.

The rest of our paper is structured as follows. Section
\ref{sec:related} discusses related work, while Section \ref{sec:obs}
describes the underlying ideas of our type system. In Section
\ref{sec:syntax} we present the syntax of \mungo, while Section
\ref{sec:semantics} presents its operational semantics. Section
\ref{sec:typesystem} describes the central features of our type
system, while Section \ref{sec:results} outlines the major soundness
results of the system. Section \ref{sec:implementation} briefly
presents our implementation, and Section \ref{sec:discussion} contains
the conclusion of our paper.


\section{Related Work}
\label{sec:related}

In our work we consider $\sf{null}$ as a value of a specific $\bot$ type (and only of that type). This idea is related to that of Vekris et al.~\cite{Vekris:2016:RTT:2980983.2908110} who modify the type system of Typescript, which uses undefined and null types to represent the eponymous values, so that such types are not related (by subtyping) to any other type: in this way misuse of such a values can be effectively detected by the type system. Our approach differs in that our type system is behavioural; fields follow a protocol of linearity, where the field type $\bot$ type indicates the introduction of a $\sf{null}$ value at the level of the semantics.

A notion similar to our notion of usages can be found in Lerner et al.~ who use a similar mechanism for controlling access to potentially uninitialised references and to reason about null references~\cite{Lerner:2013:TRT:2578856.2508170}. They use the idea in a setting that attempts to retrofit (different notions of) static typing analyses for Javascript. More concretely, they use regular expression types to describe objects that evolve over the course of their lifetime, and employ modalities to describe the presence of fields in an object such as "possibly present", "definitely present either on this object or via the prototype chain" and "cannot be accessed”. 

The idea of considering that objects have usage protocols is more closely related to how typestates are used in the languages and Plural and Plaid~\cite{GarciaEtal,plaid}. There, the approach is to define the abstract states an object can be in, and associate with methods pre- and post-conditions that describe in which states a method can be called and the state that will be the result of the execution of a method. Statically, a type system checks the correctness of the code with respect to the assertions and that client code follows the protocols of the objects it uses.

However, this approach requires a demanding logical language~\cite{DBLP:conf/oopsla/SunshineSNA11} for pre- and postconditions. It is not always not obvious or even clear which post-conditions imply which pre-conditions, so that the protocol (as, for instance, an automaton) is not easily visualisable from the annotated code.

In contrast, our approach uses simply usage annotations that are checked directly in the type rules. As method bodies are type-checked in the rules for class declarations by following class usage annotations, this prevents the problem that follows from the possibility that some pre-conditions might not even be implied by any post-conditions, namely the presence of unreachable code~\cite{DBLP:journals/eceasst/SiminiceanuAC12}.


\section{The importance of session completion} \label{sec:obs}

A crucial observation for our work is that protocol fidelity and null-dereferencing are
interdependent. 
Particularly, we require that a well-typed program never leaves the
protocol incomplete by losing a reference to an object at some point
during its execution.
To see the importance of this, consider the class 
\textsf{File} in Listing \ref{listing:file1} with the usage defined on lines 3-9.
Here, the type~$C[\U]$ of a class~$C$ contains a \emph{usage}~$\U$ that specifies the admissible sequences of method calls for any object instance of~$C$.

\lstset{escapeinside={(*@}{@*)},mathescape=true, tabsize=4,morekeywords={enum, end}}
\begin{lstlisting}[label={listing:file1},caption={An example class describing files},firstnumber=1]  
enum FileStatus { EOF NOTEOF }

class File {
  {open; X} [ (*@\label{lst:code:4}@*)
    X = {
      isEOF; <EOF: {close; end} 
              NOTEOF: {read; X}>
    }
  ]

  void open(void x) { ... }
  FileStatus isEOF(void x) { ... }
  Char read(void x) { ... }
  void close(void x) { ... }
}
\end{lstlisting}

The usage tells us that the object must first be
\textsf{open}ed, going then into state \textsf{X} (a usage variable) where only the method \textsf{isEOF} can be called. 
When
\textsf{isEOF} is called the continuation of the protocol depends on
the method result: 
if it returns \textsf{EOF} one can only 
\textsf{close} the file (line 6); 
if it returns \textsf{NOTEOF} one can
\textsf{read} and check again (line~7). 
This ensures that all methods
of a \textsf{File} object are called according to the safe order
declared. 
%
In particular, a null-dereferencing caused by calling
\textsf{read} before \textsf{open} is prohibited, and no memory leaks
can be caused by dropping an object or leaving it open before the
protocol is completed.

To achieve this, \emph{both} methods and \emph{fields} should obey their
associated protocols. 
Within any instance of a class, its field objects must respect their usage protocols and executed to completion, whereas methods must be invoked in the expected usage order until completion. 
Importantly, avoiding field null-dereferencing is crucial to attain this.
In Listing \ref{listing:file2} a class \textsf{FileReader} is
introduced. 
The usage information on line \ref{lst:code:0} requires that the \textsf{init()} method
is called first, followed by method \textsf{readFile()}.
 \lstset{escapeinside={(*@}{@*)},mathescape=true}
 \begin{lstlisting}[label={listing:file2},caption={An example class
   intended for reading files},firstnumber=17]
class FileReader {
  {init ; {readFile; end}}[] (*@\label{lst:code:0}@*)
  
  File file
  
  void init() {
    file = new File (*@\label{lst:code:1}@*)
  } (*@\label{lst:code:3}@*)
  
  void readFile() {
    file.open(unit);
    loop:
      switch(file.isEOF()) {
        EOF: file.close() (*@\label{lst:code:2}@*)
        NOTEOF: file.read();
            continue loop
      }
  }
}
\end{lstlisting}
One must also use objects of type \textsf{File} with usage information on line \ref{lst:code:4}.
%
%
Method \textsf{readFile} follows the protocol by first
\textsf{open}ing the \textsf{file}, and then entering a loop (lines 28
to 33) to read until its end. Then, it \textsf{close}s it.

Combining the checks on following the protocols of the
\textsf{FileReader} and the \textsf{File} is not trivial. 
Crucially, simply requiring protocol fidelity does not suffice to avoid certain problems: 
%
\begin{itemize}
\item Protocol fidelity allows the omission of line \ref{lst:code:1} that initialises field \textsf{file} with a
  new \textsf{File} object.  However, calling \textsf{read()} on field
  \textsf{file}, even if we follow the prescribed protocol of
  \textsf{FileReader} and invoke \textsf{init()} first, would result
  in a null-dereferencing. 
\item Similarly, protocol fidelity also permits the addition of the command
  \lstinline{file=null} at line \ref{lst:code:3}. 
  This would result
  in not only losing the reference to
  the object, but also lead to null-dereferencing again.
\item Protocol fidelity permits the addition of
  \lstinline{file=new File()} at line \ref{lst:code:2} before
  calling \lstinline{close()}.  Once more, this would result in
  losing the reference to the previous object stored in
  \textsf{file} that has not yet completed its protocol.
\end{itemize}

Usages are vital for permitting a \emph{finer} analysis that
guarantees the absence of null dereferencing that is also
\emph{compositional}. To that end, our proposed analysis provides such
guarantees for programs following the intended usage
protocols. Moreover, we are able to do so by analysing the source code
of each individual class \emph{in isolation}.
\section{The \mungo language: Syntax}
\label{sec:syntax}

\mungo is a Java-like programming language and programming tool developed at Glasgow University by Kouzapas et al. \cite{KouzapasEtal}. In this section we present the syntax of the \mungo language, including the notion of usages that will be central to our type system. The formation rules are given in Figure \ref{fig:syntax}.

A \mungo \emph{program} or declaration $\overrightarrow{D}$ is a
sequence of \emph{enumeration declarations} followed by a sequence of class declarations. An important distinction from the version of the \mungo language of \cite{KouzapasEtal} is that we now allow generic classes.

A \mungo \emph{class declaration} $\sf{class} \<\alpha[\beta]\>\ C\{\U, \vv{F}, \vv{M}\}$, where $C\in \sf{CNames}$, consists of a class name $C$ and defines a set of methods $\vv{M}$, a set of fields $\vv{F}$ and an usage $\U$; this we will describe later. The variables $\alpha$ and $\beta$ denote class and usage variables respectively and are used for our account of generic classes.

A method named $m$ is declared as $t_2 \ m(t_1 \ x)\{e\}$, where $t_1$ denotes the argument type of $m$, while the return type is $t_2$. The body of $m$ is an expression $e$. In our presentation we do not distinguish between statements and expressions; we allow values $v$ as expressions but we also include statement-like expressions such as conditionals $\sf{if}\ (e) \ \{e\} \ \sf{else} \ \{e\}$, selection statements $\sf{switch}_{r.m} \ (r.m(e)) \ \{l_i : e_i\}_{i_i \in L}$, sequential compositions $e;e$, method calls $r.m(e)$ and assignments $f \ = \ e$. Note that assignment is always made to a field, meaning that assignment can only happen to a field inside the object on which a method is evaluated. Objects therefore cannot access the fields of another object directly. 

Following \cite{KouzapasEtal}, iteration in \mungo is possible by means of jumps. Expressions can be labelled as $k : e$, where $k$ must be a unique label; the execution can then jump back to the expression labelled $k$ by evaluating the expression $\sf{continue} \ k$. 

Classes in \mungo are instantiated with either $\sf{new}\ C$ for simple classes or $\sf{new}\ C\<g\>$ for generic classes. 

The set \textbf{Values} is ranged over by $v$ and denotes the set of all values and contains boolean values, \sf{unit}, labels and \sf{null}. References ranged over by $r$ describe how objects are accessed, either as method parameters ranged over by $x \in \textbf{PNames}$ or field names ranged over by $f \in \textbf{FNames}$. The enumeration types are ranged over by $L\in \textbf{ENames}$; each enum contains a set of labels $\vv{l}$. 

Fields, classes and methods are annotated with \emph{types}. The set of base types \textbf{BTypes} consists of the type for booleans \textsf{bool}, the \textsf{void} type and enumerations $L$. A central component of the behavioural type system presented later is that of typestates $C\<t\>[\U] \in \textbf{Typestates}$. Typestates appear as part of class declarations and can involve generic classes; we use the typestate notion to capture the current value of class and usage parameters $\alpha[\beta]$ by using a type parameter in the typestate $C\<t\>[\U]$, where $t$ is the actual parameter of the class and $\alpha$ and $\beta$ are type and usage variables, respectively.

\begin{figure}
\begin{multicols}{2}
\noindent
\begin{abstractsyntax*}
{D} ::=\ &\sf{enum}\ L\ \{\tilde{l}\}  \mid  \sf{class}\< \alpha[\beta] \>\ C\ \{\mathcal{U}, \vv{M}, \vv{F}\} \\ 
   \mid \ &\sf{class}\ C\ \{\mathcal{U}, \vv{M}, \vv{F}\}  \\ 
F ::=\ &z \ f \\
M ::=\ &t_2 \ m(t_1 \ x)\{e\}\\
v ::=\ &\sf{unit}  \mid  \sf{true}  \mid  \sf{false}  \mid  l  \mid  \sf{null} \\
r ::= \ &x  \mid  f \\
{e} ::= \ &v  \mid  r  \mid  \sf{new} \ C  \mid  \sf{new} \ C\<g\>  \mid  f \ = \ e  \mid  r.m(e)  \\
& \mid  e;e  \mid  \sf{if}\ (e) \ \{e\} \ \sf{else} \ \{e\}  \\
& \mid  \sf{switch}_{r.m} \ (r.m(e)) \ \{l_i : e_i\}_{i_i \in L}  \\
       \mid  \ &k : e  \mid  \sf{continue} \ k 
\end{abstractsyntax*}
\columnbreak%
\begin{abstractsyntax*}
b ::=\ &\sf{void}  \mid  \sf{bool}  \mid  L \\
{g} ::=\ & \eta[\W]  \mid  \alpha [\beta] \\
{\eta} ::=\ &C  \mid  C\<\eta[\W]\>  \mid  C\<\alpha[\beta]\>\\
{z} ::=\ &b  \mid  \eta  \mid  \alpha\\
{t} ::=\ &b  \mid  \eta[\U]  \mid  \alpha[\beta] \\
u ::=\ &\{m_i;w_i\}_{i \in I}  \mid  X \\
w ::=\ &u  \mid  \langle l_i : u_i\rangle_{l_i \in L} \\
E ::=\ &X=u \\
\mathcal{U} ::=\ &u^{\vv{E}}
\end{abstractsyntax*}
\end{multicols}
\caption{Syntax of Mungo}
\label{fig:syntax}
\end{figure}

A central notion from Kouzapas et al. is that of \textit{usages}; in the syntax usages are ranged over by $\U$ and $\W$. A usage specifies the admissible sequences of method calls allowed for an object. A usage $\{m_i; w_i\}_{i \in I}$ describes that any one of the methods $m_i$ can be called, following which the usage is then $w_i$. A usage $\<l_i : u_i\>$ specifies that the usage continues as $u_i$ depending on a label $l_i$. This is useful when the flow of execution depends on the return value of a function. To allow for iterative behaviour, we introduce recursive usages. A defining equation $X = u$ specifies the behaviour of the usage given by the usage variable $X$; $X$ can occur in $u$. 

Usages have a labelled transition semantics that appears in the semantics in Section \ref{sec:semantics} as well as in the type system in Section \ref{sec:typesystem}. Transitions are of the form $\U \trans[m] \U'$ and $\U \trans[l] \U'$. The available transitions for usages are shown in Table \ref{tab:classtransitions}.

\rulestable[llll]{Transitions for usages}{tab:classtransitions}{1}{
    \rt{Branch} & \inference{j \in I}{\{m_i : w_i\}_{i\in I}^{\vv{E}} \trans[m_j] w_j^{\vv{E}}} &
    \rt{Unfold} & \inference{u^{\vv{E}\cup \{X = u\}} \trans[m] \W}{X^{\vv{E}\cup \{X = u\}} \trans[m] \W} \\
    \multicolumn{4}{c}{\rt{Sel} \quad \inference{l_j \in L}{\<l_i : u_i\>_{l_i \in L}^{\vv{E}} \trans[l_j] u_j^{\vv{E}}}}
}

We write $\W \rightarrow^{*} \W'$ if usage $\W$ evolves to $\W'$ using zero or more usage transitions.

\begin{example}\label{ex:fileusage}
  Consider the \textsf{File} class from Listing \ref{listing:file1}. Its usage $\U$ is
  %
  \[ \U = \{\textsf{open}; \textsf{X}\}^{\vv{E}}\]
  \[ \vv{E} = \{\textsf{X} = \{\textsf{isEOF}; \left\langle\begin{array}{l} \textsf{EOF} : \{\textsf{close};\textbf{end}\}  \\\textsf{NOTEOF} : \{\textsf{read}; \textsf{X}\}\end{array} \right\rangle \}\}\]

  The only transition here is $\U \trans[\mathsf{open}] \U'$ where $\U' = \textsf{X}^{\vv{E}}$. Moreover, by applying the \rt{Unfold} rule we have that $\U' \trans[\mathsf{isEOF}] \U''$ where
  $ \U'' = \langle \textsf{EOF} : \{\textsf{close}; \textbf{end}\} \ \textsf{NOTEOF} : \{\textsf{read}; \textsf{X}\} \rangle^{\vv{E}}$. 
  
\end{example}




\section{The \mungo language: Semantics}
\label{sec:semantics}

In this section we present the operational semantics of \mungo. The semantics is a small-step semantics that uses a stack-based state model.

\subsection{Run-time syntax}  

We extend the syntax to include \textit{run-time expressions} and values in order to describe partially evaluated expressions and thereby the evaluation order of expressions. The syntactic categories of values and expressions are extended as shown in Figure \ref{fig:runtime-syntax}.

\begin{figure}
\begin{abstractsyntax*}
    v ::=\ &\dots | o \\
    e ::=\ &\dots | \sf{return}\{e\} | \sf{switch}_{r.m}\ (e) \{l_i : e_i\}_{l_i \in L}
\end{abstractsyntax*}
\caption{Run-time syntax of Mungo}
\label{fig:runtime-syntax}
\end{figure}

The $\sf{return}\{e\}$ construct allows us to encapsulate the ongoing evaluation of a method body $e$. We also introduce a general \sf{switch} construct that allows arbitrary expressions in \sf{switch} statements. This will allow us to express the partial evaluation of \sf{switch} constructs. Finally, we allow for objects $o$ to appear as values in expressions. This will allow us to express that the evaluation of an expression can return an object that can then be used.

\subsection{The binding model of the semantics}

Transitions in our semantics are relative to a program definition $\vv{D}$ and are of the form $$\vddash \<h, env_S, e\> \rightarrow \<h',env_S', e'\>$$
where $h$ is a heap, $env_S$ is a parameter stack and $e$ is the run-time expression that is being evaluated.

A \emph{heap} records the bindings of object references. In the heap, every object reference $o$ is bound to some pair ($(C\<t\>[\W], \envf)$ where $C\<t\>[\W]$ is a typestate and $\envf\in\textbf{Env}_{\textbf{F}}$ is a field environment. A typestate is the current state of an object with respect to its class definition and current usage. A \emph{field environment} is a partial function that maps field names to the values stored in the fields.


We introduce notation for selectively updating components in the image of a heap. Given a heap $h = h' \uplus \{ o \mapsto (C\<t\>[\W], \envf) \}$, we use the notation $\changeval{h}{\der{h(o)}{\usageterm}}{\usage'}$ to stand for the heap $h' \uplus \{ o \mapsto (C\<t\>[\W'], \envf) \}$.\footnote{We use
  $\uplus$ to denote disjoint union.} We write $\changeval{h}{o.f}{v}$ for $h' \uplus \{ o \mapsto (C\<t\>[\W],\envf\{f \mapsto v\})\}$. And we use the notation $h\{o.f \mapsto C\<t\>[\U]\}$ to denote the heap $h' \uplus \{ o \mapsto  (C\<t\>[\W'], \envf') \}$ where $\envf' = \envf \{ f \mapsto C\<t\>[\U] \}$.

The \emph{parameter stack} records to the bindings of formal parameters. It is a sequence of bindings where each element $(o,s)$ contains an \emph{active object} $o$ and a parameter instantiation $s = [x \mapsto v]$. We define $\varepsilon$ to be the empty stack. We let $env_S\cdot env_S'$ denote the concatenation of the stacks $env_S'$ and $env_S$. 

    
%

In a parameter stack $env_S \cdot (o, s)$ we call the bottom element $o$ the \emph{active element}. Often, we shall think of the parameter stack as defining a function. The domain $\text{dom}(env_S)$ of the parameter stack $env_S$ is the multiset of all object names on the stack. The range of the parameter stack $\text{ran}(env_S)$ is the multiset of all parameter instantiations on the stack.




\begin{definition}[Class information]Let $C$ be a class.
If $\sf{class}\< \alpha[\beta] \>\ C\ \{\mathcal{U}, \vv{F}, \vv{M}\}$, then $C$ is a generic class, and we define the following functions. If $\sf{class}\  C\ \{\mathcal{U}, \vv{F}, \vv{M}\}$ then only the functions with $C\<\bot\>$ are defined.

\begin{multicols}{2}
\noindent
\begin{align*}
    C\<\eta[\U']\>.\sf{methods}_{\vv{D}} &\eqdef \vv{M}\{\alpha[\beta] / \eta[\U']\} \\
    C\<\eta[\U']\>.\sf{fields}_{\vv{D}} &\eqdef \vv{F}\{\alpha / \eta\} \\
    C\<\eta[\U']\>.\sf{usage}_{\vv{D}} &\eqdef \U
\end{align*}%
\columnbreak%
\begin{align*}
    C\<\bot\>.\sf{methods}_{\vv{D}} &\eqdef \vv{M} \\
    C\<\bot\>.\sf{fields}_{\vv{D}} &\eqdef \vv{F}\\
    C\<\bot\>.\sf{usage}_{\vv{D}} &\eqdef \U
\end{align*}
\end{multicols}


\end{definition}

Likewise, we introduce the following notation that lets us refer to the attributes of an object $o$ bound in a heap $h$. 

\begin{definition}[Heap information]{Let $h(o)=\<C\<t\>[\W], \envf\>$, then:}

\begin{center}
\def\arraystretch{1.5}
\begin{tabular}{ll}
    $h(o).\sf{class} \eqdef C\<t\>$ & $h(o).\sf{env}_{\sf{F}} \eqdef \envf$  \\
    $h(o).\sf{usage} \eqdef \W$ & $h(o).f \eqdef \envf(f)$ \\
    \multicolumn{2}{c}{$h(o).\sf{fields} \eqdef \dom{\envf}$}
\end{tabular}
\end{center}
\end{definition}

We define $\vv{F}.\sf{initvals}_{\vv{D}}$ to be a field environment where fields are set to an initial value, which is null for class types, \sf{false} for booleans and \sf{unit} for void types. $\vv{F}.\sf{inittypes}_{\vv{D}}$ is a field type environment, where fields have an initial type -- $\bot$ for classes, \sf{Bool} for booleans and \sf{void} for unit. 

\subsection{The transition rules}

We now define the set of valid transitions. The transition rules are split into four parts. First, we present the transition rules describing field access. Here the concept of linearity is also described in detail as it is especially important for assignment and dereferencing. Then, we describe the transition rules for method calls and control structures. Finally, the transition rules for composite expressions are presented.

\paragraph{Field access}

An important aspect of our semantics is that field access must follow a protocol of \emph{linearity}, if the field value is that of an available object. If the field denotes a terminated object or a ground value, field access is \emph{unrestricted}.

\begin{definition}[Linearity]
A value $v$ is said to be linear w.r.t. a heap $h$ written $\lin(v, h)$ iff $v$ has type $C\<t\>[\U]$ in $h$ and $\U \neq \sf{end}$.
\end{definition}

In Table \ref{tab:GroundReductionField} we define the ground reduction rules for parameters, fields and objects. In \rt{New} we use the type $\bot$ in order to type the class as non-generic. The rules for reading linear parameters \rt{lParam} and reading linear fields \rt{lDeref} illustrate the linearity concept we enforce, where we update a linear field or parameter to \sf{null} after we have read it. Additionally, in the rule for assigning to fields \rt{Upd} we check that the value we overwrite is not linear.

\rulestable[ll]{Ground reduction rules for fields and objects}{tab:GroundReductionField}{2.5}{
\rt{uParam} & \inference{\neg \lin(v,h)}{ \vdash_{\vv{D}} \<h,(o,[x \mapsto v]) \cdot env_S, x \> \trans \<h, (o, [x \mapsto v])\cdot env_S, v\>} \\
    \rt{lParam} & \inference{\lin(v,h)}{ \vdash_{\vv{D}} \<h,(o,[x \mapsto v]) \cdot env_S, x \> \trans \<h, (o, [x \mapsto \sf{null}])\cdot env_S, v\>} \\
    \rt{uDeref} & \inference{h(o).f = v & \neg \lin(v,h)}{\vdash_{\vv{D}} \<h, (o,s)\cdot env_S, f\> \trans \<h, (o, s)\cdot env_S, v\>} \\
    \rt{lDeref} & \inference{h(o).f=v & \lin(v,h)}{\vdash_{\vv{D}} \< h, (o, s) \cdot env_S, f\> \trans \< h\{\sf{null} / o.f \}, (o,s)\cdot env_S, v \>} \\
    \rt{Upd} & \inference{h(o).f = v' & \neg \lin(v', h)}{\vddash \<h, (o,s) \cdot env_S, f=v \> \trans\<h\{v/o.f\}, (o,s)\cdot env_S, \sf{unit}\>} \\
    \rt{New} & \inference{o\text{ fresh} & C\<\bot\>.\sf{fields}_{\vv{D}} = \vv{F} & C\<\bot\>.\sf{usage}_{\vv{D}} = \W}{\vddash \<h, env_S, \sf{new}\ C \> \trans \<h \cup \{o \mapsto \<C\<\bot\>[\W], \vv{F}.\sf{initvals}_{\vv{D}}\>\}, env_S, o\>} \\
    \rt{NewGen} & \inference{o \ \text{fresh} & C\<g\>.\sf{fields}_{\vv{D}} = \vv{F} & C\<g\>.\sf{usage}_{\vv{D}} = \W }{\vddash \<h, env_S, \sf{new}\ C\<g\> \> \trans \<h \cup \{o \mapsto \<C\<g\>[\W], \vv{F}.\sf{initvals}_{\vv{D}}\>\}, env_S, o\>}
}

\paragraph{Method calls}

The ground reduction rules for method calls are found in Table \ref{tab:GroundReductionMethod}. In the rules for calling methods on parameters \rt{CallP} and fields \rt{CallF} we see the second property we enforce, namely, that objects must follow the usage described by its type. In the premise of both rules we have that a method $m$ can only be called if the usage of the object allows an $m$ transition and the result of this evaluation is that the usage of the object is updated and the next evaluation step is set to the method body $e$ by wrapping the expression in a $\sf{return}\{e\}$ statement. The \rt{Ret} rule describes how a value is returned from a method call, by unpacking the \sf{return} statement into the value, while popping the call stack.

\rulestable[ll]{Ground reduction rules for method calls}{tab:GroundReductionMethod}{4}{
    \rt{CallP} & \inference{env_S = (o, [x' \mapsto o'])\cdot env_S'  & \_\ m(\_\ x)\{e\}\in h(o').\sf{class}.\sf{methods}_{\vv{D}} \\ h(o').\sf{usage} \trans[m] \W}{\vddash \<h, env_S, x'.m(v)\> \trans \<h\{\W/h(o').\sf{usage}\}, (o', [x \mapsto v])\cdot env_S,\sf{return}\{e\}\>} \\
    \rt{CallF} & \inference{env_S=(o,s)\cdot env_S' & o' = h(o).f \\ \_\ m(\_\ x)\{e\} \in h(o').\sf{class}.\sf{methods}_{\vv{D}} & h(o').\sf{usage} \trans[m] \W}{\vddash \<h, env_S, f.m(v)\> \trans \<h\{\W/h(o').\sf{usage}\}, (o', [x \mapsto v])\cdot env_S, \sf{return}\{e\}\>} \\
    \rt{Ret} & \inference{v \neq v' \Rightarrow \neg \lin(v', h)}{\vddash \<h, (o, [x \mapsto v'])\cdot env_S,\sf{return}\{v\}\> \trans \<h, env_S, v\>}
}

\paragraph{Control Structures}

The ground reduction rules for control structures are found in Table \ref{tab:GroundReductionControl}. The rules \rt{SwF} and \rt{SwP} specify how a switch involving a field or parameter is performed. The expression $\textsf{switch}_{f.m}(l_i)$ tell us that we switch on label $l_i$ which is returned from a method call $f.m$. In both \rt{SwF} and \rt{SwP} a usage $\mathcal{U}$ is chosen based on the label $l_i$ and the updated usage is reflected in the heap by updating either field $f$ or parameter $x$. The rule \rt{Lbl} shows how a loop iteration is performed by substituting instances of $\sf{continue \ k}$ with the expression defined for the associated label. Finally, in rule \rt{Seq} we see that linearity is also relevant for sequential expressions since a value on the left-hand side of a sequential expression should be non-linear. Otherwise, we risk loosing a reference to a linear object and thereby breaking protocol completion.

\rulestable[ll]{Ground reduction rules for control structures}{tab:GroundReductionControl}{2.5}{
    \rt{IfTrue} & ${\vddash \<h, env_S, \sf{if }(\sf{true}) \{e'\}\sf{ else } \{e''\}\> \trans \<h, env_S, e'\>}$\\
    \rt{IfFls} & ${\vddash \<h, env_S, \text{if }(\sf{false}) \{e'\}\sf{ else } \{e''\}\> \trans \<h, env_S, e''\>}$\\
    \rt{SwF} & \inference{h(o).f = o' & h(o').\sf{usage} \trans[l_i] \U & l_i \in L}{\begin{gathered}\vddash \<h, (o,s) \cdot env_S, \sf{switch}_{f.m}(l_i)\{l_j : e_j\}_{l_j \in L}\> \trans \\ \hspace{3.5cm}\<h\{o.f \mapsto C\<t'\>[\U]\}, (o,s) \cdot env_S, e_i\>\end{gathered}} \\
    \rt{SwP} & \inference{h(o').\sf{usage} \trans[l_i] \U & l_i \in L}{ \begin{gathered} \vddash \<h, (o, [x \mapsto o']) \cdot env_S, \sf{switch}_{x.m}(l_i)\{l_j : e_j\}_{l_j \in L}\> \trans \\ \hspace{3.5cm} \<h\{o' \mapsto C\<t'\>[\U]\}, (o,[x \mapsto o']) \cdot env_S,e_i\>\end{gathered}}
    \\
    \rt{Lbl} & ${\vddash \<h, env_S, k:e\> \trans \<h, env_S, e\{k:e/\sf{continue}\ k\}\>}$ \\
    \rt{Seq} & \inference{\neg \lin(v,h)}{\vddash \<h, env_S, v;e\> \trans \< h, env_S, e\>}
}

\paragraph{Composite expressions}

The transition rules for composite expressions are found in Table \ref{tab:compositereductions}. The rules define a call-by-value evaluation strategy for expressions, that is, expressions are evaluated when they are first encountered: parameters are evaluated before the method body, the left-hand side of sequential composition is evaluated before the right-hand side, etc.

\rulestable[ll]{Reduction rules for composite expressions}{tab:compositereductions}{2.5}{
    \rt{FldC} & \inference{\vddash\<h, env_S, e\>\trans \<h', env_S', e'\>}{\vddash\<h, env_S, f = e\> \trans \<h', env_S', f = e'\>} \\
    \rt{MthdC} & \inference{\vddash\<h, env_S, e\>\trans \<h', env_S', e'\>}{\vddash\<h, env_S, r.m(e)\> \trans \<h', env_S', r.m(e')\>} \\
    \rt{RetC} & \inference{\vddash\<h, env_S, e\>\trans \<h', env_S', e'\>}{\vddash\<h, env_S\cdot (o, s), \return{e}\> \trans \<h', env_S'\cdot(o, s), \return{e'}\>} \\
    \rt{SeqC} & \inference{\vddash\<h, env_S, e\>\trans \<h', env_S', e'\>}{\vddash\<h, env_S, e;e''\> \trans \<h', env_S', e';e''\>} \\
    \rt{IfC} & \inference{\vddash\<h, env_S, e\>\trans \<h', env_S', e'\>}{\vddash\<h, env_S, \sif{e}{e_1}{e_2}\> \trans \<h', env_S', \sif{e'}{e_1}{e_2}\>} \\
    \rt{SwC} & \inference{\vddash\<h, env_S, e\>\trans \<h', env_S', e'\>}{\vddash\<h, env_S, \sf{switch}\ (e) \{l_i : e_i\}_{l_i \in L}\> \trans \<h', env_S', \sf{switch}\ (e') \{l_i : e_i\}_{l_i \in L}\>}
}

\subsection{Well-formed and reachable configurations}

We only consider configurations that are reachable from our initial configuration. A configuration is initial if it is the configuration right after $\sf{main}()$ was invoked in an object of class $\sf{Main}$.

\FloatBarrier

\begin{definition}[Initial configuration] An \emph{initial} configuration $ic$ is of the form 
\label{def:initial_configuration}
$$ ic = \<\{o_{\sf{main}} \mapsto \<\sf{Main}\<\bot\>[\sf{end}^{\emptyset}], \sf{Main}\<\bot\>.\sf{fields}_{\vv{D}}.\sf{initvals}_{\vv{D}}\>\}, (o_{\sf{main}}, s_{\sf{main}}), e\> $$

where:
\begin{itemize}
    \item $o_{\sf{main}}$ is the main object mapped to a typestate with the usage $[\sf{end}^{\emptyset}]$, since at this point we have just called main, and a field environment where all fields values are initial values. The initial values of fields is defined in the following way: $\sf{initval}(\sf{bool}) \eqdef \sf{false}$, $\sf{initval}(\sf{unit}) \eqdef \sf{void}$ and $\sf{initval}(C\<t\>) \eqdef \sf{null}$.
    \item $s_{\sf{main}}$ is the initial parameter binding $[x \mapsto \sf{unit}]$
    \item $\sf{void} \ \sf{main}(\sf{void} \ x)\{e\} \in \sf{Main}\<\bot\>.\sf{methods}_{\vv{D}}$.
\end{itemize}
\end{definition}
\FloatBarrier
We define reachability for both configurations and usages. A configuration $c$ is reachable if $ic \rightarrow^\ast c$, while a usage $\W'$ is reachable from $\W$ if $\W \rightarrow^{\ast} \W'$.




One of the approximations made in our analysis is that we only consider expressions where the number of $\sf{return}$ statements in $e$ does not depend on conditional execution and where objects are only mentioned once.

In order to define this, we need to define the sets of $\sf{return}$ statements and objects mentioned in an expression as functions. 

\begin{definition}[Occurrences of returns and object references] Let $e$ be an expression, then we define the following functions on $e$.  \begin{enumerate}
    \item $\sf{returns}(e)$ is the multiset of all subexpressions of the form $\textsf{return}\{e'\}$ (where $e'$ is some expression) that occur in $e$
    \item $\sf{objects}(e)$ is the multiset of the objects occurring in $e$
\end{enumerate}
\end{definition}

We call expressions that satisfy the above conditions \emph{well-formed}.

\begin{definition}[Well-formed expressions] An expression $e$ is well-formed if it satisfies the following conditions:

\begin{enumerate}
    \item For each subexpression $e'$ in $e$ we have that:
        \begin{itemize}
            \item if $e'=e'';e'''$ then $\sf{returns}(e''') = \emptyset$
            \item if $e'=\sf{if}(e'') \{e'''\} \{e''''\}$ then $\sf{returns}(e''') = \emptyset$ and $\sf{returns}(e'''')= \emptyset$
            \item if $e'=\sf{switch}_{r.m}(r.m(e''))\{l_i : e_i\}_{l_i \in L}$ then $\forall l_i \in L \ . \ \sf{returns}(e_i)=\emptyset$
        \end{itemize}
    \item $\sf{objects}(e)$ is a set
    \item If $\sf{returns}(e) \neq \emptyset$ and $\return{e'}$ is the innermost return, then we have that $\sf{objects}(e)=\sf{objects}(e')$ 
\end{enumerate}
\end{definition}

A configuration $\<h, env_S, e\>$ is well-formed if its information is consistent. The expression $e$ must be well-formed and all objects mentioned must be bound in the heap and the stack and contain the fields required by the declaration of $C$. The objects mentioned in the stack must be ones that are also mentioned in the heap. And finally, the objects must satisfy a property of no aliasing, namely that they are mentioned exactly once.

To capture the property of no aliasing precisely, we define the multiset of objects occurring in fields and the multiset of objects that occur in the parameter stack. If both these multisets are sets, we know that there can be no aliasing involved.

\begin{definition}[Object occurrences] \hfill
\begin{itemize}
    \item $\sf{objargs}(env_S)$ is the multiset of objects occurring in the parameter stack: $$\sf{objargs}(env_S) = \{o \mid \exists s \in \text{ran}(env_S) \ . \ x : o = s(x) \}$$
    \item $\sf{objfields}(h)$ is the multiset of objects occurring in fields: $$\sf{objfields}(h)=\{ o \mid \exists o' \in \text{dom}(h) \ . \ f : o = h(o').f \}$$
\end{itemize}
\end{definition}

\begin{definition}[Well-formed configurations] A configuration $\<h, env_S, e\>$ is well-formed if it satisfies the following conditions:

\begin{enumerate}
    \item $env_S = env_S' \cdot (o_{\sf{main}}, s_{\sf{main}})$ and $\norm{env_S'} = \sf{returns}(e)$
    \item $env_S$ is a set $\text{dom}(env_S) \subseteq \text{dom}(h)$
    \item $e$ is a well-formed expression
    \item $\sf{objects}(e) \Cup \sf{objfields}(h) \Cup \sf{objargs}(env_S)$ is a set and it is a subset of $\text{dom}(h)$. 
    \item For every $o \in \text{dom}(h)$ we have that if $h(o) = (C\<t\>[\W], \envf)$ then $C\<t\>$ is declared in $\vv{D}$, $\W$ is a reachable usage, and $h(o).\sf{fields} = C\<t\>.\sf{fields}$
\end{enumerate}
\end{definition}

\begin{proposition} If $c$ is a reachable configuration, then $c$ is well-formed.
\label{prop:reachable_well_formed}
\end{proposition}


\section{Type System}
\label{sec:typesystem}
In this section we describe the type system for \mungo. The type system described check that the order of method calls imposed by a class usage is correct. This is achieved by statically following usage transitions and ensuring the method body associated with the transitions do not result in protocol errors. We proceed by defining types, linearity in the typesystem, and type environments before explaining the type rules in detail. 

\subsection{Types in the type system}
In the type system values have either a base type $b$, a typestate $C\<t\>[\U]$ or the type of \sf{null} objects, $\bot$. 

The notion of linearity reappears in the type system; a type is linear if it is a typestate whose usage is not terminated.

\begin{definition}[Linear type] Let $t$ be a type. We write $\sf{lin}(t)$ if $t = C\<t'\>[\mathcal{W}]$ for some usage $\mathcal{W}$ where $\mathcal{W} \neq \sf{end}$\label{def:lin}.
\end{definition}

In our treatment of generics we need a special type in order to describe instantiations.  We define a special usage $\top_{\U}$ that is different from $\sf{end}$, but has no transitions. We define a special class $\sf{class}\ \topclass\  \{\topusage; \emptyset;\emptyset\}$ We see directly from the definition that $\toptype$ is linear.

\subsection{The binding model of the type system}

The intention is that the type system will give us a sound overapproximation of the operational semantics. The bindings in the type system should therefore be counterparts of the associated notions in the semantics.

A \textit{field type environment} $envT_F$ is the approximation of field environments from the operational semantics. A field type environment maps field name to types, that is, typestates, base types or the bottom type $\bot$.

\textit{Object field environments}, denoted by $\Lambda$, is the approximation of the heap, and maps objects to their typestate along with a field typing environment, describing the fields of the object. 

We add an additional environment to capture the typestate of the objects mentioned in the expression that is checked (free object names). We call this environment an \textit{object type environment} $envT_O$. It simply maps objects to typestates.

To model the parameter stack of a configuration we introduce a \textit{parameter type stack environment} $envT_S$ and define it as a sequence of pairs $(o, S)$ where $o$ is an objects and $S$ is a parameter binding $[x \mapsto t]$.

For readability in the typing rules, we collapse $envT_O$ and $envT_S$ in a combined environment $\Delta$. $\Delta$ is a sequence of object type and parameter type environments defined as $\Delta = envT_O \cdot envT_S$.

Finally, since usages can be recursive, we need to keep track of usages associated with recursion variables. We do this by introducing a \textit{usage recursion environment} $\Theta$ which is a partial function from usage variables to field typing environments. They are used to ensure consistency between the field of an object before and after a labelled expression.

\subsection{The type rules}

We now present the central type rules. The description of the type rules are divided into three main parts: expressions, declarations, and usages.

\subsubsection{Expressions}

The main intention behind our type system is to provide a sound
overapproximation of terminating transition sequences in the semantics
of expressions. Typing judgements for expressions are therefore of the
form
\[ \Lambda;\Delta \vddash^\Omega e : t \triangleright \Lambda';\Delta' \]
These are to be read as follows: Evaluating $e$ in the initial environments $\Lambda$, $\Delta$ and $\Omega$ will result in final typing environments $\Lambda'$ and $\Delta'$. Here $\Omega$ is a label environment that is used to map a label $k$ to a pair of environments ($\Lambda,\Delta$), $\Omega$ is only used in continue expressions and is therefore omitted in most rules.

\paragraph{Type rules for fields and methods} Table
\ref{tab:typingexpression} contains the rules of the
type system that describe the protocols followed by fields and
objects. The rule \rt{TFld} is concerned with field assignment, it tells us that a field assignment is well-typed if the field is unrestricted and the type of expression $e$ and field $f$ agree. Here we must to ensure that the type of the right-hand side is compatible with the field. In order to capture this we define the \sf{agree} predicate. It denotes that types agree with a declaration of that specific type. For objects, the usage is ignored, such that a field declared for type $C\<t\>$ can contain objects of type $C\<t\>[\U]$ for any usage $\U$.

\begin{definition}[Agree predicate]
\[
\sf{agree}(b,b)\qquad\sf{agree}(C\<t\>, C\<t\>[\mathcal{W}]) \qquad \sf{agree}(C, C\<\bot\>[\W]) \qquad \sf{agree}(C\<t\>, \bot)
\]
\end{definition}

The type rule \rt{TNew} shows how objects without a generic class
parameter are to be typed; the class parameter is instantiated with
$\bot$ and \rt{TNewGen} tell us how objects with a generic class parameter are typed.

For method calls we must check if the usage of the object that owns
the method will permit the call. The rule \rt{TCallF} handles method
calls on a field; for a method call $o.f$ to be well-typed, the usage
$\mathcal{U}$ of the field $f$ in $o$ must allow an $m$ transition and
the final type environment contains an updated typestate for $f$. The rule \rt{TCallP} handles method calls on a parameter in a similar way to \rt{TCallF}, but the updated typestate associated with parameter $x$ is contained in the parameter stack type environment $envT_S$ instead of the object field environment $\Lambda$. Note, that the body of the method is not type checked in this rule; this will be handled at the class declarations and is explained later.

\begin{example}
  We again return to the \textsf{File} class from Listing
  \ref{listing:file1}. In the rule \rt{TCallF} we see that in order
  for a call of the \textsf{close} method to be well-typed, there must
  be a transition labelled with \textsf{close}. Therefore, this method
  call would fail to be well-typed under the initial usage $\U$ from
  Example \ref{ex:fileusage}. \end{example}

The rule \rt{TRet} tell us that a \textsf{return}$\{ e \}$ expression
is well-typed if the expression body $e$ is well-typed and the method
parameter used in the expression body is terminated after the
execution. The final type environment is one, in which the parameter
stack environment does not mention the called object and its
associated parameter. The intention is that this mirrors the
modification of the stack environment in the semantics as expressed in
the reduction rule \rt{Ret}.

\vspace{-1cm}

\rulestable[ll]{Type rules for fields and methods}{tab:typingexpression}{4}{
    \vspace{-0.8cm}
    \rt{TNew} & ${\Lambda; \Delta \vdash_{\vv{D}} \sf{new} \ C : C\<\bot\>[C\<\bot\>.\sf{usage}] \triangleright \Lambda; \Delta}$ \\
    \vspace{-.2cm}
    \rt{TNewGen} & ${\Lambda; \Delta \vdash_{\vv{D}} \sf{new} \ C\<g\> : C\<g\>[C\<g\>.\sf{usage}] \triangleright \Lambda; \Delta}$ \\
    \rt{TFld} & \inference{C\<t''\>=\Lambda(o).\sf{class} & \sf{agree}(C\<t''\>.\sf{fields}_{\vv{D}}(f), t') \\\Lambda; \Delta \cdot (o,S) \vdash e:t' \triangleright \Lambda', o.f \mapsto t; \Delta' \cdot (o,S')& \neg \lin (t) }{\Lambda; \Delta \cdot (o,S) \vdash_{\vv{D}} f = e : \sf{void} \triangleright \Lambda' \{o.f \mapsto t'\}; \Delta' \cdot (o, S') } \\
    \rt{TCallF} & \inference{\Lambda;\Delta \cdot (o,S) \vdash e:t \triangleright \Lambda'\{o.f \mapsto C\<t''\>[\mathcal{U}]\};\Delta' \cdot (o,S')  & \\ t' \ m(t\ x)\{e'\} \in C\<t''\>.\sf{methods}_{\vv{D}} & \mathcal{U} \trans[m] \mathcal{W}}{\Lambda; \Delta \cdot (o, S) \vdash f.m(e) : t' \triangleright \Lambda'\{o.f \mapsto C\<t''\>[\mathcal{W}]\}; \Delta' \cdot (o,S')}  \\
    \rt{TCallP} & \inference{\Lambda;\Delta \cdot (o, S)
    \vdash_{\vv{D}} e:t\triangleright \Lambda';\Delta' \cdot (o,
    [x\mapsto C\<t''\>[\mathcal{U}]])  & \\ t'\ m(t\ x)\{ e' \}\in
    C\<t''\>.\sf{methods}_{\vv{D}}& \mathcal{U}\trans[m]
    \mathcal{W}}{\Lambda;\Delta\cdot(o,S) \vdash_{\vv{D}} x.m(e):t'
    \triangleright \Lambda';\Delta'\cdot(o,[x\mapsto
    C\<t''\>[\mathcal{W}]])}\\
    \rt{TRet} & \inference{\Lambda; \Delta \vddash e:t \triangleright \Lambda';\Delta' & \Delta' = \Delta'' \cdot (o', [x \mapsto t']) & \sf{terminated(t')}}{\Lambda; \Delta \cdot (o,S) \vddash \sf{return}\{ e \} : t \triangleright \Lambda';\Delta'' \cdot (o,S)}
}
\vspace{-1cm}

\paragraph{Type rules for control structures} Table \ref{tab:structualexpressions} contains type rules for
control structures in \mungo. The importance of terminal and final environments becomes apparent in
the type rule \rt{TSeq} for sequential composition. This type rule
tell us that $e_1;e_2$ is well-typed if the type of $e_1$ is
terminated, and that $e_2$ can be evaluated using the resulting
environments from the first expression. The rule \rt{TIf} handles if-expressions here it is important to note that the type and the final environments of both branches must be the same, otherwise typing later expressions would be impossible.

The rule \rt{TSwF} is a rule for typing
a $\textsf{switch}_{f.m}$ expression that involves a field. An expression of
this form is well-typed if the type of the condition $e$ is a set of
labels $L$ and the typestate for the field $f$ allows for branching on the
labels in $L$. The rule \rt{TSwP} handle $\textsf{switch}_{x.m}$ expressions that involve a reference. It is similar to \rt{TSwF}, except the updated typestate is contained in parameter stack type environment $envT_S$ instead of the field type environment $\Lambda$.


The rules \textsc{TLab} and \textsc{TCon} for labelled expressions allows environments to change during the evaluation of continue-style loops. However, if a \textsf{continue} expression is encountered, the environments must match the original environments, in order to allow an arbitrary number of iterations. 

\vspace{-1cm}

\rulestable[ll]{Typing rules for control structures}{tab:structualexpressions}{4}{
     \rt{TSeq} & \inference{\Lambda;\Delta \vddash e : t
       \triangleright \Lambda''; \Delta'' & \neg \lin(t) &
       \Lambda'';\Delta'' \vddash e':t' \triangleright
       \Lambda';\Delta'}{\Lambda;\Delta \vddash e;e' : t'
       \triangleright \Lambda';\Delta'} \\
    \rt{TIf} & \inference{
         \Lambda; \Delta \vddash e: \sf{Bool} \triangleright \Lambda''; \Delta'' \\ \Lambda''; \Delta'' \vddash e' : t \triangleright \Lambda'; \Delta'  & \Lambda''; \Delta'' \vddash e'':t \triangleright \Lambda'; \Delta'
    }{
         \Lambda; \Delta \vddash \sf{if}\ (e)\ \{e'\} \ \sf{else}\ \{e''\}:t \triangleright \Lambda'; \Delta'
    } \\
    \rt{TSwP} & \inference{
         \Lambda;\Delta\cdot (o, S) \vddash e : L \triangleright \Lambda'';\Delta''\cdot(o, [x \mapsto C\<t'\>[(\<l_i : u_i\>_{l_i \in L})^{\vv{E}}]]) \\
         \forall l_i \in L . \ \Lambda'';\Delta''\cdot(o, [x\mapsto C\<t'\>[u_i^{\vv{E}}]]) \vddash u_i : t \triangleright \Lambda';\Delta' \cdot (o, S')
    }{\Lambda\Delta \cdot (o, S) \vddash \sf{switch}_{x.m}\ (e)\{l_i : u_i\}_{l_i \in L} : t \triangleright \Lambda';\Delta' \cdot (o, S')} \\
    \rt{TSwF} & \inference{
        \Lambda;\Delta \cdot (o, S) \vddash e : L \triangleright \Lambda'',o.f \mapsto C\<t'\>[(\<l_i : u_i\>_{l_i \in L})^{\vv{E}}]];\Delta'' \\
        \forall l_i \in L . \ \Lambda'', o.f \mapsto C\<t'\>[u_i^{\vv{E}}];\Delta'' \vddash e_i : t \triangleright \Lambda';\Delta' \cdot (o, S')
    }{
        \Lambda;\Delta \cdot (o, S) \vddash \sf{switch}_{f.m}\ (e)\{l_i : e_i\}_{l_i \in L} : t \triangleright \Lambda';\Delta'\cdot (o, S')
    }\\
    \rt{TLab} & \inference{
    \Omega' = \Omega, k: (\Lambda, \Delta) & \Lambda; \Delta \vddash^{\Omega'} e : \sf{void} \triangleright \Lambda'; \Delta'
    }{
    \Lambda; \Delta \vddash^{\Omega} k : e : \sf{void} \triangleright \Lambda'; \Delta'
    } \\
    \rt{TCon} & \inference{
    \Omega' = \Omega, k:(\Lambda, \Delta)
    }{
    \Lambda; \Delta \vddash^{\Omega'} \sf{continue }k : \sf{void} \triangleright \Lambda'; \Delta'
    }
}

In Section \ref{sec:obs} a class \textsf{FileReader} was introduced with a loop repeated in Listing \ref{lst:loop}. Even though calling the \textsf{close} method leaves the field in another state than calling \textsf{read}, the code is well typed. The reason is that after calling \textsf{read}, the field is left in the initial state when entering the loop, and another iteration occurs. When calling \textsf{close} the loop is ended. Hence the only resulting state for the field after the loop, is $\textsf{File}[\textsf{end}]$.

 \begin{lstlisting}[firstnumber=26,caption={Loop from class \textsf{FileReader}},label={lst:loop}]
[...]
file.open(unit)
loop:
    switch(file.isEOF()) {
      EOF: file.close()
      NOTEOF: file.read();
              continue loop
    }
[...]
\end{lstlisting}


\paragraph{Typing rules for values} Table \ref{tab:typingvalues} define rules for typing values in \mungo. The rules are straightforward since typing any of these values do not change any environments. The rules cover four types from the core \mungo language: booleans, litterals, void, and $\bot$.  

\rulestable[llll]{Typing rules for values}{tab:typingvalues}{2.5}{
     \rt{TLit} & \inference{l\in L}{\Lambda; \Delta \vddash l:L \triangleright \Lambda; \Delta} &
     \rt{TVoid} & ${\Lambda; \Delta \vddash \sf{unit}:\sf{void} \triangleright \Lambda; \Delta}$ \\
     \rt{TBool} & \inference{v \in \{\sf{true}, \sf{false}\}}{\Lambda; \Delta \vddash v: \sf{Bool} \triangleright \Lambda; \Delta} &
     \rt{TBot} & ${\Lambda; \Delta \vddash \sf{null} : \bot \triangleright \Lambda; \Delta}$ \\
 }

\paragraph{Type rules for objects and references} Table \ref{tab:typingreferences} contains type rules for typing references and object values. The rule \rt{TObj} handles object typing and it describes that once we type an object, corresponding to reading it, we remove it from the object type environment. 
%
%
%
\rt{TNoLpar} and \rt{TNoLFld} handle typing non-linear parameters and fields, where no updates should happen to the environments. The rules \rt{TLinPar} and \rt{TLinFld} handle typing linear parameters and fields such that after typing the value, the linear parameter or field is updated to the type $\bot$ in either the parameter stack environment or field type environment. 

\rulestable[ll]{Typing rules for
  references}{tab:typingreferences}{2.5}{
 \rt{TObj}  & \inference{envT_O = envT_O', o
      \mapsto t}{\Lambda; envT_O \cdot envT_S \vddash o:t
      \triangleright \Lambda; envT_O' \cdot envT_S} \\
    \rt{TLinPar} & \inference{\lin(t)}{\Lambda;\Delta \cdot (o, [x \mapsto t]) \vddash x : t \triangleright \Lambda; \Delta \cdot (o, [x \mapsto \bot])} \\
    \rt{TNoLPar} & \inference{\neg \lin(t)}{\Lambda;\Delta \cdot (o, [x \mapsto t]) \vddash  x : t \triangleright \Lambda; \Delta \cdot (o, [x \mapsto t])} \\
    \rt{TLinFld} & \inference{t = \Lambda(o).f & \lin(t)}{\Lambda;\Delta \cdot (o, S) \vddash f : t \triangleright \Lambda\{o.f \mapsto \bot\}; \Delta \cdot (o,S)} \\
    \rt{TNoLFld} & \inference{\neg \lin(t)}{\Lambda\{o.f \mapsto t\};\Delta \cdot (o,S) \vddash f : t \triangleright \Lambda\{o.f \mapsto t\};\Delta \cdot (o,S)}
}

\subsubsection{Declarations}

The type rules for program and class declarations are seen in Table
\ref{tab:typeingprogram}. A program declaration $\vv{D}$ is well-typed
if each of its class declarations is well-typed; this is rule
\rt{TProg}. For a class to be well-typed, it usage must be well-typed
and the result of following the class usage $\U$ must result in a
field typing environment $envT_F$ that is terminated i.e. all fields
are unrestricted. This is captured by the rules \rt{tClass} and
\rt{TClassGen}. The latter rule captures our treatment of generic
classes. A generic class declaration is well-typed if we can
substitute the type variable and usage parameter with the top type $\toptype$. This substitution ensures that
the class is well-typed for all types, by checking that the class does
not call any of the methods of the type variable.

\rulestable[ll]{Typing program and class definitions}{tab:typeingprogram}{3.5}{
    \rt{TProg} & \inference{\forall D \in \vv{D}\ . \vdash_{\vv{D}} D}{\vdash D} \\
    \rt{TClass} & \inference{\emptyset;\vv{F}.\sf{inittypes}_{\vv{D}} \vdash_{\vv{D}} C\<\bot\>[\mathcal{U}]\triangleright envT_F\\ \sf{terminated}(envT_F)
    }{ \vdash_{\vv{D}} \sf{class}\ C \{\mathcal{U}, \vv{F}, \vv{M}\}}
    \\
    \rt{TClassGen} & \inference{\emptyset;\vv{F}.\sf{inittypes}_{\vv{D}} \vdash_{\vv{D}} C\<\toptype\>[\mathcal{U}]\triangleright envT_F\\ \sf{terminated}(envT_F)
    }{ \vdash_{\vv{D}} \sf{class}\<\alpha[\beta]\>\ C \{\mathcal{U}, \vv{F}, \vv{M}\}}
}

\subsubsection{Class usages}

Table \ref{tab:classusages} contains the rules for typing class
usages; it is here that method bodies are examined. The rule \rt{TCBr}
is of particular importance, since it tells us that method bodies are
type-checked according to the class usage and that only reachable
methods are examined. For each method $m_i$ mentioned in the class
usage we check that the body $e_i$ is well-typed.

\rt{TCCh} handles the case where a usage is a choice usage. A usage of
this form is well-typed if every branch of the usage is well-typed.

The rules \rt{TCVar} and \rt{TCRec} handle recursive usages.
\rt{TCVar} tells us that a recusive usage variable is well-typed if
the variable $X$ is in the usage environment $\Theta$ and mapped to
the initial field type environment and \rt{TCRec} tells us that a
recursive usage is well-typed if the recursive usage variable is in
the usage environment and the class where we use the usage associated
with the recursive usage variable is well-typed.

\rulestable[ll]{Typing class usage definitions}{tab:classusages}{2.5}{
    \rt{TCBr} & 
        \inference{
            I \neq \emptyset & \forall i \in I \ .\ \exists envT_{F}''\ . \\
            \{\sf{this} \mapsto envT_F\}; \emptyset \cdot (\sf{this}, [x_i \mapsto t_i']) \vddash e_i : t_i \triangleright \{ \sf{this} \mapsto envT_F'' \};\emptyset\cdot(\sf{this}, [x_i \mapsto t_i'']) \\
            \sf{terminated}(t_i'') & t_i\ m_i(t_i'\ x_i)\{ e_i \}\in C\<t\>.\sf{methods}_{\vv{D}} & \Theta;envT_F''\vddash C\<t\>[u_i^{\vv{E}}]\triangleright envT_F'
        }{
            \Theta;envT_F \vddash C\<t\>[\{m_i;u_i\}^{\vv{E}}_{i\in I}] \triangleright envT_F'
        } \\
    \rt {TCCh} &
        \inference{
            \forall l_i \in L \ . \ \Theta;envT_F \vddash C\<t\>[u_i^{\vv{E}}] \triangleright envT_F'
        }{
            \Theta;envT_F \vddash C\<t\>[\<l_i : u_i\>_{l_i \in L}^{\vv{E}}]\triangleright envT_F'
        }\\
    \rt{TCEn} &
        ${
            \Theta;envT_F \vddash C\<t\>[\sf{end}^{\vv{E}}] \triangleright envT_F
        }$ \\
    \rt{TCVar} &
        ${
            (\Theta, [X \mapsto envT_F]);envT_F \vddash C\<t\>[X^{\vv{E}}] \triangleright envT_F'
        }$ \\
    \rt{TCRec} &
        \inference{
            (\Theta, [X \mapsto envT_F]);envT_F \vddash C\<t\>[u^{\vv{E}}]_{\vv{D}} \triangleright envT_F'
        }{
            \Theta; envT_F \vddash C\<t\>[X^{\vv{E} \uplus \{X = u\}}] \triangleright envT_F'
        }
}

\begin{example}
  Let us return to the class \textsf{File} in the code example of Listing
  \ref{listing:file1}. Recall the usage for the \textsf{File} class is 
  \[ \U = \{\textsf{open}; \textsf{X}\}^{\vv{E}}\]
  \[ \vv{E} = \{\textsf{X} = \{\textsf{isEOF}; \left\langle\begin{array}{l} \textsf{EOF} : \{\textsf{close}; \textbf{end}\}  \\\textsf{NOTEOF} : \{\textsf{read}; \textsf{X}\}\end{array} \right\rangle \}\}\]
  To type this class, we inspect the \textsf{File} usage. The usage $\U$ starts as a branch usage where only method \textsf{open} is available. The rule \rt{TCBr} is then applied which checks that the method body of \textsf{open} is well-typed and move on to check the remaining usage $\U'$. $\U'$ is $\textsf{X}^{\vv{E}}$ which is a recursive variable, hence \rt{TCRec} is used to unfold it. Unfolding \textsf{X} results in a branch usage where the body of method \textsf{isEOF} is checked using \rt{TCBr}. The resulting usage $\U''$ is a choice usage so rule \rt{TCCh} is used to check the usages associated with each label. The usage associated with label \textsf{EOF} is $\{\textsf{close};\textbf{end}\}$ a branch usage, hence \rt{TCBr} is used. The remaining usage is $\textbf{end}$ and rule \rt{TCEn} is used to terminate the deviation. 
  The usage associated with label \textsf{NOTEOF} is also a branch usage, hence \rt{TCBr} is used again and its checks that the body of method \textsf{read} is well-typed. The resulting usage is the recursive variable \textsf{X} that has already been checked, so \rt{TCVar} is applied to terminate the derivation. This concludes the type check of class $\sf{File}[\mathcal{U}]$.
\end{example}



\section{The type system is sound}
\label{sec:results}
In this section we state and prove the soundness of our type system. We present a safety theorem, which guarantees that a well-typed program does not attempt \textit{null-derefencing}, and that programs follow the specified usages. 

\subsection{Relating program states and types}

Since the type system must be an overapproximation of the semantics of \mungo, we must relate the notions of program states and type bindings. Central to this is that the type information at run-time (as expressed in the heap) can be retrieved and compared to the type information in a type environment. In order to do this, we define a partial function \sf{getType} that returns the type of a value $v$ in a given heap $h$.  \begin{definition}[Heap value types] \[
    \sf{getType}(v,h) = \begin{cases}
    \sf{void} & \text{if } v= \sf{unit} \\
    \sf{bool} & \text{if } v = \sf{true} \text{ or } v = \sf{false} \\
    \bot & \text{if } v = \sf{null} \\
    L & \text{if } v = l \text{ and } l \in L \\
    C\<t\>[\W] & \text{if } h(v) = \< C\<t\>[\W], env_F \> \\
    \end{cases}
\]
\end{definition}

We now use this function to define what it means to be a well-typed configuration wrt. type environments $\Lambda$ and $\Delta$, written $\Lambda; \Delta \vddash \<h, env_S, e\> : t \triangleright \Lambda'; \Delta'$. A configuration is well-typed if its bindings match the type information given: The heap matches the field typing environment $\Lambda$, the stack $\Delta$ in the type system matches the stack from the semantics, the objects mentioned in the type system match those of $e$, the expression $e$ itself is well typed and the field type environment $\Lambda$ is compatible with the program $\vv{D}$.



\paragraph{Well typed heaps} A heap $h$ is well typed  in a field typing environment $\Lambda$ ($\Lambda \vddash h$) if the types of all objects in the heap matches those of the field typing environment. Furthermore all objects mentioned in the heap must also be mentioned in the field typing environment and vice versa.

\[\rt{WTH}\, \inference{\forall f \in h(o).\sf{fields}\ .\ \Lambda(o).envT_F(f) = \sf{getType}(h(o).f, h) \\ \text{dom}(\Lambda)=\text{dom}(h) & \forall o \in \text{dom}(h)\ .\ h(o).\sf{fields} = \Lambda(o).\sf{fields} = h(o).\sf{class}.\sf{fields}_{\vv{D}}}{ \Lambda \vddash h}\]

\paragraph{Well typed parameter stacks} A parameter stack is well typed in a parameter typing sequence in the context of expression $e$ ($envT_S \vdhash_{e} env_S$), if all levels of the parameter stack match with what is known in the parameter typing sequence. The parameter typing sequence can contain more information than what is in the parameter stack. If the expression does not match any of the rules, then \rt{WTP-Base} is used.

\begin{tabular}{llll}
    \rt{WTP-Base} & \multicolumn{3}{l}{\inference{t_\bot = \sf{getType}(v, h) }{ envT_S \cdot (o, [x\mapsto t_\bot]) \vdhash_{\_} (o, [x \mapsto v])}} \\
    \rt{WTP-Ret} & \multicolumn{3}{l}{\inference{envT_S \vdhash_{e} env_S & \sf{getType}(v, h) = t_{\bot}}{ envT_S \cdot (o, [x \mapsto t_{\bot}]) \vdhash_{\return{e}}  env_S \cdot (o, [x \mapsto v])}} \\
    \rt{WTP-Sw} & \multicolumn{3}{l}{\inference{envT_S \vdhash_{e} env_S}{envT_S \vdhash_{\sf{switch}_{r.m}\ (e) \{l_i: e_i\}_{l_i \in L}} env_S}} \\
    \rt{WTP-Mthd} & \inference{envT_S \vdhash_{e} env_S}{envT_S \vdhash_{r.m(e)} env_S} &
    \rt {WTP-Fld} & \inference{envT_S \vdhash_{e} env_S}{envT_S \vdhash_{f = e} env_S} \\
    \rt{WTP-Seq} & \inference{envT_S \vdhash_{e} env_S}{envT_S \vdhash_{e;e'} env_S} &
    \rt{WTP-If} & \inference{envT_S \vdhash_{e} env_S}{envT_S \vdhash_{\sif{e}{e_1}{e_2}} env_S} 
\end{tabular}

\paragraph{Well typed expressions} An expression is well typed in a object typing environment ($envT_O \vdhash e $) if the objects mentioned in $e$ are precisely those bound in the $envT_O$.

\[
\rt{WTE}\, \inference{\forall o\in \text{dom}(envT_O)\ .\ envT_O(o)=\sf{getType}(o,h) & \text{dom}(envT_O)=\sf{objects}(e)}{envT_O \vdhash e}
\]

\paragraph{Well typed declarations} A program is well typed in a field typing environment ($\Lambda \vdhash \vv{D}$) if the current objects mentioned in the field typing environment or the heap, can all continue execution and terminate -- no objects are stuck.

\[\rt{WTD}\, \inference{\exists \Gamma' & \emptyset;\Lambda(o).envT_F \vddash \sf{getType}(o,h) \triangleright \Gamma' & \sf{terminated}(\Gamma') \\ \forall o \in \text{dom}(h) & \Lambda(o).\sf{class} = h(o).\sf{class} & \text{dom}(h) = \text{dom}(\Lambda) }{\Lambda \vdhash \vv{D}}\]

\paragraph{Well typed configurations} We now combine the notion of well typed heaps, stacks, expressions and declarations into \textit{well typed configurations}. A well typed configuration describes that the current state of the evaluating program is in accordance with the type system, and is defined below.

\[\rt{WTC}\, \inference{\Lambda \vddash h & envT_S \vdhash_e env_S & envT_O \vdhash e \\
& \Lambda; envT_O \cdot envT_S \vddash e : t_\bot \triangleright \Lambda';\Delta & \Lambda \vdhash \vv{D}}{\Lambda;envT_O \cdot envT_S \vddash \< h, env_S, e \> : t_\bot \triangleright \Lambda';\Delta}\]

We now show that the initial configuration presented in Definition \ref{def:initial_configuration} is in fact a well-typed configuration and thereby relating initial configurations to well typed programs.

\begin{lemma}[Well typed initial configuration]
Let $\vv{D}$ be a program such that $\vdash \vv{D}$. The initial configuration of $\vv{D}$ is well typed.
\end{lemma}

\begin{proof}
Recall the initial configuration is defined as the configuration right after $\sf{main}$ was invoked in an object of class $\sf{Main}$: 
$$ ic = \<\{o_{\sf{main}} \mapsto \<\sf{Main}\<\bot\>[\sf{end}^{\emptyset}], \sf{Main}\<\bot\>.\sf{fields}_{\vv{D}}.\sf{initvals}_{\vv{D}}\>\}, (o_{\sf{main}}, s_{\sf{main}}), e\> $$

Assume that $\vdash \vv{D}$ then from \rt{TProg} we also have the class $\sf{Main}$ is well typed.

From \rt{TClass} we know:
\begin{itemize}
    \item $\emptyset; \vv{F}.\sf{inittypes} \vddash \sf{Main}\<\bot\>[\{\sf{main};\sf{end}\}^{\emptyset}] \triangleright envT_F$
    \item $\sf{terminated}(envT_F)$
\end{itemize}

From \rt{TCBr} we know:
\begin{itemize}
    \item $\{\sf{this} \mapsto \vv{F}.\sf{inittypes}\}; \emptyset \cdot (\sf{this}, [x \mapsto \sf{void}]) \vddash e : \sf{void} \triangleright \{\sf{this} \mapsto envT_F\}; \emptyset \cdot (\sf{this}, [x \mapsto t])$
    \item $\emptyset; envT_F \vddash \sf{Main}[\sf{end}]^{\emptyset} \triangleright envT_F'$
    \item $envT_F = envT_F'$ follows from \rt{TCen}, since an $\sf{end}$ usage does not change any fields.
\end{itemize}

We now show that
\[\Lambda;\Delta \vddash \<\{o_{\sf{main}} \mapsto \<\sf{Main}\<\bot\>[\sf{end}^{\emptyset}], \sf{Main}\<\bot\>.\sf{fields}_{\vv{D}}.\sf{initvals}_{\vv{D}}\>\}, (o_{\sf{main}}, [x \mapsto \sf{unit}]), e\> \triangleright \Lambda'; \Delta'\]

Where $\Lambda = \{\sf{this} \mapsto \vv{F}.\sf{inittypes}\}$ and $\Delta = \emptyset \cdot (\sf{this}, [x \mapsto \sf{void}])$.

By showing the premises of \rt{WTC} are satisfied.
\begin{itemize}
    \item $\{\sf{this} \mapsto \vv{F}.\sf{inittypes}\} \vddash \{o_{\sf{main}} \mapsto \<\sf{Main}\<\bot\>[\sf{end}^{\emptyset}], \sf{Main}\<\bot\>.\sf{fields}_{\vv{D}}.\sf{initvals}_{\vv{D}}\>\}$ is satisfied since the field types in $\Lambda$ clearly correspond with the heap and their domains are the same. 
    \item $(\sf{this}, [x \mapsto \sf{void}]) \vdhash_e (o_{\sf{main}}, [x \mapsto \sf{unit}])$ is concluded with \rt{WTP-Base} (potentially after applying other \rt{WTP} rules except \rt{WTP-Ret}).
    \item $\emptyset \vdhash e$ is satisfied since $e$ is not a run-time expression hence $\sf{objects}(e) = \emptyset$.
    \item $\Lambda; \Delta \vddash e : t \triangleright \Lambda'; \Delta'$ from \rt{TCBr}.
    \item $\{\sf{this} \mapsto \vv{F}.\sf{inittypes}\} \vdhash \vv{D}$ satisfied from \rt{TCBr} and \rt{TClass}. Namely, $\emptyset; envT_F \vddash \sf{Main}[\sf{end}]^{\emptyset} \triangleright envT_F$ and $\sf{terminated}(envT_F)$.
\end{itemize}

\end{proof}

\subsection{The subject reduction theorem}

We can finally present the subject reduction theorem for our system. It states that if a well-typed configuration can perform a step, then the resulting configuration will also be well-typed. 

\begin{theorem}[Subject reduction]
\label{lemma:subject_reduction}
Let $\vv{D}$ be such that $\vdash \vv{D}$ and let $\<h, env_S, e\>$ be a configuration. If $\vv{D} \vddash \<h, env_S, e\> \trans \<h', env_S', e'\>$ then:
$$
    \exists \Lambda, \Delta\ .\ \Lambda, \Delta \vddash \<h, env_S, e\> : t \triangleright \Lambda'; \Delta' \implies \exists \Lambda^N, \Delta^N\ .\ \Lambda^N, \Delta^N \vddash \<h', env_S', e' \> : t \triangleright \Lambda'';\Delta'$$

where $\Lambda'(o) = \Lambda''(o)$ and $o$ is the active object in the resulting configuration.
\end{theorem}

\begin{proof}
By induction in the structure of the reduction rules. 

\end{proof}

\subsection{Error freedom}

We can now formulate an important result guaranteed by the type system, that well-typed programs do not exhibit run-time errors.

Here we need to define the notion of run-time error by means of an error predicate for configurations. If a configuration $\<h, env_S, e\>$ has an error, we write $\<h, env_S, e\>\err$. 

Examples of the types of error we can catch using this notation includes the two important problems we set out to solve, namely \textit{null-dereferencing} and protocol errors. 

In Table \ref{tab:errorpred} we present the rules defining the error predicate. Each rule captures a particular run-time error, and each of corresponds to one of four kinds of errors for methods and fields.

Thus, the rules \rt{MthNotAv-1} and \rt{MthNotAv-2} together define the occurrences of the run-time error \emph{method not available}. The rules \rt{NullCall-1} and \rt{NullCall-2} describe two cases of null dereferencing that occur when the object whose method $m$ is to be called has been nullified. \rt{NullCall-1} is an instance of a \emph{field not available} error, whereas \rt{NullCall-2} is an instance of a \emph{parameter not available} error.


\rulestable[ll]{Error predicates}{tab:errorpred}{2.5}{
    \rt{NullCall-1} & \inference{h(o).f = \sf{null}}{\<h, (o, s) \cdot env_S, f.m(v)\> \err} \\
    \rt{NullCall-2} & $\<h, (o, [x \mapsto \sf{null}]) \cdot env_S, x.m(v)\> \err$ \\
    \rt{MthdNotAv-1} & \inference{h(o).f = o' & h(o').usage \nottrans[m]}{\<h, (o, s) \cdot env_S, f.m(v)\> \err} \\
    \rt{MthdNotAv-2} & \inference{h(o').usage \nottrans[m]}{\<h, (o, [x \mapsto o']) \cdot env_S, x.m(v)\> \err} \\
    \rt{FldErr} & \inference{h(o)=\<C\<t\>[\U], env_F\> & f \not\in \dom{env_F}}{\<h, (o, S) \cdot env_S, f\>\err} \\
    \rt{IfCErr} & \inference{\<h, env_S, e\>\err}{\<h, env_S, \sif{e}{e_1}{e_2}\>\err} \\
    \rt{FldCErr} & \inference{\<h, env_S, e\>\err}{\<h, env_S, f = e\>\err} \\
    \rt{CallCErr} & \inference{\<h, env_S, e\>\err}{\<h, env_S, r.m(e)\>\err} \\
    \rt{RetCErr} & \inference{\<h, env_S, e\>\err}{\<h, env_S \cdot (o, s), \return{e}\>\err} \\
    \rt{SeqCErr} & \inference{\<h, env_S, e\>\err}{\<h, env_S, e;e'\>\err} \\
    \rt{SwCErr} & \inference{\<h, env_S, e\>\err}{\<h, env_S, \sf{switch}_{r.m}\ (e) \{l_i : e_i\}_{l_i \in L}\>\err}
}

\begin{lemma}[Error freedom]
\label{thm:errorfreedom}
If $\exists \Lambda,\Delta.\ \Lambda;\Delta\vddash \<h, env_S, e\> : t \triangleright \Lambda';\Delta'$ then $\<h, env_S, e\>\noterr$
\end{lemma}
\begin{proof}
Induction in the structure of $\err$.

\end{proof}

\subsection{Type safety}

We can now state our main theorem: that a well-typed configuration will never experience the run-time errors captured by $\err$. This is a direct consequence of Theorem \ref{lemma:subject_reduction}
 and Lemma \ref{thm:errorfreedom}.

\begin{theorem}[Safety]
\label{thm:safety}
If $\Lambda;\Delta \vddash \<h, env_S, e\> : t \triangleright \Lambda';\Delta'$ and $\<h, env_S, e\>\rightarrow^{*}\<h', env_S', e'\>$ then $\<h', env_S', e'\>\noterr$
\end{theorem}

\section{Implementation}
\label{sec:implementation}

A prototype of the type system for the presented version of \mungo has been implemented in Haskell. The language supported by the implementation is the language presented in this report, extended with integers and boolean expressions. These extensions does not introduce any interesting challenges in the type system, and as such is not presented in this work. The implementation is available at \texttt{\url{https://github.com/MungoTypesystem/Mungo-Typechecker}}.

The implementation has three main components, the parser, the type system and an inference module for usages. The inference module can infer usages for classes based on the usages of its fields. A description of the method for usage inference is outside the scope of this report, and will be presented elsewhere. In this presentation we focus on the type system.

\subsection{A \mungo Program}
A \mungo program consists of a number of enum declarations, followed by a number of class declarations, as illustrated in Listing \ref{lst:mungoprogram}. A class declaration can be seen on lines 6-18, and consists of the class usage, followed by the fields and methods. As shown in line 8, the syntax for usages is slightly adapted for writability in the implementation and is on the form $u[X_i = u_i]_{i \in I}$ rather than the usual $u^{\vv{E}}$.

\begin{lstlisting}[language=mungo, label={lst:mungoprogram}, caption={Structure of a \mungo program in the implementation}]
enum State {
    l1 
    l2
}

class<A[b]> C {
    // Usage
    {m; <l1 : end l2 : {n; X}>}[X = {o; end   m; X}]

    // Fields
    bool f1
    OtherClass f2
    
    // Methods
    State m(void x) { ...; l1 }
    A[b] n(A[b] x) { ...; x }
    void o(void x) { ... }
}
\end{lstlisting} 

\subsection{Type Checking}
The implementation of the type checking of \mungo programs follows the structure of the type system presented in Section \ref{sec:typesystem}. Type checking of a class follows the defined usage and a field typing environment is updated in accordance with the bodies of the methods declared by the usage. It is ensured that following the usage of a class results in a terminated field typing environment, and if this is not the case for a class in the program, this is reported to the programmer. Further work the implementation includes translating real-life Java programs to \mungo, and verify their correctness and the absence of errors. 
\section{Conclusions and future work}
\label{sec:discussion}

In this paper we present a behavioural type system for the \mungo language that has also been implemented in the form of a type-checker.

In the system, each class is annotated with usages that describe the protocol to be followed by method calls to object instances of the class. Moreover, object fields that are references in the form of objects can only be used in a linear fashion. The type system extends that of \cite{KouzapasEtal}, and the soundness results that we have established provide a formal guarantee that a well-typed program will satisfy two important properties: That null dereferencing does not occur and that objects complete the protocol that their usage annotations require. Here, behavioural types are essential, as they allow the type of a field to evolve to $\bot$, the only type inhabited by the null value $\textsf{null}$. This is in contrast to most type systems for Java like languages that do not let types evolve during a computation and overload $\textsf{null}$ to have any type.


Our notion of generics is similar to that of universal types from the typed $\lambda$-calculus, and does not allow for bounds to be placed on the typing parameters. On one hand this creates a limitation for the use of such generics, since class parameters of type $\alpha[\beta]$ cannot be used for method calls, as there is no knowledge of the actual type when typechecking the class. As it stands now, the generics can be used to type collections such as the classes in \texttt{java.util.Collections}.

To be able to type a larger subset of Java, than what \mungo currently allows, further work also includes adding inheritance to the language in a type-safe manner. Inheritance is common in object oriented programming, and would allow Mungo to be used for a larger set of programs. This is particularly important, since classes in languages like Java always implicitly inherit from the class \texttt{Object}. 
However, Amin and Tate have shown that the type system of Java \cite{Amin2016}, which uses bounds on type parameters in definitions of generic classes, is unsound. Moreover, Grigore has shown that type checking in the presence of full subtyping is undecidable \cite{Grigore2017}. Therefore, in further work, we need to be extremely careful when introducing subtyping into our system. 

Our present type system requires a non-aliasing property of fields; it would be highly desirable to also deal with this limitation of the system. A possible approach would be to use ideas from the work on behavioural separation of Caires and Seco \cite{DBLP:conf/popl/CairesS13} and on SHAPES \cite{DBLP:journals/corr/abs-1901-08006} by Franco et al.


\bibliographystyle{plain}

\newpage


\appendix
\section{Appendix: Full proofs from \emph{Behavioural Types for Memory and Method Safety in Java}}
In this appendix we include the full proofs for the properties
described in the main paper. 

\subsection{Weakening proof}
\label{app:weakening}
\tagthm{32}

The weakening lemma tells us we can expand the context of a typing judgment for expressions without affecting its validity. The context we expand is the three environments; $\Lambda$, $envT_O$, and $envT_S$. Note that we require the active object of $envT_S$ to remain the active object, hence we cannot add elements to the bottom of this environment. We use weakening in the subject reduction proof specifically the cases \rt{CallF} and \rt{CallP} where we need to extend the environments $\Lambda$ and $envT_S$, used when the method body was typed in \rt{TCbr}, with the additional elements of the current context for the method body.  

\begin{lemma}{Weakening Lemma.}
\label{lemma:weakening}
Suppose $\Lambda;envT_O \cdot envT_S \vddash e : t' \triangleright \Lambda';envT_O' \cdot envT_S'$, $o \not\in \text{dom}(\Lambda)$ and $o' \not\in \text{dom}(envT_O)$, then $\Lambda, o \mapsto (C\<t\>, envT_F);(envT_O, o' \mapsto t'') \cdot (o'', S) \cdot envT_S \vddash e : t' \triangleright \Lambda', o \mapsto (C\<t\>, envT_F); (envT_O', o' \mapsto t'') \cdot (o'', S) \cdot envT_S'$.
\end{lemma}

\begin{proof}
Proof by induction in the structure of the derivation for expression type judgements. Most of the cases follow from applying the induction hypothesis to the premises of a rule, we therefore only include the cases that do not follow this approach.  

\begin{enumerate}[ncases]
    \item[\rt{TFld}] 
    \begin{subproof}
    
    \[\inference{C\<t''\>=\Lambda(o).\sf{class}\\\Lambda; \Delta \cdot (o,S) \vdash e:t' \triangleright \Lambda', o.f \mapsto t; \Delta' \cdot (o,S') & \neg \lin (t) & \sf{agree}(C\<t''\>.\sf{fields}_{\vv{D}}(f), t')}{\Lambda; \Delta \cdot (o,S) \vdash_{\vv{D}} f = e : \sf{void} \triangleright \Lambda' \{o.f \mapsto t'\}; \Delta' \cdot (o, S') } \]
    
    We assume that $\Lambda; \Delta \cdot (o,S) \vdash_{\vv{D}} f = e : \sf{void} \triangleright \Lambda' \{o.f \mapsto t'\}; \Delta' \cdot (o, S')$ is correct, $o'' \not\in \text{dom}(\Lambda)$ and $o^{(3)} \not\in \text{dom}(envT_O)$.
    
    We then show that 
    \begin{multline}
    \label{eq:tfldweak}
        \Lambda, o'' \mapsto (C'\<t^N\>, envT_F);  (envT_O, o^{(3)} \mapsto t^{(3)}) \cdot (o', S') \cdot envT_S \cdot (o,S) \vdash_{\vv{D}} \\ f = e : \sf{void} \triangleright \Lambda', o'' \mapsto (C'\<t^N\>, envT_F) \{o.f \mapsto t'\}; (envT_O', o^{(3)} \mapsto t^{(3)}) \cdot (o', S') \cdot envT_S' \cdot (o,S'')
    \end{multline}
    
    \begin{itemize}
        \item $C\<t''\>=\Lambda(o).\sf{class}$ from the definition of weakening we have that $o'' \not\in \text{dom}(\Lambda)$ hence $o \neq o''$ and
        \begin{equation} 
        \label{eq:tfldweak-2}
            (\Lambda, o'' \mapsto (C'\<t^N\>, envT_F))(o).\sf{class} = \Lambda(o).\sf{class} = C\<t''\>
        \end{equation}
        The premise is therefore satisfied from the assumption and (\ref{eq:tfldweak-2})
        \item $\lin(t)$ from the assumption.
        \item $\sf{agree}(C\< t'' \>.\sf{fields}_{\vv{D}}(f),t')$ from the assumption. 
        \item $ \Lambda, o'' \mapsto (C'\<t^N\>, envT_F);  (envT_O, o^{(3)} \mapsto t^{(3)}) \cdot (o', S') \cdot envT_S \cdot (o,S)  \vdash e:t' \triangleright \Lambda', o'' \mapsto (C'\<t^N\>, envT_F), o.f \mapsto t; (envT_O', o^{(3)} \mapsto t^{(3)}) \cdot  (o', S') \cdot envT_S' \cdot (o,S'')$ from the induction hypothesis.
    \end{itemize}
    
    We can now conclude (\ref{eq:tfldweak}).
    \end{subproof}
    
    \item[\rt{TIf}]
    \begin{subproof}
    \[\inference{
        \Lambda; \Delta \vddash e: \sf{Bool} \triangleright \Lambda''; \Delta'' & \Lambda''; \Delta'' \vddash e' : t \triangleright \Lambda'; \Delta' & \Lambda''; \Delta'' \vddash e'':t \triangleright \Lambda'; \Delta'
    }{
        \Lambda; \Delta \vddash \sf{if}\ (e)\ \{e'\} \ \sf{else}\ \{e''\}:t \triangleright \Lambda'; \Delta'
    }\]
    
    We assume that $\Lambda; \Delta \vddash \sf{if}\ (e)\ \{e'\} \ \sf{else}\ \{e''\}:t \triangleright \Lambda'; \Delta'$ is correct, $o \not\in \text{dom}(\Lambda)$ and $o' \not\in \text{dom}(envT_O)$.
    
    We then show that $\Lambda, o \mapsto (C\<t^N\>, envT_F); (envT_O, o' \mapsto t') \cdot (o'', S) \cdot envT_S \vddash \allowbreak \sf{if}\ (e)\ \{e'\} \ \allowbreak \sf{else}\ \{e''\}:t \triangleright \Lambda', o \mapsto (C\<t^N\>, envT_F); (envT_O', o' \mapsto t') \cdot (o'', S) \cdot envT_S')$
    
    \begin{itemize}
        \item $\Lambda, o \mapsto (C\<t^N\>, envT_F); (envT_O, o' \mapsto t') \cdot (o'', S) \cdot envT_S \vddash e: \sf{Bool} \triangleright \Lambda'', o \mapsto (C\<t^N\>, envT_F); (envT_O'', o' \mapsto t') \cdot (o'', S) \cdot envT_S''$ follows from the induction hypothesis.
        \item $\Lambda, o \mapsto (C\<t^N\>, envT_F); (envT_O'', o' \mapsto t') \cdot (o'', S) \cdot envT_S'' \vddash e': t \triangleright \Lambda', o \mapsto (C\<t^N\>, envT_F); (envT_O', o' \mapsto t') \cdot (o'', S) \cdot envT_S'$ follows from the induction hypothesis.
        \item $\Lambda, o \mapsto (C\<t^N\>, envT_F); (envT_O'', o' \mapsto t') \cdot (o'', S) \cdot envT_S'' \vddash e'': t \triangleright \Lambda', o \mapsto (C\<t^N\>, envT_F); (envT_O', o' \mapsto t') \cdot (o'', S) \cdot envT_S'$ follows from the induction hypothesis.
    \end{itemize}
    
    We can now conclude that $ \Lambda, o \mapsto (C\<t^N\>, envT_F); (envT_O, o' \mapsto t') \cdot (o'', S) \cdot envT_S \vddash \sf{if}\ (e)\ \{e'\} \ \sf{else}\ \{e''\}:t \triangleright  \Lambda', o \mapsto (C\<t^N\>, envT_F); (envT_O', o' \mapsto t') \cdot (o'', S) \cdot envT_S'$
    \end{subproof}
    \item[\rt{TSeq}]
    \begin{subproof}
    \[\inference{\Lambda;\Delta \vddash e : t \triangleright \Lambda''; \Delta'' & \neg \lin(t) & \Lambda'';\Delta'' \vddash e':t' \triangleright \Lambda';\Delta'}{\Lambda;\Delta \vddash e;e' : t' \triangleright \Lambda';\Delta'}\]
    
    We assume that $\Lambda;\Delta \vddash e;e' : t' \triangleright \Lambda';\Delta'$ is correct, $o'' \not\in \text{dom}(\Lambda)$ and $o^{(3)} \not\in \text{dom}(envT_O)$.
    
    We then show that
    \begin{multline}
    \label{eq:tseqweak}
    \Lambda, o'' \mapsto (C\<t^N\>, envT_F) ; (envT_O, o^{(3)} \mapsto t'') \cdot (o', S') \cdot envT_S \vddash \\ e;e' : t' \triangleright \Lambda', o'' \mapsto (C\<t^N\>, envT_F) ; (envT_O', o^{(3)} \mapsto t'') \cdot (o', S') \cdot envT_S'
    \end{multline}
    
    From our induction hypothesis we have:
    \begin{multline}
    \label{eq:tseqweak-2}
        \Lambda, o'' \mapsto (C\<t^N\>, envT_F) ; (envT_O, o^{(3)} \mapsto t'') \cdot (o', S') \cdot envT_S \vddash \\ e : t \triangleright \Lambda'', o'' \mapsto (C\<t^N\>, envT_F) ; (envT_O'', o^{(3)} \mapsto t'') \cdot (o', S') \cdot envT_S''
    \end{multline}
    
    From our assumption we have: $\neg \lin(t)$.
    
    From our induction hypothesis and (\ref{eq:tseqweak-2}) we have:
    \begin{multline}
    \label{eq:tseqweak-3}
    \Lambda'', o'' \mapsto (C\<t^N\>, envT_F) ; (envT_O'', o^{(3)} \mapsto t'') \cdot (o', S') \cdot envT_S'' \vddash \\ e':t' \triangleright \Lambda', o'' \mapsto (C\<t^N\>, envT_F) ; (envT_O', o^{(3)} \mapsto t'') \cdot (o', S') \cdot envT_S'
    \end{multline}
    
    From (\ref{eq:tseqweak-2}) and (\ref{eq:tseqweak-3}) we can conclude (\ref{eq:tseqweak}).
    
    \end{subproof}
    \item[\rt{TCallF}]
    \begin{subproof}
    \[ \inference{\Lambda;\Delta \cdot (o,S) \vdash e:t \triangleright \Lambda'\{o.f \mapsto C\<t''\>[\mathcal{U}]\};\Delta' \cdot (o,S') & \mathcal{U} \trans[m] \mathcal{W} & \\ t' \ m(t\ x)\{e'\} \in C\<t''\>.\sf{methods}_{\vv{D}}}{\Lambda; \Delta \cdot (o, S) \vddash f.m(e) : t' \triangleright \Lambda'\{o.f \mapsto C\<t''\>[\mathcal{W}]\}; \Delta' \cdot (o,S')} \]
    
    Assume that $\Lambda; \Delta \cdot (o, S) \vddash f.m(e) : t' \triangleright \Lambda'\{o.f \mapsto C\<t''\>[\mathcal{W}]\}; \Delta' \cdot (o,S')$ is correct, $o' \not\in \text{dom}(\Lambda)$ and $o'' \not\in \text{dom}(envT_O)$.
    
    We then show that
    \begin{multline}
    \label{eq:tcallfweak}
        \Lambda, o' \mapsto (C'\<t^N\>, envT_F); (envT_O, o'' \mapsto t^{(3)}) \cdot (o^{(3)}, S') \cdot envT_S \cdot (o, S) \vddash \\ f.m(e) : t' \triangleright \Lambda',o' \mapsto (C'\<t^N\>, envT_F)\{o.f \mapsto C\<t''\>[\mathcal{W}]\}; (envT_O', o'' \mapsto t^{(3)}) \cdot (o^{(3)}, S') \cdot envT_S' \cdot (o, S'')
    \end{multline}
    
    \begin{itemize}
        \item $\Lambda, o' \mapsto (C'\<t^N\>, envT_F); (envT_O, x' \mapsto t^{(3)}, (o', S') \cdot envT_S \cdot (o, S)) \vddash e:t \triangleright \Lambda',o' \mapsto (C'\<t^N\>, envT_F)\{o.f \mapsto C\<t''\>[\mathcal{W}]\}; (envT_O', x' \mapsto t^{(3)}) \cdot (o', S') \cdot envT_S' \cdot (o, S'')$ from the induction hypothesis.
        \item $U \trans[m] \W$ from the assumption
        \item $t' \ m(t\ x)\{e'\} \in C\<t''\>.\sf{methods}_{\vv{D}}$ from the assumption
    \end{itemize}
    
    We can now conclude (\ref{eq:tcallfweak}) is correct.
    \end{subproof}
    
    \item[\rt{TCallP}]
    \begin{subproof}
    \[\inference{\Lambda;\Delta \cdot (o, S) \vdash_{\vv{D}} e:t\triangleright \Lambda';\Delta' \cdot (o, [x\mapsto C\<t''\>[\mathcal{U}]]) & \mathcal{U}\trans[m] \mathcal{W} & \\ 
    t'\ m(t\ x')\{ e' \}\in C\<t''\>.\sf{methods}_{\vv{D}}}{\Lambda;\Delta\cdot(o,S) \vdash_{\vv{D}} x.m(e):t' \triangleright \Lambda';\Delta' \cdot (o,[x\mapsto C\<t''\>[\mathcal{W}]])}\]
    
    We assume that $\Lambda;\Delta\cdot(o,S) \vdash_{\vv{D}} x.m(e):t' \triangleright \Lambda';\Delta' \cdot (o,[x\mapsto C\<t''\>[\mathcal{W}]])$ is correct, $o'' \not\in \text{dom}(\Lambda)$ and $o^{(3)} \not\in \text{dom}(envT_O)$.
    
    We then show that 
    \begin{multline}
    \label{eq:tcallpweak}
        \Lambda, o'' \mapsto (C'\<t^N\>, envT_F); (envT_O, o^{(3)} \mapsto t^{(3)}) \cdot (o', S') \cdot envT_S \cdot(o,S) \vdash_{\vv{D}} x.m(e):t' \triangleright \\ \Lambda', o'' \mapsto (C'\<t^N\>, envT_F); (envT_O, o^{(3)} \mapsto t^{(3)}) \cdot (o', S') \cdot envT_S' \cdot(o,[x\mapsto C\<t''\>[\mathcal{W}]])
    \end{multline}
    
    \begin{itemize}
        \item $\Lambda, o'' \mapsto (C'\<t^N\>, envT_F); (envT_O, o^{(3)} \mapsto t^{(3)}) \cdot (o', S') \cdot envT_S \cdot(o,S) \vdash_{\vv{D}} e:t \triangleright \Lambda', o'' \mapsto (C'\<t^N\>, envT_F); (envT_O', o^{(3)} \mapsto t^{(3)}) \cdot (o', S') \cdot envT_S' \cdot (o,[x\mapsto C\<t''\>[\mathcal{U}]])$ from the induction hypothesis.
        \item $\mathcal{U}\trans[m] \mathcal{W}$ from the assumption.
        \item $t'\ m(t\ x')\{ e' \}\in C\<t''\>.\sf{methods}_{\vv{D}}$ from the assumption.
    \end{itemize}
    
    We can now conclude (\ref{eq:tcallpweak}).
    \end{subproof}
    \item[\rt{TSwF} \& \rt{TSwP}] 
    
    \begin{subproof}
    The proof for \rt{TSwF} and \rt{TSwP} are very similar. So below we only show the proof for \rt{TSwF}.
    \[\inference{
        \Lambda; \Delta \cdot (o, S) \vddash e:L \triangleright \Lambda'',o.f \mapsto C\<t\>[(\<l_i :u _i\>_{l_i\in L})^{\vv{E}}];\Delta'' \cdot (o,S'') \\
        \forall i. \Lambda'',o.f \mapsto C\<t\>[u_i^{\vv{E}}];\Delta'' \cdot (o,S'') \vddash e_i:t' \triangleright \Lambda';\Delta' \cdot (o, S')
    }{
    \Lambda; \Delta \cdot (o, S) \vddash \switchf{e}: t' \triangleright \Lambda'; \Delta' \cdot (o, S')
    }
    \]
        
    Assume that:
    $\Lambda; \Delta \cdot (o, S) \vddash \switchf{e} : t' \triangleright \Lambda'; \Delta' \cdot (o, S')$ where $\Delta = (envT_O \cdot envT_S)$, $o'' \not\in \text{dom}(\Lambda)$ and $o^{(3)} \not\in \text{dom}(envT_O)$.
    
    From the premise of \rt{TSwF} we have the following:
    
    \begin{itemize}
        \item $\Lambda; \Delta \cdot (o, S) \vddash e : L \triangleright \Lambda'', o.f\mapsto C\<t\>[(\<l_i : u_i\>_{l_i \in L})^{\vv{E}}]; \Delta'' \cdot (o, S'')$
        \item $\forall i.\ \Lambda'', o.f\mapsto C\<t\>[u_i^{\vv{E}}]; \Delta \cdot (o, S) \vddash e_i : t' \triangleright \Lambda'; \Delta' \cdot (o, S')$
    \end{itemize}
    
    We must show
    \begin{multline} 
    \label{eq:tswfweak}
    \Lambda, o'' \mapsto (C'\<t^N\>, envT_F); (envT_O, o^{(3)} \mapsto t^{(3)}) \cdot (o',S^{(3)}) \cdot envT_S \cdot (o, S) \vddash \\ \switchf{e} : t' \triangleright \Lambda', o''  \mapsto (C'\<t^N\>, envT_F); (envT_O', o^{(3)} \mapsto t^{(3)}) \cdot (o', S^{(3)}) \cdot envT_S' \cdot (o,S')
    \end{multline}
    
    From our induction hypothesis we have 
    \begin{multline}
    \label{eq:tswfweak-2}
        \Lambda, o'' \mapsto (C'\<t^N\>, envT_F); (envT_O, o^{(3)} \mapsto t^{(3)}) \cdot (o',S^{(3)}) \cdot envT_S \cdot (o, S) \vddash \\ e : L \triangleright \Lambda'', o'' \mapsto (C'\<t^N\>, envT_F); (envT_O'', o^{(3)} \mapsto t^{(3)}) \cdot (o', S^{(3)}) \cdot envT_S'' \cdot (o, S'')
    \end{multline}
    
    And from our induction hypothesis and (\ref{eq:tswfweak-2}) we have
    \begin{multline}
    \label{eq:tswfweak-3}
        \forall i.\ \Lambda'', o'' \mapsto (C'\<t^N\>, envT_F), o.f \mapsto C\<t\>[u_i^{\vv{E}}]; (envT_O, o^{(3)} \mapsto t^{(3)}) \cdot (o',S^{(3)}) \cdot envT_S \cdot (o, S) \vddash \\ e_i :t' \triangleright \Lambda', o'' \mapsto (C'\<t^N\>, envT_F); (envT_O', o^{(3)} \mapsto t^{(3)}) \cdot (o', S^{(3)}) \cdot envT_S' \cdot (o,S')
    \end{multline}
    
    From (\ref{eq:tswfweak-2}) and (\ref{eq:tswfweak-3}) we can conclude (\ref{eq:tswfweak}).
    \end{subproof}
    \item[\rt{TLinPar} \& \rt{TNoLPar}]
    
    \begin{subproof}
    The proof for \rt{TNoLPar} is similar to \rt{TLinPar}, hence we only show \rt{TLinPar}.
    We assume that  $\Lambda; \Delta \cdot (o, [x \mapsto t]) \vddash x : t \triangleright \Lambda;\Delta \cdot (o, [x \mapsto t])$ is correct, $o'' \not\in \text{dom}(\Lambda)$ and $o^{(3)} \not\in \text{dom}(envT_O)$.
    
    We show that
    \begin{multline}
    \label{eq:tlinparweak}
        \Lambda, o'' \mapsto (C\<t^N\>, envT_F); (envT_O, o^{(3)} \mapsto t', (o', S') \cdot envT_S \cdot (o,[x \mapsto t])) \vddash \\ x : t \triangleright \Lambda, o'' \mapsto (C\<t^N\>, envT_F); (envT_O, o^{(3)} \mapsto t', (o', S') \cdot envT_S \cdot (o,[x \mapsto t]))
    \end{multline}
    
    From our assumption we have that $\neg \sf{lin}(t)$, hence (\ref{eq:tlinparweak}) is correct.
    \end{subproof}
    \item[\rt{TLinFld} \& \rt{TNoLFld}] 
    
    \begin{subproof}
    The proofs for \rt{TLinFld} and \rt{TNoLFld} are very similar so we only show the former.
    Assume $\Lambda; \Delta \cdot (o, S) \vddash f : t \triangleright \Lambda\{o.f \mapsto  \bot\}; \Delta \cdot (o, S)$, $o' \not\in \text{dom}(\Lambda)$ and $o'' \not\in \text{dom}(envT_O)$, then from $\rt{TLinFld}$ we have that $t=\Lambda(o).f$ and $\lin(t)$.
    
    We must now show 
    \begin{multline}
        \label{eq:tlinfldweak}
        \Lambda, o' \mapsto (C\<t^N\>, envT_F); (envT_O, o'' \mapsto t') \cdot (o^{(3)}, S') \cdot envT_S \cdot (o, S) \vddash f : t \triangleright \\
        (\Lambda\{o.f \mapsto t\}), o' \mapsto (C\<t^N\>, envT_F); (envT_O, o'' \mapsto t') \cdot (o^{(3)}, S') \cdot envT_S \cdot (o, S)
    \end{multline}
    
    Since $o' \not\in \text{dom}(\Lambda)$ we can see that $((\Lambda\{o.f \mapsto t\}), o' \mapsto (C\<t^N\>, envT_F))(o).f = (\Lambda\{o.f \mapsto t\})(o).f=t$, and from our assumption, we know that $\lin(t)$, hence we can conclude that (\ref{eq:tlinfldweak}) is satisfied.
    \end{subproof}

    \item[\rt{TRet}]
    \begin{subproof}
    
    \[\inference{\Lambda; \Delta \vddash e:t \triangleright \Lambda';\Delta' & \Delta' = \Delta'' \cdot (o', [x \mapsto t']) & \sf{terminated(t')}}{\Lambda; \Delta \cdot (o,S) \vddash \sf{return}\{ e \} : t \triangleright \Lambda';\Delta'' \cdot (o,S)}\]
    
    We assume that $\Lambda; \Delta \cdot (o,S) \vddash \sf{return}\{ e \} : t \triangleright \Lambda';\Delta'' \cdot (o,S)$, $o^{(3)} \not\in \text{dom}(\Lambda)$ and $o^{(4)} \not\in \text{dom}(envT_O)$. We let $\Delta=envT_S \cdot envT_O$.
    
    We then show that:
    \begin{multline}
    \label{eq:tretweak}
    \Lambda, o^{(3)} \mapsto (C\<t^N\>, envT_F); (envT_O, o^{(4)} \mapsto t'') \cdot (o'',S') \cdot envT_S \cdot (o,S) \vddash \\ \sf{return}\{ e \} : t \triangleright \Lambda', o^{(3)} \mapsto (C\<t^N\>, envT_F); (envT_O', o^{(4)} \mapsto t'') \cdot (o'',S') \cdot envT_S' \cdot (o,S)
    \end{multline}
    
    From our induction hypothesis:
    \begin{multline}
    \label{eq:tretweak-2}
    \Lambda, o^{(3)} \mapsto (C\<t^N\>, envT_F); (envT_O, o^{(4)} \mapsto t'') \cdot (o'',S') \cdot envT_S \vddash \\ e : t \triangleright \Lambda', o^{(3)} \mapsto (C\<t^N\>, envT_F); (envT_O', o^{(4)} \mapsto t'') \cdot (o'',S') \cdot envT_S'
    \end{multline}
    
    
    From our assumption we have: $\sf{terminated}(t')$ and from (\ref{eq:tretweak-2}) we can now conclude (\ref{eq:tretweak}).
    \end{subproof}
    
    \item[\rt{TObj}]
    \begin{subproof}
    
    \[\inference{envT_O = envT_O', o \mapsto t}{\Lambda; envT_O \cdot envT_S \vddash o:t \triangleright \Lambda; envT_O' \cdot envT_S}\]
    
    We assume that $\Lambda; envT_O \cdot envT_S \vddash o:t \triangleright \Lambda; envT'_O \cdot envT_S$, $o' \not\in \text{dom}(\Lambda)$ and $o'' \not\in \text{dom}(envT_O)$.
    
    We then show that
    \begin{multline}
    \label{eq:tobjweak}
        \Lambda, o' \mapsto (C\<t^N\>, envT_F); (envT_O, o'' \mapsto t') \cdot (o^{(3)}, S) \cdot envT_S \vddash \\ o : t \triangleright \Lambda, o' \mapsto (C\<t^N\>, envT_F); (envT'_O, o'' \mapsto t') \cdot (o^{(3)}, S) \cdot envT_S
    \end{multline}
    
    From our assumption we have that
    \begin{equation}
    \label{eq:tobjweak-2}
        (envT_O, o'' \mapsto t') = (envT_O', o \mapsto t), o'' \mapsto t' = (envT_O', o'' \mapsto t), o \mapsto t
    \end{equation}
    Because $o \neq o''$.
    
    We can now conclude (\ref{eq:tobjweak}) is correct from (\ref{eq:tobjweak-2}).
    \end{subproof}
    
    \item[\rt{TCon}]
    \begin{subproof}
    
    Assume $\Lambda; \Delta \vddash^{\Omega'} \sf{continue}\ k:\sf{void} \triangleright \Lambda'; \Delta'$, $o \not\in \text{dom}(\Lambda)$ and $o' \not\in \text{dom}(envT_O)$. 
    
    Then we have $\Omega'=\Omega, k: (\Lambda, \Delta)$. We now show that $\Lambda, o \mapsto (C\<t^N\>, envT_F);(envT_O, o' \mapsto t) \cdot (o'', S) \cdot envT_S \vddash^{\Omega'} \sf{continue}\ k : \sf{void} \triangleright \Lambda', o \mapsto (C\<t^N\>, envT_F);(envT_O', o' \mapsto t) \cdot (o'', S) \cdot envT_S'$. This is trivially satisfied, as the premise for $\Omega'$ is still satisfied, and the environments are not changed in the rule.
    \end{subproof}
    
\end{enumerate}

\end{proof}
\subsection{Well-typed sub-configurations proof}
\label{app:subexp}
\tagthm{35}

Well-typed sub-configurations tells us that we can type the first sub-configuration of a larger configuration using the same environments. We use this lemma in the subject reduction proof in all cases of composite expressions, except \rt{RetC} where a similar lemma is employed. This lemma allows us say that the sub-configuration before a small-step transition is well-typed in the same context as the overall configuration in a composite expression.

\begin{lemma}[Well-typedness of sub-configurations]
\label{lemma:subexpression_welltyped}

For any configuration $c$ on the following form:
\begin{itemize}
    \item $\<h, env_S, \sif{e}{e_1}{e_2}\>$
    \item $\<h, env_S, r.m(e)\>$
    \item $\<h, env_S, f = e\>$
    \item $\<h, env_S, e;e'\>$
    \item $\<h, env_S, \sf{switch}_{r.m}\ (e) \{l_i : e_i\}_{l_i \in L}\>$
\end{itemize}

We have that:

\[\exists \Lambda, \Delta . \ \Lambda;\Delta \vddash c : t \triangleright \Lambda';\Delta' \implies \Lambda;\Delta \vddash \<h, env_S, e\> : t' \triangleright \Lambda'';\Delta''\]
\end{lemma}
\begin{proof}
We prove this result by a case analysis of the typing rules for expression used to show that the expression in the configuration is well typed.

\begin{enumerate}[ncases]
    \item[\rt{IfC}]
    
    \begin{subproof}
    Assume that the configuration is well typed, that is: 
    
    $\Lambda;\Delta \vddash \<h,env_S, \sif{e}{e_1}{e_2}\> : t \triangleright \Lambda';\Delta'$ where $\Delta=envT_O\cdot envT_S$. We now show that $\<h, env_S, e\>$ is a well typed configuration.

From \rt{WTC} we know:
\begin{itemize}
    \item $\Lambda \vddash h$
    \item $envT_S \vdash^h_{\sif{e}{e_1}{e_2}} env_S$
    \item $envT_O \vhdash \sif{e}{e_1}{e_2}$
    \item $\Lambda;\Delta\vddash \sif{e}{e_1}{e_2} : t \triangleright \Lambda';\Delta'$
    \item $\Lambda \vhdash \vv{D}$
\end{itemize}

From \rt{TIf} we know that $\Lambda;\Delta \vddash e : \sf{Bool} \triangleright \Lambda'';\Delta''$.

We can show that $\Lambda;\Delta \vddash \<h, env_S, e\> : \sf{Bool} \triangleright \Lambda''; \Delta''$, by showing each premise of \rt{WTC}.
\begin{itemize}
    \item $\Lambda \vddash h$ follows from assumption
    \item $envT_S \vdash^{h}_e$ follows from the premise of \rt{WTP-If} that we know is true from the assumptions.
    \item $envT_O \vhdash e$ follows from our assumption, and that it is a well formed expression, hence no objects can be mentioned in neither $e_1$ nor $e_2$, before $e$ is evaluated completely, hence $\objects{\sif{e}{e_1}{e_2}}=\objects{e}$.
    \item $\Lambda;\Delta \vddash e : \sf{Bool} \triangleright \Lambda'';\Delta''$ follows directly from assumptions.
    \item $\Lambda \vhdash \vv{D}$ follows from assumption
\end{itemize}

All premises of \rt{WTC} are satisfied, hence we can conclude $\Lambda;\Delta \vddash \<h, env_S, e\> : \sf{Bool} \triangleright \Lambda'';\Delta''$.

    \end{subproof}
    
    \item[\rt{MthdC}]
    
    \begin{subproof}
    The case is split into two cases, one for parameter calls and one for field calls. We only show the case for fields, since the two cases are very similar.

Assume $\Lambda;\Delta \vddash \<h, env_S, f.m(e)\> : t' \triangleright \Lambda'; \Delta'$ where $\Delta = envT_O \cdot envT_S$.

From \rt{WTC} we know:
\begin{itemize}
    \item $\Lambda \vddash h$
    \item $envT_S \vdash^h_{f.m(e)} env_S$
    \item $envT_O \vhdash f.m(e)$
    \item $\Lambda;\Delta\vddash f.m(e) : t' \triangleright \Lambda';\Delta'$
    \item $\Lambda \vhdash \vv{D}$
\end{itemize}

We can now show that $\Lambda;\Delta\vddash \<h, env_S, e\> : t \triangleright \Lambda'\{o.f \mapsto C\<t''\>[\W]\};\Delta''\cdot(o, S')$

\begin{itemize}
    \item $\Lambda \vddash h$ follows from assumption
    \item $envT_S \vdash^h_e env_S$ follows from \rt{WTP-Mthd} premise
    \item $envT_O \vhdash e$ is trivial since $\objects{f.m(e)}=\objects{e}$
    \item $\Lambda;\Delta\vddash e : t \triangleright \Lambda'\{o.f \mapsto C\<t''\>[\W]\};\Delta''\cdot(o, S') $ follows from \rt{TCallF}
    \item $\Lambda \vhdash \vv{D}$ follows from assumption
\end{itemize}

Hence we have $\Lambda;\Delta\vddash \<h, env_S, e\> : t \triangleright \Lambda'\{o.f \mapsto C\<t''\>[\W]\};\Delta''\cdot(o, S')$.

    \end{subproof}
    
    \item[\rt{FldC}]
    
    \begin{subproof}
    Assume that the configuration is well typed, that is: $\Lambda;\Delta\vddash\<h, env_S, f = e\> : \sf{void} \triangleright \Lambda'; \Delta'$, where $\Delta = envT_O \cdot envT_S$. 

From \rt{WTC}:
\begin{itemize}
    \item $\Lambda \vddash h$
    \item $envT_S \vdash^h_{f = e} env_S$
    \item $envT_O \vhdash f = e$
    \item $\Lambda;\Delta\vddash f = e : \sf{void} \triangleright \Lambda';\Delta'$
    \item $\Lambda \vhdash \vv{D}$
\end{itemize}

From \rt{TFld} we know $\Lambda;\Delta \vddash e : t \triangleright \Lambda\{o.f \mapsto t\};\Delta''\cdot(o, S)$. We can now show that $\Lambda;\Delta\vddash \<h, env_S, e\> : t \triangleright \Lambda\{o.f \mapsto t\};\Delta''\cdot(o, S)$ is a well typed configuration, by showing the premises of \rt{WTC}.

\begin{itemize}
    \item $\Lambda \vddash h$ follows from assumption
    \item $envT_S \vdash^h_e env_S$ follows from the premise of \rt{WTP-Fld}
    \item $envT_O \vhdash e$, trivial since $\objects{f = e} = \objects{e}$.
    \item $\Lambda;\Delta \vddash e : t \triangleright \Lambda'\{o.f \mapsto t\};\Delta''\cdot(o, S)$ follows from \rt{TFld}
    \item $\Lambda \vhdash \vv{D}$ follows from assumption
\end{itemize}
All premises of \rt{WTC} are satisfied, hence we can conclude $\Lambda;\Delta \vddash \<h, env_S, e\> : t \triangleright \Lambda'\{o.f \mapsto t\};\Delta''\cdot (o, S)$.

    \end{subproof}
    
    \item[\rt{SeqC}]
    
    \begin{subproof}
    Assume that $\Lambda;\Delta \vddash \<h, env_S, e;e'\> : t' \triangleright \Lambda'; \Delta'$. where $\Delta = envT_O \cdot envT_S$.

From \rt{WTC} we know:
\begin{itemize}
    \item $\Lambda \vddash h$
    \item $envT_S \vdash^h_{e;e'} env_S$
    \item $envT_O \vhdash e;e'$
    \item $\Lambda;\Delta\vddash e;e' : t' \triangleright \Lambda';\Delta'$
    \item $\Lambda \vhdash \vv{D}$
\end{itemize}

Then we can show the premises of \rt{WTC} for $\Lambda;\Delta \vddash \<h, env_S, e\> : t \triangleright \Lambda''; \Delta''$.

\begin{itemize}
    \item $\Lambda \vddash h$
    \item $envT_S \vdash^h_e env_S$ follows from the premise of \rt{WTP-Seq}
    \item $envT_O \vhdash e$ is satisfied, since because $e;e'$ is well formed, we know that $\objects{e;e'}=\objects{e}$. 
    \item $\Lambda;\Delta \vddash e : t \triangleright \Lambda'';\Delta''$ follows from \rt{TSeq}.
    \item $\Lambda \vhdash \vv{D}$ follows from assumption
\end{itemize}

Hence we have $\Lambda;\Delta \vddash \<h, env_S, e\> : t \triangleright \Lambda'';\Delta''$.
    \end{subproof}
    
    \item[\rt{SwC}]
    
    \begin{subproof}
    The proof is split into two cases, one for parameters and one for fields. Since the two cases are very similar, we only show the case for parameters.

Assume $\Lambda;\Delta \vddash \<h, env_S, \sf{switch}_{x.m}\ (e) \{l_i : e_i\}_{l_i \in L}\> : t' \triangleright \Lambda'; \Delta'$ where $\Delta = envT_O \cdot envT_S$.

From \rt{WTC} we know:
\begin{itemize}
    \item $\Lambda \vddash h$
    \item $envT_S \vdash^h_{\sf{switch}_{x.m}\ (e) \{l_i : e_i\}_{l_i \in L}} env_S$
    \item $envT_O \vhdash \sf{switch}_{x.m}\ (e) \{l_i : e_i\}_{l_i \in L}$
    \item $\Lambda;\Delta\vddash \sf{switch}_{x.m}\ (e) \{l_i : e_i\}_{l_i \in L} : t' \triangleright \Lambda';\Delta'$
    \item $\Lambda \vhdash \vv{D}$
\end{itemize}

We then show $\Lambda;\Delta\vddash \<h, env_S, e\> : L \triangleright \Lambda''; \Delta''$.

\begin{itemize}
    \item $\Lambda \vddash h$ follows from assumption
    \item $envT_S \vdash^h_e env_S$ follows from \rt{WTP-Sw}
    \item $envT_O \vhdash e$ follows since $\sf{switch}_{x.m}\ (e) \{l_i : e_i\}_{l_i \in L}$ is well formed, hence
    
    $\objects{\sf{switch}_{x.m}\ (e) \{l_i : e_i\}_{l_i \in L}} = \objects{e}$
    \item $\Lambda;\Delta\vddash e : L \triangleright \Lambda''; \Delta''$ follows from \rt{TSwP}
    \item $\Lambda \vhdash \vv{D}$ follows from assumption
\end{itemize}

Hence we have $\Lambda;\Delta\vddash \<h, env_S, e\> : L \triangleright \Lambda''; \Delta''$
    \end{subproof}
\end{enumerate}
\end{proof}

\subsection{Object agreement proof}
\label{app:object_agreement}
\tagthm{39}
Here we present results about typing expressions in the same environments. These results are used when typing branching expressions, where all branches are typed using the same environments.

\begin{proposition}
\label{lemma:object_start}
If $\Lambda;envT_O \cdot envT_S \vddash e : t \triangleright \Lambda';envT_O'\cdot envT_S'$ and $o\in \objects{e}$ then $o \in \dom{envT_O}$ and $o \not\in \dom{envT_O'}$.
\end{proposition}
\begin{proof}
We proceed with proof by induction in the type rules. In the cases \rt{TNew}, \rt{TNewGen}, \rt{TLit}, \rt{TVoid}, \rt{TBool}, \rt{TBot}, \rt{TLinPar}, \rt{TNoLPar}, \rt{TLinFld} and \rt{TNoLFld} this is trivially true, since $\objects{e}$ is empty. 

In the case of \rt{TObj} we can see that if $\Lambda;envT_O \cdot envT_S \vddash o : t \triangleright \Lambda;envT_O' \cdot envT_S$, then $envT_O = envT_O', o \mapsto t$. From this definition we can see that $o \in \dom{envT_O}$ but $o \not\in\dom{envT_O'}$, hence the proposition is satisfied. 

In the remaining cases, it follows directly from the induction hypothesis.
\end{proof}

\begin{proposition}
\label{lemma:object_end}
If $\Lambda;envT_O \cdot envT_S \vddash e : t \triangleright \Lambda';envT_O'\cdot envT_S'$, $o\in \dom{envT_O}$ and $o \not\in \objects{e}$ then $o \in \dom{envT_O'}$.
\end{proposition}
\begin{proof}
We proceed with proof by induction in the type rules. In the cases \rt{TNew}, \rt{TNewGen}, \rt{TLit}, \rt{TVoid}, \rt{TBool}, \rt{TBot}, \rt{TLinPar}, \rt{TNoLPar}, \rt{TLinFld} and \rt{TNoLFld} $envT_O=envT_O'$, hence the result is trivially true. 

In the case of \rt{TObj} we assume $\Lambda;envT_O \cdot envT_S \vddash o' : t \triangleright \Lambda;envT_O'\cdot envT_S'$. Since $o \not\in\objects{o'}$, we must have $o'\neq o$. But then $envT_O'$ is equal to $envT_O$ except in the case of $o'$, so if $o\in\dom{envT_O}$, we must have that $o\in\dom{envT_O'}$, hence the proposition is satisfied.

In the remaining cases, it follows from the induction hypothesis.
\end{proof}

\begin{lemma}{(Object agreement.)}
\label{lemma:objectagreement}
If $\Lambda;envT_O \cdot envT_S \vddash e : t \triangleright \Lambda'; envT_O'\cdot envT_S'$ and $\Lambda;envT_O \cdot envT_S \vddash e' : t' \triangleright \Lambda'; envT_O'\cdot envT_S'$ then $\sf{objects}(e)=\sf{objects}(e')$.

\end{lemma}
\begin{proof}
We proceed with proof by contradiction. Assume the contradictory situation, where $o \in \sf{objects}(e)$ and $o \not\in \sf{objects}(e')$. If $o \in \sf{objects}(e)$ then by Lemma \ref{lemma:object_start} we have $o \in envT_O$ and $o \not\in {envT_O'}$. But since $o\not\in\objects{e'}$ and $o\in\dom{envT_O}$ then by Lemma \ref{lemma:object_end} $o\in envT_O$. This is a contradiction, hence $\objects{e}=\objects{e'}$.
\end{proof}

\begin{corollary}[Empty objects for expressions]
\label{corollary:emptyobjects}
If $\Lambda;\Delta \vddash e : t \triangleright \Lambda;\Delta$ then $\sf{objects}(e) = \emptyset$
\end{corollary}
\begin{proof}
Since $\Lambda;\Delta\vddash \sf{true} : \sf{Bool} \triangleright \Lambda;\Delta$ can be concluded using \rt{TBool}, we can conclude using Lemma \ref{lemma:objectagreement} that $\sf{objects}(e)=\sf{objects}(\sf{true})=\emptyset$.
\end{proof}

\subsection{Consistency proof}
\label{app:consistency}
\tagthm{43}
When typechecking expressions, we only consider the updates to the currently active object. This means that field updates in other objects are assumed to be correct, since it has been checked previously by \rt{TClass}. Therefore we only require the field typing environment to be consistent across reductions, for the active object. 

\begin{definition}[Field typing environment consistency]
We say that $\Lambda$ is consistent with $\Lambda'$ w.r.t object $o$, written $\Lambda \lambdaeq{o} \Lambda'$ iff $\Lambda(o)=\Lambda'(o)$.
\end{definition}

\begin{lemma}[Field typing consistency reduction]
\label{lemma:consistency}
If $\Lambda;\Delta \vddash e : t \triangleright \Lambda';\Delta'$, $\Lambda^N\lambdaeq{o}\Lambda$ and $\sf{returns}(e)=0$ then $\Lambda^N;\Delta \vddash e : t \triangleright \Lambda^{N'};\Delta'$ where $\Lambda^{N'}\lambdaeq{o}\Lambda'$.
\end{lemma}
\begin{proof}
Induction in the structure of type judgments of expressions. It is easy to see that the lemma must be satisfied, since only \rt{TRet} can access other objects in $\Lambda$ than the currently active object, and in that case we do not have $\sf{returns}(e)=0$. 
\end{proof}
\subsection{Subject reduction proof}
\label{app:subject_reduction}
\tagthm{34}

In this section we present the complete proof of the subject reduction (preservation) property of our type system which represents one part of the safety theorem. Subject reduction tells us that an well-typed expression remains well-typed after a reduction step in its evaluation. The general approach for this proof is to use the information gained from the premises of \rt{WTC} and other rules that must have been used to type the expression of the configuration. This information tell us about the context after a reduction step and then we can verify that the premises of \rt{WTC} are still valid in the new configuration. 

\begin{theorem}[Subject reduction]
\label{thm:subjectreduction}
Let $\vv{D}$ be such that $\vdash \vv{D}$ and let $\<h, env_S, e\>$ be a configuration. If $\vv{D} \vddash \<h, env_S, e\> \trans \<h', env_S', e'\>$ then:
\[
    \exists \Lambda, \Delta\ .\ \Lambda, \Delta \vddash \<h, env_S, e\> : t \triangleright \Lambda'; \Delta' \implies \exists \Lambda^N, \Delta^N\ .\ \Lambda^N, \Delta^N \vddash \<h', env_S', e' \> : t \triangleright \Lambda'';\Delta'
\]

where $\Lambda' \lambdaeq{o} \Lambda''$ and $o$ is the active object in the resulting configuration.
\end{theorem}
\begin{proof}
By induction in the structure of the reduction rules. For each reduction rule, we assume that the configuration is well typed, and show that the resulting configuration after the reduction is also well typed. We construct new environments $\Lambda^N, \Delta^N$ and use the premises of the \rt{WTC} rule, to conclude that the resulting configuration is well typed, based on the construction.

\begin{enumerate}[ncases]
    \item[\rt{uDeref}]
    \begin{subproof}
    \[\inference{h(o).f = v & \neg \lin(v,h)}{\vdash_{\vv{D}} \<h, (o,s)\cdot env_S, f\> \trans \<h, (o, s)\cdot env_S, v\>}\]
    
    Assume that $\exists \Lambda, \Delta\ .\ \Lambda; \Delta \vddash \<h, env_S, f\> : t \triangleright \Lambda'; \Delta'$, where $\Delta = (envT_O \cdot envT_S)$.
    
    From \rt{WTC} we know:
    \begin{itemize}
        \item $\Lambda \vddash h$
        \item $envT_S\vhdash_{f} (o, s) \cdot env_S$
        \item $envT_O \vhdash f$
        \item $\Lambda;envT_O \cdot envT_S \vddash f : t_\bot \triangleright \Lambda'; \Delta'$
        \item $\Lambda \vhdash \vv{D}$
    \end{itemize}
    
    
    From \rt{WTE} we know:
    \begin{itemize}
        \item $\dom{envT_O} = \objects{f}=\emptyset$
    \end{itemize}

    From \rt{uDeref} we know:
    \begin{itemize}
        \item $h(o).f = v$
        \item $\neg \lin(v, h)$
        \item $h' = h$
        \item $(o, s) \cdot env_S' = (o, s) \cdot env_S$
    \end{itemize}
    
    Since $\neg\lin(v, h)$ we must have $\neg \lin(t_{\bot})$, hence we must have concluded $\Lambda;envT_O \cdot envT_S \vddash f : t_\bot \triangleright \Lambda'; \Delta'$ using \rt{TNoLFld}.
    
    From \rt{TNoLFld} and \rt{WTC} we have:
    \begin{itemize}
        \item $\Lambda' = \Lambda$
        \item $\Delta' = \Delta = envT_O \cdot envT_S$
    \end{itemize}
    
    We wish to show that $\exists \Lambda^N, \Delta^N.\ \Lambda^N;\Delta^N \vddash \<h, (o, s) \cdot env_S, v\> : t \triangleright \Lambda'; \Delta'$.
    
    For the cases where $v$ is not an object, we can let $\Lambda^N=\Lambda, \Delta^N = (\emptyset \cdot envT_S)$. 
    We have the following from our hypothesis:
    \begin{itemize}
        \item $\Lambda \vddash h$
        \item $envT_S \vdhash_{v} (o, s) \cdot env_S$, by definition of \rt{WTP-Base}
    \end{itemize}
    
    Since $\dom{\emptyset} = \sf{objects}(v) = \emptyset$, we can conclude that $\emptyset \vhdash v$. 
    For each of the basic types we can use the corresponding rule \rt{TLit}, \rt{TBool}, \rt{TVoid}, \rt{TNull} to conclude that $\Lambda^N; \Delta^N \vddash v : t \triangleright \Lambda^N; \Delta^N$. Finally we can conclude that since $\Lambda^N = \Lambda = \Lambda'$ we have that $\Lambda^N \vhdash \vv{D}$. In fact, since $\Lambda^N=\Lambda'$ and $\Delta^N = \emptyset \cdot envT_S$, we have $\Lambda^N; \Delta^N \vddash v : t \triangleright \Lambda'; \Delta'$. This means that all premises are satisfied and that $\Lambda^N;\Delta^N \vddash \<h, (o, s) \cdot env_S, v\> : t \triangleright \Lambda';\Delta'$. 
    
    For the case where $v$ is an object that is terminated, we let $\Lambda^N = \Lambda, \Delta^N = (\emptyset, v \mapsto t) \cdot envT_S$. Now we have $\emptyset, v \mapsto t \vhdash v$ since $\sf{objects}(v)=\{v\} = \dom{\emptyset, v \mapsto t}$ and $\sf{getType}(v, h) = t = (\emptyset, v \mapsto t)(v)$.

    We can conclude that $\Lambda^N;\Delta^N \vddash v : t \triangleright \Lambda^N; \emptyset \cdot envT_S$ using the \rt{TObj}. Since $envT_O = \emptyset$, we have that $\Lambda^N;\Delta^N \vddash v : t \triangleright \Lambda'; \Delta'$ Using the same reasoning as the previous case, we can conclude that $\Lambda^N;\Delta^N \vddash \<h, env_S, v\> : t \triangleright \Lambda'; \Delta'$. $\Lambda'\lambdaeq{o} \Lambda'$ is satisfied because of reflexivity of $\lambdaeq{o}$.
    \end{subproof}
    
    \item[\rt{TCallF}]
    \begin{subproof}
    \[\inference{env_S=(o,s)\cdot env_S' & o' = h(o).f \\ \_\ m(\_ x)\{e\} \in h(o').\sf{class}.\sf{methods}_{\vv{D}} & h(o').\sf{usage} \trans[m] \W}{\vddash \<h, env_S, f.m(v)\> \trans \<h\{\W/h(o').\sf{usage}\}, (o', [x \mapsto v])\cdot env_S, \sf{return}\{e\}\>}\]
    
    Assume $\Lambda; \Delta \vddash \< h, env_S, f.m(v) \> : t \triangleright \Lambda'; \Delta'$ where $\Delta = envT_O \cdot envT_S$
    
    From \rt{CallF} we know
    \begin{itemize}
        \item $t\ m(t'\ x)\{e\} \in h(o').\sf{class}.\sf{methods}_{\vv{D}}$
        \item $h(o').\sf{usage} \trans[m] \W$
    \end{itemize}
    
    From \rt{WTC} we know
    \begin{itemize}
        \item $\Lambda \vddash h$
        \item $envT_S \vhdash_{f.m(v)} env_S$
        \item $envT_O \vhdash f.m(v)$
        \item $\Lambda; \Delta \vddash f.m(v) : t \triangleright \Lambda'; \Delta'$
        \item $\Lambda \vhdash \vv{D}$
    \end{itemize}
    
    From \rt{WTP} we know
    \begin{itemize}
        \item $envT_S = envT_S'' \cdot (o, S)$
        \item $env_S = (o, s)$
    \end{itemize}
    
    From \rt{TCallF} and typing rules for values we know
    \begin{itemize}
        \item $\Lambda;\Delta \vdash v:t' \triangleright \Lambda\{o.f \mapsto C\<t''\>[\mathcal{U}]\};\Delta'$
        \item $\Lambda(o).f = C\<t''\>[\U]$
        \item $\U \trans[m] \W$
        \item $\Lambda' = \Lambda\{o.f \mapsto C\<t''\>[\W]\}$
    \end{itemize}
    
    From \rt{WTE} we know 
    \begin{itemize}
        \item $\dom{envT_O} = \objects{f.m(v)}$
    \end{itemize}
    
    From \rt{TCallF} and typing rules for values we know
    \begin{itemize}
        \item $\Lambda' = \emptyset \cdot envT_S$
    \end{itemize}
    
    Let $h' = h\substitute{\W}{h(o').\sf{usage}}$
    
    Let $\Lambda^N = \Lambda'$, $\Delta^N = \emptyset \cdot envT_S'' \cdot (o', [x \mapsto t']) \cdot (o, S)$
    
    We must now conclude $\Lambda^N; \Delta^N \vddash \<h', (o', [x \mapsto v]) \cdot env_S, \return{e}\> : t \triangleright \Lambda^{N'}; \Delta^{N'}$
    \begin{itemize}
        \item $\Lambda^N \vddash h'$ since we only update the field $o.f$ in both $\Lambda^N$ and $h'$ and by assumption we know $\Lambda^{N}(o).\sf{envT}_{\sf{F}}(f) = \sf{getType}(o', h') = C\<t''\>[\W]$.
        \item $envT_S^N \vhpdash_{\return{e}} (o', [x \mapsto v]) \cdot env_S$ is satisfied by \rt{WTP-Ret}
        \item $envT_O^N \vdhpash \return{e}$ since $\objects{\return{e}} = \emptyset$
    \end{itemize}
    
    We wish to show $\Lambda^N; \emptyset \cdot envT_S'' \cdot (o', [x \mapsto t']) \cdot (o, S) \vddash \return{e} : t \triangleright \Lambda^{N'}; \Delta^{N'}$
    
    From \rt{TCBR} have that:
    \begin{itemize}
        \item $\{\sf{this} \mapsto \Lambda^N(o').\sf{envT}_\sf{F}\}; \emptyset \cdot (o', [x \mapsto t']) \vddash e : t \triangleright \{this \mapsto  \sf{envT}_\sf{F}'\}; \emptyset \cdot (o', [x \mapsto t'''])$
        \item $\sf{terminated}(t''')$
    \end{itemize}
    
    We can then use Lemma \ref{lemma:weakening} to conclude
    $\Lambda^N; \emptyset \cdot envT_S'' \cdot (o', [x \mapsto t']) \vddash e : t \triangleright \Lambda^N\{o' \mapsto (C\<t''\>, \sf{envT}_\sf{F}')\}; \emptyset \cdot envT_S'' \cdot (o', [x \mapsto t'''])$ 
    
    Since we have that $\sf{terminated}(t''')$, we can conclude with \rt{TRet}:
    \begin{itemize}
        \item $\Lambda^N ; \Delta^N \vddash \return{e} : t \triangleright \Lambda^{N'}; \Delta'$
    \end{itemize}
    
    The only difference between $\Lambda'$ and $\Lambda^{N'}$ is for the object $o'$, hence we must have that $\Lambda' \lambdaeq{o} \Lambda^{N'}$.

    \end{subproof}
    
    \item[\rt{TCallP}]
    
    \begin{subproof}
    \[\inference{env_S = (o, [x' \mapsto o'])\cdot env_S'  & \_\ m(\_\ x)\{e\}\in h(o').\sf{class}.\sf{methods}_{\vv{D}} \\ h(o').\sf{usage} \trans[m] \W}{\vddash \<h, env_S, x'.m(v)\> \trans \<h\{\W/h(o').\sf{usage}\}, (o', [x \mapsto v])\cdot env_S,\sf{return}\{e\}\>} \]
    
    Assume $\Lambda;\Delta \vddash \<h, env_S, x'.m(v)\> : t \triangleright \Lambda';\Delta'$, where $\Delta=envT_O \cdot envT_S$.
    
    From \rt{CallP} we know
    \begin{itemize}
        \item $env_s = (o, [x' \mapsto o']) \cdot env_S'$
        \item $t\ m(t'\ x)\{e\}\in h(o').\sf{class}.\sf{methods}_{\vv{D}}$
        \item $h(o').\sf{usage} \trans[m] \W$
    \end{itemize}
    
    From \rt{WTC} we know
    \begin{itemize}
        \item $\Lambda \vddash h$
        \item $envT_S \vhdash_{x'.m(v)} env_S$
        \item $envT_O \vhdash x'.m(v)$
        \item $\Lambda;\Delta \vddash x'.m(v) : t \triangleright \Lambda';\Delta'$
        \item $\Lambda \vhdash \vv{D}$
    \end{itemize}
    
    From \rt{WTP-Base} we know
    \begin{itemize}
        \item $envT_S = envT_S'' \cdot (o, [x' \mapsto C\<t'''\>[\U]])$
        \item $env_S = (o, [x' \mapsto o'])$
    \end{itemize}
    
    From \rt{WTE} we know:
    \begin{itemize}
        \item $\dom{envT_O} = \objects{x'.m(v)}=\objects{v}$
    \end{itemize}
    
    From \rt{TCallP} and typing rules for values we have
    \begin{itemize}
        \item $\U \trans[m] \W$
        \item $\Lambda' = \Lambda$
        \item $\Delta' = \emptyset \cdot env_T'' \cdot (o, [x' \mapsto C\<t'''\>[\W]])$
    \end{itemize}
    
    Now let $\Lambda^N = \Lambda$, $\Delta^N = \emptyset \cdot envT_S'' \cdot (o', [x \mapsto t']) \cdot (o, [x' \mapsto C\<t'''\>[\W]])$
    
    We must now conclude $\Lambda^N;\Delta^N \vddash \<h\substitute{\W}{h(o').usage}, (o', [x \mapsto v]) \cdot env_S, \return{e}\>$.
    
    \begin{itemize}
        \item $\Lambda^N \vddash h'$ since because of linearity we know that $o'$ does not appear in any fields, and we have $\dom{h'}=\dom{h}=\dom{\Lambda^N}$.
        \item $envT_S^N \vhpdash_{\return{e}} (o', [x \mapsto v]) \cdot env_S$ is satisfied with \rt{WTP-Ret}
        \item $\emptyset \vhdash \return{e}$ is trivially true since $\objects{\return{e}} = \emptyset$
        \item $\Lambda^N \vhpdash \vv{D}$ is fulfiled since per \rt{TCbr} we are following the usage of $o'$, which will eventually lead to a terminated environment.
    \end{itemize}
    
    We must show that $\Lambda^N;\Delta^N \vddash \return{e} : \triangleright \Lambda^N;\Delta'$.
    
    From \rt{TCbr} we have that 
    \begin{itemize}
        \item $\{\sf{this} \mapsto \Lambda^N(o').\sf{envT}_{\sf{F}}\};\emptyset \cdot (o', [x \mapsto t']) \vddash e : t \triangleright \{\sf{this} \mapsto envT_F'\};\emptyset \cdot (\sf{this}, [x \mapsto t'''])$
        \item $\sf{terminated}(t''')$
    \end{itemize}
    
    We can then use Lemma \ref{lemma:weakening} to conclude $\Lambda^N;\emptyset\cdot envT_S'' \cdot (o', [x \mapsto t']) \vddash e : t \triangleright \Lambda^{N'};\emptyset \cdot envT_S''\cdot(o', [x \mapsto t'''])$.
    
    Since we have that $\sf{terminated(t''')}$, we can conclude with \rt{TRet} that:
    \begin{itemize}
        \item $\Lambda^N;\Delta^N\vddash \return{e} : t \triangleright \Lambda^{N'};\Delta'$
    \end{itemize}
    
    The only difference between $\Lambda'$ and $\Lambda^{N'}$ is for the object $o'$, hence we must have that $\Lambda' \lambdaeq{o} \Lambda^{N'}$.
    
    \end{subproof}
    
    \item[\rt{Ret}]
    \begin{subproof}
    \[\inference{v \neq v' \Rightarrow \neg \lin(v', h)}{\vddash \<h, (o, [x \mapsto v'])\cdot env_S,\sf{return}\{v\}\> \trans \<h, env_S, v\>}\]
    
    We assume that $\Lambda; envT_O \cdot envT_S \cdot (o', S) \vddash \<h, (o, [x \mapsto v']) \cdot env_S, \sf{return}\{v\}\> : t \triangleright \Lambda'; envT_O' \cdot envT_S'' \cdot (o', S)$.
    
    From \rt{TRet} we have:
    \begin{itemize}
        \item $\Lambda, envT_O \cdot envT_S \vddash v : t \triangleright \Lambda'; envT_O' \cdot envT_S'$
        \item $envT_S' = envT_S'' \cdot (o'', S')$
    \end{itemize}
    
    From the typing rules for values we have:
    \begin{itemize}
        \item $\Lambda' =\Lambda$
    \end{itemize}
    
    From \rt{WTP-Ret} we have:
    \begin{itemize}
        \item $envT_S'' \cdot (o'', [x \mapsto t']) \cdot (o', [x' \mapsto t'']) \vdhash_{\sf{return}\ v} (o, [x \mapsto v']) \cdot (o', [x' \mapsto v''])$
        \item $o = o''$
        \item $\sf{getType}(v'', h) = t''$
        \item $env_S = (o', [x' \mapsto v''])$
    \end{itemize}
    From \rt{WTP-Base} we have:
    \begin{itemize}
        \item $envT_S'' \cdot (o, [x \mapsto t']) \vdhash_{v} (o, [x \mapsto v'])$
        \item $\sf{getType}(v', h) = t'$
    \end{itemize}
 
    Let $\Lambda^N = \Lambda$, $envT_O^N = envT_O$, $envT_S^N = envT_S'' \cdot (o', [x \mapsto t''])$.
    We now need to show that the premises of \rt{WTC} are satisfied. 
    
    Show \rt{WTP} using $envT_S^N \vhdash_{v} env_S$ then we have $envT_S'' \cdot (o', [x \mapsto t'']) \vhdash_{v} (o', [x' \mapsto v''])$. From assumption we have that $\sf{getType}(v'', h) = t''$, hence we conclude using \rt{WTP-Base}
    
    Case where $v$ is an object.
    \begin{itemize}
        \item $envT_O = envT_O', v \mapsto t$
        \item $\dom{envT_O} = \{v\}$
        \item $\Lambda^N \vddash h$ (From hypothesis)
        \item $envT_O^N \vhdash v$ trivial since $\sf{objects}(\sf{return}\ \{v\})=\sf{objects}(v) = \{v\}$.
        \item $\Lambda^N;envT_O^N\cdot envT_S^N \vddash v : t \triangleright \Lambda'; envT_O' \cdot envT_S''\cdot(o', S)$ from \rt{TObj}
        \item $\Lambda^N \vhdash \vv{D}$ (From hypothesis)
    \end{itemize}
    
    Case where $v$ is not object.
    \begin{itemize}
        \item $envT_O = envT_O'$
        \item $\dom{envT_O} = \emptyset$
        \item $\Lambda^N \vddash h$ (From hypothesis)
        \item $envT_O^N \vhdash v$ trivial since $\sf{objects}(\sf{return}\ \{v\})=\sf{objects}(v) = \emptyset$.
        \item $\Lambda^N;envT_O^N\cdot envT_S^N \vddash v : t \triangleright \Lambda'; envT_O' \cdot envT_S''\cdot(o', S)$ from \rt{TLit}
        \item $\Lambda^N \vhdash \vv{D}$ (From hypothesis)
    \end{itemize}
    Hence we can conclude 
    $\Lambda^N;envT_O^N\cdot envT_S^N \vddash \<h, env_S, v\> : t \triangleright \Lambda';envT_O' \cdot envT_S'' \cdot (o', S)$
    
    $\Lambda'\lambdaeq{o} \Lambda'$ is satisfied because of reflexivity of $\lambdaeq{o}$.
    \end{subproof}

    \item[\rt{New}]
    \begin{subproof}
    \[\inference{o\text{ fresh} & C\<\bot\>.\sf{fields}_{\vv{D}} = \vv{F} & C\<\bot\>.\sf{usage}_{\vv{D}} = \W}{\vddash \<h, env_S, \sf{new}\ C \> \trans \<h \cup \{o \mapsto \<C\<\bot\>[\W], \vv{F}.\sf{initvals}_{\vv{D}}\>\}, env_S, o\>}\]
    
    Let $h'=h \cup \{o \mapsto \<C\<\bot\>[\W], \vv{F}.\sf{initvals}_{\vv{D}}\>\}$.
    
    Assume $\Lambda; envT_O \cdot envT_S \vddash \<h, env_S, \sf{new}\ C\> : C\<t\>[\W] \triangleright \Lambda';envT_O' \cdot envT_S'$.
    
    From \rt{TNew} we have
    \begin{itemize}
        \item $\Lambda'=\Lambda$
        \item $envT_O' = envT_O$
        \item $envT_S' = envT_S$
    \end{itemize}
    
    Now let $\Lambda^N = \Lambda\{o \mapsto \<C\<\bot\>, \vv{F}.\sf{inittypes}_{\vv{D}}\>\}$, $envT_O^N = envT_O, o\mapsto C\<\bot\>[\W], envT_S^N = envT_S$, and show that $\Lambda^N;envT_O^N\cdot envT_S^N \vddash \<h\cup\{o\mapsto C\<\bot\>[\W], \vv{F}.\sf{initvals}_{\vv{D}}\>\}, env_S, o\> : C\<\bot\>[\W] \triangleright \Lambda';\Delta'$.
    
    To show $\Lambda^N\vddash h'$, we must show that \rt{WTH} holds for $o$, since the remaining objects follows from the assumption that $\Lambda \vddash h$.
    
    From the updates to $h$ and $\Lambda$, we can conclude that 
    
    $h'(o).\sf{fields}\allowbreak=\Lambda^N(o).\sf{fields}= h'(o).\sf{class}.\sf{fields}_{\vv{D}}=\vv{F}$. It follows from the definition of $\vv{F}.\sf{initvals}_{\vv{D}}$ and $\vv{F}.\sf{inittypes}_{\vv{D}}$ that $\forall f\in\vv{F} . \Lambda^N(o).\sf{envT}_\sf{F}(f) = \sf{getType}(h'(o).f, h')$. 
    
    Since we know that $\dom{\Lambda} = \dom{h}$, we can conclude that $\dom{\Lambda^N} = \dom{h'}$, since $o$ is now in both domains.
    
    We can conclude that $envT_S^N \vhdash_{o} env_S$ since $o$ must be concluded with \rt{WTP-Base} and $\sf{new}\ C$ was concluded with \rt{WTP-Base}, hence it follows from the assumption $envT_S \vhdash_{\sf{new}\ C} env_S$.
    
    $envT_O^N \vdash^{h'} o$ is satisfied, since based on the assumption $\dom{envT_O}=\sf{objects}(\sf{new}\ C) = \emptyset$, hence $\dom{envT_O^N}=\sf{objects}(o) = \llbar o \rrbar$.
    
    $\Lambda^N;envT_O^N\cdot envT_S \vddash o : C\<\bot\>[\W] \triangleright \Lambda^N ; \Delta$ follows directly from \rt{TObj}.
    
    Finally $\Lambda^N \vhpdash \vv{D}$ must be shown to hold for $o$. Other objects follow from the assumption $\Lambda \vhdash \vv{D}$. It follows from \rt{TProg} and in turn \rt{TClass} that $\emptyset; \vv{F}.\sf{inittypes}_{\vv{D}}\vddash C\<\bot\>[\W]\triangleright \Gamma$ and that $\sf{terminated}(\Gamma)$.
    
    All premises of \rt{WTC} are satisfied, hence we can conclude $\Lambda^N;envT_O^N\cdot envT_S^N \vddash \<h\cup\{o\mapsto C\<\bot\>[\W], \vv{F}.\sf{initvals}_{\vv{D}}\>\}, env_S, o\> : C\<\bot\>[\W] \triangleright \Lambda^N; \Delta$.
    
    Since $o\not\in\dom{\Lambda'}$, and $\Lambda^N$ is equal to $\Lambda$ except for $o$. We have that $\Lambda' \lambdaeq{o'} \Lambda^N$ for the active object $o'$.

    \end{subproof}
    \item[\rt{NewGen}]
    This case is the same as \rt{New}, except all occurrences of $\bot$ is replaced with $g$.

    \item[\rt{Seq}]
    \begin{subproof}
    \[ \inference{\neg \lin(v,h)}{\vddash \<h, env_S, v;e\> \trans \< h, env_S, e\> }\]
    
    Assume $\Lambda; envT_O \cdot envT_S \vddash \<h, env_S, v;e\> : t \triangleright \Lambda'; envT_O' \cdot envT_S'$
    
    From \rt{WTC} we have
    
    \begin{itemize}
        \item $\Lambda \vdash h$
        \item $envT_S \vdhash_{v;e} env_S$
        \item $envT_O \vdhash v ; e$
        \item $\Lambda; envT_O \cdot envT_S \vddash v ; e :t \triangleright \Lambda'; envT_O' \cdot envT_S'$
        \item $\Lambda \vdhash \vv{D} $ 
    \end{itemize}
    
    From \rt{WTE} we have
    \begin{itemize}
        \item $\forall o \in \dom{envTo}\ .\ envT_O(o) = \sf{getType}(o, h)$
        \item $\dom{envT_O} = \sf{objects}(e)$
    \end{itemize}
    From \rt{Seq} we have
    \begin{itemize}
        \item $\neg \lin(v,h)$
    \end{itemize}
    
    From \rt{TSeq} we have
    \begin{itemize}
        \item $\Lambda; envT_O \cdot envT_S \vdash v : t' \vddash v : t' \triangleright \Lambda''; envT_O'' \cdot envT_S''$
        \item $\neg \lin(t')$
        \item $\Lambda''; envT_O'' \cdot envT_S'' \vdash e  : t\triangleright \Lambda'; envT_O' \cdot envT_S'$
    \end{itemize}
    
    From \rt{WTP-Seq} we have $envT_S \vhdash_{v} env_S$
    
    From \rt{WTP-Base} we have:
    \begin{itemize}
        \item $envT_S' \cdot (o, [x \mapsto t']) \vhdash_{v} (o, [x \mapsto v'])$
        \item $\sf{getType}(v', h) = t'$
    \end{itemize}
    
    $\Lambda; envT_O \cdot envT_S \vddash v : t' \triangleright \Lambda'' \cdot envT_S''$ must have been concluded with typing rules for values, hence we have:
    \begin{itemize}
        \item case $v \neq o$: $envT_O = envT_O''$ from \rt{TLit}, \rt{TVoid}, \rt{TBool} or \rt{TBot}.
        \item case $v = o$: $envT_O = envT_O'', o \mapsto t''$
        \item $\Lambda = \Lambda''$
        \item $envT_S = envT_S''$
    \end{itemize}

    Show $\exists \Lambda^N, \Delta^N . \Lambda^N; \Delta^N \vddash \<h, env_s, e\>$ Let $\Lambda ^N = \Lambda$, $\Delta^N = envT_O^N \cdot envT_S^N $, $envT_O^N = envT_O''$ and $envT_S^N = envT_S$ then we have 
    
    \begin{itemize}
        \item $\Lambda \vdash h$ from assumption
        \item $envT_s^N \vdhash_{e} env_S$ from assumption since $\sf{returns}(e) = 0$ because $e$ is a well formed expression, hence it is concluded with \rt{WTP-Ret}.
        \item $\Lambda^N; envT_O^N \cdot envT_S^N \vddash e \triangleright \Lambda'; envT_O' \cdot envT_S' $ from assumption
        \item $\Lambda^N \vdhash \vv{D}$ from assumption
    \end{itemize}
    
    We need to show $envT_O^N \vhdash e$. We have two cases
    
    case $v \neq o$: we know $envT_O = envT_O''$ then $envT_O'' \vhdash e$
    
    case $v = o$: 
    \begin{itemize}
        \item $envT_O = envT_O'', v \mapsto t''$.
        \item $\dom{envT_O''} = \sf{objects}(e)$
        \item $\forall o \in \dom{envT_O''} \ . \ \sf{getType}(o, h) = envT_O(o)$, we know that $v \not\in \dom{envT_O''}$ 
    \end{itemize}
    
    All premises of \rt{WTC} are satisfied, hence we can conclude $\Lambda^N;envT_O^N\cdot envT_S^N \vddash \< h, env_S, e\> : t \triangleright \Lambda'; envT_O' \cdot envT_S'$.
    
    Finally we can conclude that $\Lambda' \lambdaeq{o} \Lambda'$, because of reflexivity. 
    
    \end{subproof}
    \item[\rt{IfTrue}]
    \begin{subproof}
    \[
    \inference{}{\vddash \<h, env_S, \sf{if }(\sf{true}) \{e'\}\sf{ else } \{e''\}\> \trans \<h, env_S, e'\>}
    \]
    
    Assume $\Lambda;envT_O \cdot envT_S \vddash \<h, env_S, \sf{if}\ (\sf{true})\ \{e'\}\ \sf{else}\ \{e''\} \> : t \triangleright \Lambda'; \Delta'$.
    
    From \rt{WTP-If} we know $envT_S \vhdash_{v} env_S$ which must have been concluded using \rt{WTP-If} and then \rt{WTP-Base}.
    
    Let $\Lambda^N = \Lambda$, $envT_O^N = envT_O$ and $envT_S^N = envT_S$. We now show that $\Lambda^N;envT_O^N \cdot envT_S^N \vddash \<h, env_S, e'\> : t \triangleright \Lambda'; \Delta'$. 
    
    By Lemma \ref{lemma:objectagreement} we have that $\sf{objects}(e') = \sf{objects}(e'')$. Since $\sf{objects}(\sf{true})=\emptyset$ we get that $\sf{objects}(e')=\sf{objects}(\sf{if}\ (\sf{true})\ \{e'\}\ \sf{else}\ \{e''\} )$, hence from our assumption we get $envT_O^N \vhdash e'$.
    
    We have $envT_S^N \vhdash_{e'} env_S$, $\Lambda^N \vddash h$, $e'$ is a well formed expression hence we have $\sf{returns}(e') = 0$, hence it is eventually concluded using \rt{WTP-Base}. 
    
    $\Lambda \vddash \vv{D}$ follows directly from our assumptions. 
    
    We can conclude $\Lambda^N;envT_O^N\cdot envT_S^N \vddash e' : t \triangleright \Lambda';\Delta'$ from the premises of \rt{TIf}.
    
    We have shown all premises of \rt{WTC}, hence we can conclude $\Lambda^N;envT_O^N \cdot envT_S^N \vddash \<h, env_S, e'\> : t \triangleright \Lambda'; \Delta'$.
    
    $\Lambda'\lambdaeq{o} \Lambda'$ is satisfied because of reflexivity of $\lambdaeq{o}$.
    
    \end{subproof}
    \item[\rt{IfFls}] Follows the same structure as \rt{IfTrue}.
    
    \item[\rt{Lbl}]
    \begin{subproof}
    \[
    \inference{}{\vddash \<h, env_S, k:e\> \trans \<h, env_S, e\{k:e/\sf{continue}\ k\}\>}
    \]
    
    Assume $\Lambda; \Delta \vddash^{\Omega} \<h, env_S, k : e \> : t \triangleright \Lambda'; \Delta'$ where $\Delta = envT_O \cdot envT_S$.
    
    From \rt{WTC} we know 
    \begin{itemize}
        \item $\Lambda \vdash h$ 
        \item $envT_S \vhdash_{k : e} env_S$
        \item $envT_O \vhdash k : e$
        \item $\Lambda;\Delta \vddash^{\Omega} k : e : t \triangleright \Lambda';\Delta'$
        \item $\Lambda \vdash^h \vv{D}$
    \end{itemize}
    
    From \rt{TLab} we know
    \begin{itemize}
        \item $\Omega' = \Omega, k : (\Lambda, \Delta)$
        \item $\Lambda;\Delta \vddash^{\Omega'} e : void \triangleright \Lambda';\Delta'$
    \end{itemize}
    
    There are two cases, either $e$ contains $\sf{continue}\ k$ or it does not. The latter case is trivial, hence we will focus on the first case.
    
    If $\sf{continue}\ k$ does appear in $e$, then at some point during the type derivation, we must have concluded $\Lambda'';\Delta'' \vdash^{\Omega''} \sf{continue}\ k : \sf{void} \triangleright \Lambda''';\Delta'''$.
    
    We never remove elements from $\Omega$, hence we know that $\Omega'' = \Omega''',k : (\Lambda, \Delta)$, as the binding of $k$ was added in \rt{TLab}. 
    
    But from \rt{TCon} we must have that $\Lambda'' = \Lambda$ and $\Delta'' = \Delta$. 
    
    As $\sf{continue}$ expressions can only appear in blocks, and cannot appear on the left-hand side of a sequential expression, we know that changes to environments cannot happen after deriving the type of a $\sf{continue}$ expression. Together with the fact that all branches in if and switch-expressions must have the same resulting environments, we see that the only choice for $\Lambda'''$ and $\Delta'''$ is $\Lambda'$ and $\Delta'$ respectively.
    
    At this point we know $$\Lambda;\Delta \vddash^{\Omega''} \sf{continue k} : \sf{void} \triangleright \Lambda';\Delta',$$
    
    $$ 
    \Lambda;\Delta \vddash^{\Omega} k : e : \sf{void} \triangleright \Lambda';\Delta',
    $$
    and
    $$
        \Lambda;\Delta \vddash^{\Omega'} = e : \sf{void} \triangleright \Lambda';\Delta'
    $$
    
    Weakening and strengthening of $\Omega$ is trivial so replacing $\Omega$ with $\Omega''$ will leave $k: e$ well typed, hence replacing $\sf{continue}\ k$ with $k : e$ in $e$ will leave the resulting environments the same. This means that we can conclude
    
    $$ \Lambda;\Delta \vddash^{\Omega} e\substitute{k : e}{\sf{continue}\ k} : \sf{void} \triangleright \Lambda';\Delta'.$$
    
    $envT_O \vhdash e\substitute{k : e}{\sf{continue}\ k}$ follows from our assumption since $\objects{e\substitute{k : e}{\sf{continue}\ k}} = \objects{k : e}=\emptyset$ as $e$ has not been evaluated yet.
    
    $envT_S \vhdash_{e\substitute{k : e}{\sf{continue}\ k}} env_S$ is concluded with \rt{WTP-Base} which follows from our assumption.
    
    The remaining premises for \rt{WTC} follows directly from our assumptions.
    
    $\Lambda'\lambdaeq{o} \Lambda'$ is satisfied because of reflexivity of $\lambdaeq{o}$.

    \end{subproof}
    \item[\rt{LDeref}]
    \begin{subproof}
    
    \[ \inference{h(o).f=v & \lin(v,h)}{\vdash_{\vv{D}} \< h, (o, s) \cdot env_S, f\> \trans \< h\{\sf{null} / o.f \}, (o,s)\cdot env_S, v \>} \]
    
    let $h' = h\{\sf{null} / o.f \}$
    
    Assume $\Lambda; envT_O  \cdot envT_S \vddash \<h, (o, s) \cdot env_S, f\> : t \triangleright \Lambda' ; \Delta'$
    
    From \rt{LDeref} we know:
    \begin{itemize}
        \item $\lin(v, h)$
        \item $h(o).f = v$
    \end{itemize}
    
    From \rt{WTP-Base} we know $envT_S \vhdash_{f} (o, s) \cdot env_S$

    From \rt{TLinFld} we have:
    \begin{itemize}
        \item $envT_S = envT_S' \cdot (o, S)$
        \item $t = \Lambda(o).f$
        \item $\lin(t)$
        \item $\Lambda; envT_O \cdot envT_S' \cdot (o,S) \vddash f : t \triangleright \Lambda \{ o.f \mapsto \bot\}; envT_O \cdot envT_S' \cdot (o, S)$
    \end{itemize}
    
    Let $\Lambda^N = \Lambda\{o.f \mapsto \bot\}$, $envT_O^N = envT_O, v \mapsto t$ and $envT_S^N = envT_S$.
    
    Show $\Lambda^N; envT_O^N \cdot envT_S^N \vddash \<h\substitute{\sf{null}}{h(o).f}, (o,s)\cdot env_s, v\> : t \triangleright \Lambda^N; \emptyset \cdot envT_S^N$. We need to show \rt{WTC}.
    
    To show $\Lambda^N \vdash h'$, we have to show that the premise is satisfied.
    \begin{itemize}
        \item $\sf{getType}(h'(o).f,h) = \bot = \Lambda^N(o).f$ is the only case, where the heap is changed.
        \item $\dom{\Lambda^N} = \dom{h'}$ does not change, as we do not add or remove objects.
    \end{itemize}
    Hence we can conclude $\Lambda^N \vdash h'$.
    
    We have $envT_S^N \vhpdash_{v} (o, s) \cdot env_S$, since $v$ is an object and from linearity we have that $h\substitute{\sf{null}}{o.f}$ does not change the type on the parameter stack, as $x$ and $o.f$ cannot refer to the same object. Hence it is concluded using \rt{WTP-Base} from assumption.
    
    $envT_O^N \vhpdash v$ is trivially true as $\dom{envT_O^N} = \objects{v} = \{v\}$.
    
    $\Lambda^N; envT_O^N \cdot envT_S^N \vddash v : t \triangleright \Lambda^N; \emptyset \cdot envT_S^N$ can be concluded with \rt{TObj}.
    
    To show $\Lambda^N \vhpdash \vv{D}$ is satisfied, we need to show the premise is satisfied.
    \begin{itemize}
        \item $\dom{\Lambda^N} = \dom{h'}$. We do not remove or add any objects, hence this is trivially true.
        \item $\forall o \in \dom{h'} . \Lambda(o).\sf{fields} = h'(o).\sf{fields} = h'(o).\sf{class}.\sf{fields}_{\vv{D}}$ is trivially true, since we do not remove or add any fields.
        \item $\forall f \in h'(o).\sf{fields} . \Lambda(o).\sf{envT}_\sf{F}(f) = \sf{getType}(h'(o).f, h')$ holds for all unchanged fields. From definition we have $h\substitute{\sf{null}}{o.f}$, we have that $\sf{getType}(h'(o).f, h') = \bot$ and $\Lambda(o).\sf{envT}_{\sf{F}}(f) = \bot$ from assumption. 
    \end{itemize}
    The premise is satisfied hence we have $\Lambda^N \vhpdash \vv{D}$.
    
    The premises of \rt{WTC} is satisfied, hence we have that
    
    $\Lambda^N; envT_O^N \cdot envT_S^N \vddash \<h\substitute{\sf{null}}{h(o).f}, (o,s) \cdot env_S, v\> : t \triangleright \Lambda^N; \emptyset \cdot envT_S^N$
    
    Since we have that $\Lambda^N = \Lambda\{o.f \mapsto \bot\}$, it is clear that $\Lambda^N \lambdaeq{o} \Lambda\{o.f \mapsto \bot\}$.
    
    \end{subproof}
    \item[\rt{UParam}]
    \begin{subproof}
    
    \[\inference{\neg \lin(v,h)}{ \vdash_{\vv{D}} \<h,(o,[x \mapsto v]) \cdot env_S, x \> \trans \<h, (o, [x \mapsto v])\cdot env_S, v\>} \]
    
    Assume $\Lambda; envT_S \cdot envT_S \vdash \<h,(o,[x \mapsto v]) \cdot env_S, x \> : t \triangleright \Lambda' ; envT_O' \cdot envT_S'$, $\Lambda; envT_O \cdot envT_S \vddash x : t \triangleright \Lambda'; envT_O' \cdot envT_S'$ must have been concluded with \rt{TNoLPar}, hence we know:
    \begin{itemize}
        \item $\Lambda' = \Lambda$
        \item $envT_S = envT_S'' \cdot (o, [x \mapsto t])$
        \item $envT_S' = envT_S$
    \end{itemize}
    
    from \rt{WTP-Base} we have $envT_S \vhdash_x (o, [x \mapsto v]) \cdot env_S$
    
    Consider two cases
    case 1: $v$ is an unrestricted object. From \rt{WTP} we have $t = \sf{getType}(v,h)$ hence $v \in \dom{h}$. Let $\Lambda^N=\Lambda$, $envT_O^N = envT_O, v \mapsto t$ and $envT_S^N = envT_S$. 
    
    Case 2: $v$ is a base value. Let $\Lambda^N=\Lambda$, $envT_O^N = envT_O$ and $envT_S^N = envT_S$.
    
    Both cases.
    \begin{itemize}
        \item $\Lambda^N \vddash h$ from hypothesis
        \item $envT_S^N \vhpdash_{v} env_s$ from hypothesis as $v$ must be concluded using \rt{WTP-Base}.
    \end{itemize}
    
    Case $v$ is object. $\dom{envT_O} = \sf{objects}(x) = \emptyset \implies \dom{envT_O^N} = \emptyset \cup \{v\} = \sf{objects}(v)$ and $envT_O^N(v) = t = \sf{getType}(v, h)$.
    Case $v$ is a base value. $\dom{envT_O^N} = \sf{objects}(x) = \sf{objects}(v)$ is trivial.
    
    Both cases can be concluded $\Lambda^N; envT_O^N \cdot envT_S^N \vddash v:t \triangleright \Lambda^N ;\emptyset \cdot envT_S$ and we have $\Lambda^N \vhpdash \vv{D}$ from hypothesis. The premise of \rt{WTC} is satisfied, hence we have that $\Lambda^N ; envT_O^N \cdot envT_S^N \vddash \<h, (o, [x \mapsto v]) \cdot env_S, v\> : t \triangleright \Lambda^N ; \emptyset \cdot envT_S$
    
    Since we have that $\Lambda^N=\Lambda=\Lambda'$, we also have $\Lambda^N \lambdaeq{o} \Lambda'$.
    
    \end{subproof}
    \item[\rt{LParam}]
    \begin{subproof}
    
    \[\inference{\lin(v,h)}{ \vdash_{\vv{D}} \<h,(o,[x \mapsto v]) \cdot env_S, x \> \trans \<h, (o, [x \mapsto \sf{null}])\cdot env_S, v\>}\]
    
    Assume $\Lambda; envT_O \cdot envT_S \vddash \<h, (o, [x \mapsto v])\cdot env_S, x\> : t \triangleright \Lambda; envT_O' \cdot envT_S' \cdot (o, [x \mapsto \bot])$ and $envT_S=envT_S'\cdot(o,[x\mapsto t])$
    
    From \rt{LParam} we know $\lin(v, h) \implies h(o).\sf{class} = C\<t'\> \wedge h(o).\sf{usage} \neq end$

    From \rt{TLinPar} we have
    \begin{itemize}
        \item $\Lambda; envT_O \cdot envT_S' \cdot (o, [x \mapsto t]) \vddash x : t \triangleright \Lambda'; envT_O \cdot envT_S' \cdot (o, [x \mapsto \bot])$
        \item $\lin(t) \implies t=C\<t'\>[\U] \wedge \U \neq \sf{end}$
    \end{itemize}
    
    From \rt{TObj} we have
    \begin{itemize}
        \item $\Lambda' = \Lambda$
    \end{itemize}
    
    From \rt{WTP-Base} we have 
    \begin{itemize}
        \item $envT_S \vhdash_x (o, [x \mapsto v])$
        \item $envT_S' \cdot (o, [x \mapsto t]) \vhdash (o, [x \mapsto v])$
        \item $\sf{getType}(v, h) = t$
        \item $env_S = \varepsilon$
    \end{itemize}
   
    Let $\Lambda^N = \Lambda$, $envT_O^N = envT_O, v \mapsto C\<t'\>[\U]$ and $envT_S^N = envT_S' \cdot (o, [x \mapsto \bot])$
    
    We now show that
    \begin{itemize}
        \item $\Lambda^N \vddash h$ from hypothesis
        \item To show $envT_S^N \vhdash_v (o, [x \mapsto \sf{null}]) \cdot env_S$ we have $v$ is a base type and must be concluded using \rt{WTP-Base} and $\sf{getType}(\sf{null}, h') = \bot$ from the definition of $\sf{getType}$, hence the premise is satisfied.
        \item To show $envT_O^N \vhdash v$ we need to show that the premise is satisfied.
        \begin{itemize}
            \item $\sf{objects}(x) = \emptyset = \dom{envT_O} \implies \dom{envT_O^N} = \llbar v \rrbar = \sf{objects}(v)$ 
            \item $envT_O^N(v)=t=\sf{getType}(v, h)$
            \item From \rt{WTP} $v \in \dom{h}$ and $\sf{getType}(v, h)=t$
        \end{itemize}
        \item $\Lambda^N; envT_O^N \cdot envT_S^N \vddash v : t \triangleright \Lambda^{N}; envT_O \cdot envT_S^N$
        \item $\Lambda^N \vhdash \vv{D}$ from hypothesis.
    \end{itemize}  
    
    We have concluded the premises for \rt{WTC}, hence we conclude $\Lambda^N; envT_O^N \cdot envT_S^N \vddash v : t \triangleright \Lambda^N; envT_O \cdot envT_S' \cdot (o, [x \mapsto \bot])$
    
    Since we have that $\Lambda^N=\Lambda$, we also have $\Lambda^N \lambdaeq{o} \Lambda$.

    \end{subproof}
    
    \item[\rt{SwF}]
    \begin{subproof}
    \[\inference{h(o).f = o' & h(o').\sf{usage} \trans[l_i] \U & l_i \in L}{\begin{gathered}\vddash \<h, (o,s) \cdot env_S, \sf{switch}_{f.m}(l_i)\{l_j : e_j\}_{l_j \in L}\> \trans \\ \hspace{3.5cm}\<h\{o.f \mapsto C\<t'\>[\U]\}, (o,s) \cdot env_S, e_i\>\end{gathered}}\]
    We assume $\Lambda; \Delta \cdot (o'',S) \vddash \<h, (o,s) \cdot env_S, \sf{switch}_{f.m}(l_i)\{l_j : e_j\}_{l_j \in L}\> : t \triangleright \Lambda'; \Delta'$. Let $envT_O \cdot envT_S = \Delta \cdot (o'',S)$ and $h' = h\{o.f \mapsto C\<t'\>[\U]\}$.
    
    From \rt{WTC} we know:
    \begin{itemize}
        \item $\Lambda \vddash h$
        \item $envT_S \vhdash_{\sf{switch}_{f.m}(l_i)\{l_j : e_j\}_{l_j \in L}} (o,s) \cdot env_S$
        \item $envT_O \vhdash \sf{switch}_{f.m}(l_i)\{l_j : e_j\}_{l_j \in L}$
        \item $\Lambda; \Delta \cdot (o'', S) \vddash \sf{switch}_{f.m}(l_i)\{l_j : e_j\}_{l_j \in L} : t \triangleright \Lambda'; \Delta'$ 
        \item $\Lambda \vddash \vv{D}$
    \end{itemize}
    
    From \rt{WTP-Sw} we know $envT_S \vhdash_{l_i} (o, s) \cdot env_s$
    
    From \rt{WTP-Base} we know:
    \begin{itemize}
        \item $envT_S \cdot (o'', [x \mapsto t'']) \vhdash_{l_i} (o,[x \mapsto v])$
        \item $o = o''$
    \end{itemize}
    
    From \rt{TSwF} we know:
    \begin{itemize}
        \item $\Lambda; \Delta \cdot (o,S) \vddash l_i : L \triangleright \Lambda^{(3)}, o.f \mapsto C\<t'\>[(\<l_i : u_i\>_{l_i \in L})^{\vv{E}}]; envT_O^{(3)} \cdot envT_S^{(3)} \cdot (o,S'')$
        \item $\forall l_i \in L . \ \Lambda^{(3)}, o.f \mapsto C\<t'\>[u_i^{\vv{E}}]; envT_O^{(3)} \cdot envT_S^{(3)} \cdot (o, S'') \vddash e_i : t \triangleright \Lambda''; \Delta'' \cdot (o,S')$
    \end{itemize}

    We know from \rt{WTC} and \rt{TSwF}:
    \begin{itemize}
        \item $\Lambda'' = \Lambda'$
        \item $\Lambda^{(3)} = \Lambda$
        \item $\Delta'' \cdot (o,S') = \Delta'$
    \end{itemize}
    
    Now let $\Lambda^N = \Lambda^{(3)}, o.f \mapsto C\<t'\>[u_i^{\vv{E}}]$ and $\Delta^N = envT_O^N \cdot envT_S^N = envT_O^{(3)} \cdot envT_S^{(3)} \cdot (o,S'')$. We now have to show that $\Lambda^N; \Delta^N \vddash \<h\{o.f \mapsto C\<t'\>[\U]\}, (o,s) \cdot envs, e_i\> : t \triangleright \Lambda^{N'}; \Delta^{N'}$ where $\Lambda' \lambdaeq{o} \Lambda^{N'}$ by showing that the premises of \rt{WTC} are satisfied.
    
    We can see that the only changes we have made to the heap $h$ and field type environment $\Lambda$ is updating the type state of field $f$, we therefore have that $\Lambda^N \vddash h\{o.f \mapsto C\<t'\>[\U]\}$ is satisfied. We know that we have not made any other changes to the environment because typing $l_i : L$ using \rt{TLit} does not change the environment. 
    Similarly, we have that $envT_O^N = envT_O^{(3)} = envT_O$ and $envT_S^N = envT_S^{(3)} \cdot (o,S'') = envT_S$. We now have that $envT_S^N \vhdash_{e_i} (o,s) \cdot env_S$ is satisfied from \rt{WTP} since the parameter stack environment has not changed after typing a label and $\sf{returns}(e_i) = 0$ because $e_i$ is a well-formed expression, hence it can only be concluded with \rt{WTP-Base} and we know the premise is satisfied. $envT_O^N \vhpdash e_i$ is satisfied from Lemma \ref{lemma:objectagreement} and our assumption. We know that the object type environment has not changed after typing a label and from linearity we know that the updated object in the heap does not occur in both $envT_O$ and $\Lambda$ hence we do not have to consider the newly updated object in $envT_O$ and from Lemma \ref{lemma:objectagreement} $\sf{objects}(\switchf{l_i}) = \sf{objects}(e_i)$ hence from our assumption the premises of \rt{WTE} are satisfied.
    
    
    $\Lambda^N; envT_O^N \cdot envT_S^N \vddash e_i : t \triangleright \Lambda^{N'}; \Delta^{N'}$ is satisfied from our assumption, specifically the second premise of \rt{TswF} since we know the environments are the same. 
    
    We now have to show $\Lambda^N \vhpdash \vv{D}$, from \rt{WTD} we have the premise $\forall o \in \dom{h'}$, $\Lambda^N(o).\sf{class} = h'(o).\sf{class}$ is satisfied since the field we have updated in $h'$ is also updated accordingly in $\Lambda^N$. From \rt{TClass} and \rt{TCCh} we have that $\Theta; \Lambda(h'(o).f).\sf{envT}_{\sf{F}} \vddash C\<t'\>[(\<l_i : u_i\>_{l_i \in L})^{\vv{E}}] \triangleright envT_F'$ and we reached $C\<t'\>[u_i^{\vv{E}}]$ by following the usage then we have already checked that we can reach a terminated field type environment with \rt{TCCh} and all premises of \rt{WTD} are therefore satisfied. We know from the second premise of \rt{TswF} that $\Lambda^{N'} = \Lambda'$ hence $\Lambda' \lambdaeq{o} \Lambda^{N'}$ and $\Delta^{N'} = \Delta'$. 
    
    The premises of \rt{WTC} are all satisfied and we can conclude $\Lambda^N; \Delta^N \vddash \<h\{o.f \mapsto C\<t'\>[\U]\}, (o,s) \cdot envs, e_i\> : t \triangleright \Lambda'; \Delta'$.
    
    \end{subproof}
    
    \item[\rt{SwP}]
    \begin{subproof}
    \[\inference{h(o').\sf{usage} \trans[l_i] \U & l_i \in L}{ \begin{gathered} \vddash \<h, (o, [x \mapsto o']) \cdot env_S, \sf{switch}_{x.m}(l_i)\{l_j : e_j\}_{l_j \in L}\> \trans \\ \hspace{3.5cm} \<h\{o' \mapsto C\<t'\>[\U]\}, (o,[x \mapsto o']) \cdot env_S,e_i\>\end{gathered}}\]
    We assume $\Lambda; \Delta \cdot (o'',S) \vddash \<h, (o,s) \cdot env_S, \sf{switch}_{x.m}(l_i)\{l_j : e_j\}_{l_j \in L}\> : t \triangleright \Lambda'; \Delta'$. let $envT_O \cdot envT_S = \Delta \cdot (o'',S)$.
    
    From \rt{WTC} we know:
    \begin{itemize}
        \item $\Lambda \vddash h$
        \item $envT_S \vhdash_{\sf{switch}_{x.m}(l_i)\{l_j : e_j\}_{l_j \in L}} (o,s) \cdot env_S$
        \item $envT_O \vhdash \sf{switch}_{x.m}(l_i)\{l_j : e_j\}_{l_j \in L}$
        \item $\Lambda; envT_O \cdot envT_S \vddash \sf{switch}_{x.m}(l_i)\{l_j : e_j\}_{l_j \in L} : t \triangleright \Lambda'; \Delta'$  
        \item $\Lambda \vddash \vv{D}$
    \end{itemize}
    
    From \rt{WTP-Sw} we know $envT_S \vhdash_{l_i} (o, [x \mapsto o']) \cdot env_s$
    
    From \rt{WTP-Base} we know:
    \begin{itemize}
        \item $envT_S' \cdot (o'', [x \mapsto t'']) \vhdash_{l_i} (o, [x \mapsto o'])$
        \item $\sf{getType}(v, h) = t''$
        \item $o = o''$
        \item $env_S = \varepsilon$
    \end{itemize}
    
    From \rt{TswP} we know:
    \begin{itemize}
        \item $\Lambda;\Delta \cdot (o'', S) \vddash l_i : L \triangleright \Lambda^{(3)};envT_O^{(3)} \cdot envT_S^{(3)} \cdot(o, [x \mapsto C\<t'\>[(\<l_i : u_i\>_{l_i \in L})^{\vv{E}}]])$
        \item $\forall l_i \in L . \ \Lambda^{(3)};envT_O^{(3)} \cdot envT_S^{(3)} \cdot(o, [x\mapsto C\<t'\>[u_i]]) \vddash e_i : t \triangleright \Lambda'';\Delta'' \cdot (o, S')$
    \end{itemize}
    
    We know from \rt{WTC} and \rt{TswP}:
    \begin{itemize}
        \item $\Lambda'' = \Lambda'$
        \item $\Delta'' \cdot (o,S') = \Delta'$
    \end{itemize}
    
    From \rt{TLit} we know:
    \begin{itemize}
        \item $envT_O^{(3)} \cdot envT_S^{(3)} \cdot(o, [x \mapsto C\<t'\>[(\<l_i : u_i\>_{l_i \in L})^{\vv{E}}]]) = envT_O \cdot envT_S$
    \end{itemize}

    Now let $\Lambda^N = \Lambda^{(3)}$ and $\Delta^N = envT_O^N \cdot envT_S^N \cdot (o, [x \mapsto C\<t'\>[u_i^{\vv{E}}]]) = envT_O^{(3)} \cdot envT_S^{(3)} \cdot (o, [x \mapsto C\<t'\>[u_i^{\vv{E}}]])$. We now have to show that $\Lambda^N;\Delta^N \vddash \<h\{o' \mapsto C\<t'\>[\U]\}, (o, [x \mapsto o']) \cdot envs, e_i\> : t \triangleright \Lambda'; \Delta'$ by showing that the premises of \rt{WTC} are satisfied.
    
    We have typed $l_i : L$ using \rt{TLit} with this we know that the only change made to $h$ is an update to the typestate of object $o'$, however, we still have $\text{dom}(h) = \text{dom}(\Lambda^N)$ and the fields have not changed, hence $\Lambda^N \vddash h\{o' \mapsto C\<t'\>[\U]\}$ is satisfied since the premises of \rt{WTH} are satisfied. We also have that $envT_O^N \vhpdash e_i$ is satisfied from our assumption, linearity, and Lemma \ref{lemma:objectagreement}. We have $envT_S^N \cdot (o, [x \mapsto C\<t'\>[u_i^{\vv{E}}]]) \vhpdash_{e_i} (o,[x \mapsto o']) \cdot env_S$ is satisfied from \rt{WTP-Base} as a result of $\sf{returns}(e_i) = 0$ because $e_i$ is a well-formed expression, and only the topmost element on both parameter stack and $envT_S^N$ is changed where this change is also reflected in the heap $h'$, such that $\sf{getType}(o', h') = C\<t'\>[u_i^{\vv{E}}]$ hence the premises of \rt{WTP-Base} are satisfied. 
    
    $\Lambda^N; \Delta^N \cdot (o,[x \mapsto C\<t'\>[u_i^{\vv{E}}]]) \vddash e_i : t \triangleright \Lambda^{N'};\Delta^{N'} \cdot (o, S')$ is satisfied from our assumption, specifically, the second premise of \rt{TswP} since the environments are the same. Finally, we have $\Lambda^{N} \vhpdash \vv{D}$, from the premise we have that $\forall o \in \dom{h'}$, $\Lambda^N(o).\sf{class} = h'(o).\sf{class}$ is satisfied since we do not update any fields. From \rt{TClass} and \rt{TCCh} we have that $\Theta; \Lambda(o).\sf{envT}_{\sf{F}} \vddash C\<t'\>[(\<l_i : u_i\>_{l_i \in L})^{\vv{E}}] \triangleright envT_F'$ and from reachability we have that the remaining premise of \rt{WTD} is satisfied. We can therefore conclude $\Lambda^N \vhpdash \vv{D}$.
    
    We know from the second premise of \rt{TswP} that $\Lambda^{N'} = \Lambda'$ hence $\Lambda' \lambdaeq{o} \Lambda^{N'}$ and $\Delta^{N'} = \Delta'$. 
    
    The premises of \rt{WTC} are all satisfied and we can conclude 
    
    $\Lambda^N;\Delta^N \vddash \<h\{o' \mapsto C\<t'\>[\U]\}, (o, [x \mapsto o']) \cdot envs, e_i\> : t \triangleright \Lambda'; \Delta'$.
    
    \end{subproof}
    
    \item[\rt{Upd}]
    \begin{subproof}
    \[
    \inference{h(o).f = v' & \neg \lin(v', h)}{\vddash \<h, (o,s) \cdot env_S, f=v \> \trans\<h\{v/o.f\}, (o,s)\cdot env_S, \sf{unit}\>} \]
    
    Assume that $\Lambda; \Delta \vddash \<h, (o, s), f = v\> : \sf{void} \triangleright \Lambda'; \Delta'$ where $\Delta = envT_O \cdot envT_S$
    
    From \rt{TFld} we know:
    \begin{itemize}
        \item $\Lambda; envT_O \cdot envT_S' \cdot (o', S) \vddash f = v : \sf{void} \triangleright \Lambda\{o.f \mapsto t' \}; envT_O' \cdot envT_S'' \cdot (o', S') $
        \item $\Lambda; envT_O \cdot envT_S' \cdot (o' ,S) \vddash v :t' \triangleright \Lambda'', o.f \mapsto t; envT_O' \cdot envT_S'' \cdot (o', S')$
        \item $\neg\lin(t)$
        \item $\Lambda' = \Lambda\{o.f \mapsto t'\}$
        \item $\Delta' = envT_O' \cdot envT_S'' \cdot(o', S')$
    \end{itemize}
    
    From \rt{WTP-Base} we know: $o = o'$
    
    We have two cases, one where $v$ is an object and one where $v$ is a base value. 
    
    Case 1 $v$ is an object. From \rt{TObj} we know:
    \begin{itemize}
        \item $\Lambda'', o.f \mapsto t = \Lambda$
        \item $envT_O = envT_O', v \mapsto t'$
        \item $envT_S'' = envT_S'$
        \item $S=S'$
    \end{itemize}
    
    Case 2 $v$ is not an object. From \rt{TLit}, \rt{TBool}, \rt{TVoid} or \rt{TBot} we know
    \begin{itemize}
        \item $\Lambda'', o.f \mapsto t = \Lambda$
        \item $envT_O' = envT_O$
        \item $envT_S'' = envT_S'$
        \item $S = S'$
    \end{itemize} 
    
    We know $h' = h\substitute{v}{o.f}$
     
    Let $\Lambda^N = \Lambda\{o.f \mapsto t'\}$, $envT_O^N = envT_O'$ and $envT_S^N = envT_S$
    
    First we show $\Lambda^N \vddash h'$
    \begin{itemize}
        \item $\sf{getType}(h'(o).f, h') = t'=  \Lambda^N(o).f$, is true, since we update the field in $\Lambda^N$ and $h'$.
        \item $\dom{\Lambda^N} = \dom{h'}$ is trivially true as no objects are added.
    \end{itemize}
    
    $envT_S^N \vhpdash_{\sf{unit}} env_S$ is true from the premise of the assumption.
    
    $envT_O^N \vhpdash \sf{unit}$ is trivially true because $\dom{envT_O^N} = \dom{\sf{unit}} = \emptyset$ 
    
    \rt{TVoid} is trivially true for all environments, hence we can conclude. $\Lambda^N; envT_O^N \cdot envT_S^N \vdash \sf{unit} \triangleright \Lambda\{o.f \mapsto t'\}; envT_O' \cdot envT_S$
    
    To show $\Lambda^N \vhpdash \vv{D}$ we have to show the premise
    \begin{itemize}
        \item $\dom{\Lambda^N} = \dom{h'}$ is trivially true.
        \item $\Lambda^N(o).\sf{class} = h'(o).\sf{class}$ is trivially true.
        \item $\emptyset; \Lambda^N(o).\sf{envT}_\sf{F} \vddash \sf{getType}(o, h') \triangleright \Gamma'$ we have from reachability. 
    \end{itemize}
     
     We have shown the premise of \rt{WTC} and therefore we have $\Lambda^N; envT_O^N \cdot envT_S^N \vdash \<h\{v/o.f\}, (o,s)\cdot env_S, \sf{unit}\> : \sf{void} \triangleright \Lambda\{o.f \mapsto t'\}; envT_O' \cdot envT_S$
     
    We have that $\Lambda\{o.f \mapsto t'\} \lambdaeq{o} \Lambda\{o.f \mapsto t'\}$ is trivially true.
     
    \end{subproof}
    \item[\rt{IfC}]
    
    \begin{subproof}
    Assume $\<h, env_S, \sif{e}{e_1}{e_2}\> \trans \<h', env_S', \sif{e'}{e_1}{e_2}\>$ and $\Lambda;\Delta \vddash \<h, env_S, \sif{e}{e_1}{e_2}\> : t \triangleright \Lambda' ; \Delta'$ where $\Delta=envT_O \cdot envT_S$. Then from \rt{WTC} we have:
    \begin{itemize}
        \item $\Lambda \vddash h$
        \item $envT_S \vhdash_{\sif{e}{e_1}{e_2}} env_S$
        \item $envT_O \vhdash \sif{e}{e_1}{e_2}$
        \item $\Lambda;\Delta \vddash \sif{e}{e_1}{e_2} : t \triangleright \Lambda';\Delta'$
        \item $\Lambda\vhdash \vv{D}$
    \end{itemize}
    
    From \rt{TIf} we have:
    \begin{itemize}
        \item $\Lambda;\Delta \vddash e : \sf{Bool} \triangleright \Lambda'';\Delta''$
        \item $\Lambda'';\Delta'' \vddash e_i : t \triangleright \Lambda';\Delta'$
    \end{itemize}
    
    From \rt{IfC} we have that $\<h, env_S, e\> \trans \<h', env_S', e'\>$, and from Lemma \ref{lemma:subexpression_welltyped} we know that $\Lambda;\Delta \vddash \<h, env_S, e\>$. From our induction hypothesis, we have that $\exists \Lambda^{(3)}, \Delta^{(3)}.\ \Lambda^{(3)};\Delta^{(3)} \vddash \<h', env_S', e'\> : \sf{Bool} \triangleright \Lambda^S;\Delta''$ where $\Lambda'' \lambdaeq{o} \Lambda^S$ and $\Delta^{(3)} = envT_O^{(3)} \cdot envT_S^{(3)}$ 
    
    We can now show that $\Lambda^{(3)};\Delta^{(3)} \vddash \<h', env_S', \sif{e'}{e_1}{e_2}\> : t \triangleright \Lambda^{S'}; \Delta'$, by showing that all premises of \rt{WTC} holds.
    
    \begin{itemize}
        \item $\Lambda^{(3)}\vddash h'$ follows from IH
        \item $envT_S^{(3)} \vdash^{h'}_{\sif{e'}{e_1}{e_2}} env_S$ is satisfied since we have $envT_S^{(3)} \vdash^{h'}_{e'} env_S$ which is the premise of \rt{WTP-If}
        \item $envT_O^{(3)} \vhdash \sif{e'}{e_1}{e_2}$ since $\objects{\sif{e'}{e_1}{e_2}}=\objects{e'}$ because of well-formedness
        \item $\Lambda^{(3)} \vhpdash \vv{D}$ follows from IH.
    \end{itemize}
    
    Since we know $\Lambda^{(3)};\Delta^{(3)}\vddash e' : \sf{Bool} \triangleright \Lambda^S;\Delta''$ from Lemma \ref{lemma:consistency} and that $\Lambda''\lambdaeq{o} \Lambda^S$, we can use weakening to conclude $\Lambda^S;\Delta''\vddash e_i : t \triangleright \Lambda^{S'};\Delta'$ where $\Lambda' \lambdaeq{o} \Lambda^{S'}$.
    
    These are the two premises of $\rt{TIf}$, hence we can conclude $\Lambda^{(3)};\Delta^{(3)} \vddash \sif{e'}{e_1}{e_2} : t \triangleright \Lambda^{S'};\Delta'$.
    
    We have shown all premises of \rt{WTC}, hence we can conclude
    
    $\Lambda^{3};\Delta^{3}\vddash \<h', env_S', \sif{e'}{e_1}{e_2}\> : t \triangleright \Lambda^{S'};\Delta'$.
    \end{subproof}
    
    \item[\rt{RetC}]
    
    \begin{subproof}
    Assume $\<h, env_S \cdot (o, s), \return{e}\> \trans \<h', env_S' \cdot (o, s), \return{e'}\> $ and that $\Lambda;\Delta\vddash \<h, env_S \cdot (o, s), \return{e}\> : t \triangleright \Lambda'; \Delta'$. Let $\Delta=envT_O \cdot envT_S$.
    
    From \rt{WTC} we have:
    \begin{itemize}
        \item $envT_S \vdash^h_{\return{e}} env_S \cdot (o, s)$
        \item $envT_S = envT_S' \cdot (o, S)$ from \rt{WTP-Ret}
        \item $\Lambda;\Delta \vddash \return{e} : t \triangleright \Lambda';\Delta'$
    \end{itemize}
    
    From \rt{TRet} we have that $\Lambda;envT_O\cdot envT_S' \vddash e : t \triangleright \Lambda';\Delta''$, where $\Delta''=\Delta^{(4)} \cdot (o', [x \mapsto t'])$, $\sf{terminated}(t')$ and that $\Delta' = \Delta^{(4)}\cdot(o, S)$.
    
    From Lemma 36 we know that $\Lambda;envT_O \cdot envT_S' \vddash \<h, env_S, e\> : t \triangleright \Lambda';\Delta''$. 
    
    Then from our induction hypothesis we have $\exists \Lambda^{(3)}, \Delta^{(3)}.\ \Lambda^{(3)};\Delta^{(3)} \vddash \<h', env_S', e'\> : t \triangleright \Lambda^S;\Delta''$ where $\Lambda'\lambdaeq{o'}\Lambda^S$. We can also assume $\Lambda'\lambdaeq{o}\Lambda^S$, since because of linearity we cannot access any fields in $o$ inside $e$, hence $e$ will also be well-typed if $\Lambda'(o)=\Lambda^S(o)$.
    
    We now show that $\Lambda^{(3)};\Delta^{(3)} \cdot (o, S) \vddash \<h', env_S' \cdot (o, s), \return{e'}\> : t \triangleright \Lambda^S;\Delta'$, by showing the premises of \rt{WTC} are satisfied.
    
    \begin{itemize}
        \item $\Lambda^{(3)} \vddash h'$ is satisfied from our assumptions
        \item $envT_S^{(3)}\cdot (o, S) \vdash^{h'}_{\return{e'}} env_S' \cdot (o, s)$ is satisfied since we know that $envT_S^{(3)} \vdash^{h}_{e'} env_S'$, we also know from linearity that the only binding of $o$ is in the parameters, that are not present when executing $e'$, hence we must have $h'(o)=h(o)$. But if this is the case, then from $envT_S \vhdash_{\return{e}} env_S \cdot (o, s)$ we know that $S$ and $s$ matches types in $h'$.
        \item $envT_O^{(3)} \vhpdash \return{e'}$ is trivial since $\objects{\return{e'}}=\objects{e'}$.
        \item $\Lambda^{(3)} \vddash \vv{D}$ is satisfied by our induction hypothesis.
    \end{itemize}
    
    We know from \rt{WTC} that $\Lambda^{(3)};\Delta^{(3)} \vddash e' : t \triangleright \Lambda^S;\Delta''$. We can then use \rt{TRet} to conclude $\Lambda^{(3)};\Delta^{(3)} \cdot (o, S) \vddash \return{e'} : t \triangleright \Lambda^S; \Delta'$.

    Hence we can conclude that $\Lambda^{(3)};\Delta^{(3)}\vddash \<h', env_S' \cdot (o, s), \return{e'}\> : t \triangleright \Lambda^S;\Delta'$.
    
    \end{subproof}
    
    \item[\rt{FldC}]
    
    \begin{subproof}
    
    Assume $\<h, env_S, f = e\> \trans \<h', env_S', f = e'\>$ and $\Lambda;\Delta \vddash \<h, env_S, f = e\> : \sf{void} \triangleright \Lambda';\Delta'$
    
    From \rt{FldC} we have that $\<h, env_S, e\> \trans \<h', env_S', e'\>$.
    
    From \rt{WTC} we know:
    \begin{itemize}
        \item $\Lambda \vddash h$
        \item $envT_O \vhdash f = e$
        \item $envT_S \vdash^h_{f = e} env_S$
        \item $\Lambda;\Delta \vddash f = e : \sf{void} \triangleright \Lambda';\Delta'$
        \item $\Lambda \vhdash \vv{D}$
    \end{itemize}
    
    From \rt{TFld} we have:
    \begin{itemize}
        \item $\Lambda;\Delta \vddash e : t \triangleright \Lambda'', o.f \mapsto t';\Delta'$
        \item $\Delta = \Delta'' \cdot (o, S)$
        \item $C\<t'''\> = \Lambda(o).\sf{class}$
        \item $\sf{agree}(C\<t'''\>.\sf{fields}(f), t)$
        \item $\neg\lin(t')$
    \end{itemize}
    
    From Lemma \ref{lemma:subexpression_welltyped} we know that $\Lambda;\Delta\vddash \<h, env_S, e\> : t \triangleright \Lambda'', o.f \mapsto t';\Delta'$.
    
    From our induction hypothesis we have
    \[\exists \Lambda^{(3)},\Delta^{(3)}.\ \Lambda^{(3)};\Delta^{(3)}\vddash \<h', env_S', e'\> : t \triangleright \Lambda^S; \Delta'\]
    where $\Lambda'', o.f \mapsto t' \lambdaeq{o} \Lambda^S$
    
    From \rt{WTC} we now know:
    \begin{itemize}
        \item $\Lambda^{(3)} \vddash h'$
        \item $envT_O^{(3)} \vhpdash e'$
        \item $envT_S^{(3)} \vdash^{h'}_{e'} env_S'$
        \item $\Lambda^{(3)};\Delta^{(3)} \vddash e' : t \triangleright \Lambda^S;\Delta'$
        \item $\Lambda^{(3)} \vhpdash \vv{D}$
    \end{itemize}
    
    We then show $\Lambda^{(3)};\Delta^{(3)}\vddash \<h', env_S', f = e'\> : \sf{void} \triangleright \Lambda^S\{o.f \mapsto t\};\Delta'$ by showing that all premises of \rt{WTC} are satisfied. 
    
    \begin{itemize}
        \item $\Lambda^{(3)} \vddash h'$ follows from assumption
        \item $envT_O^{(3)} \vhpdash f = e'$ follows from assumption since $\objects{f = e'} = \objects{e'}$
        \item $envT_S^{(3)} \vdash^{h'}_{f=e'} env_S'$ follows from \rt{WTP-Fld} since the premise $envT_S^{(3)} \vdash^{h'}_{e'} env_S'$ follows from our assumption
    \end{itemize}
    
    We now show that $\Lambda^{(3)};\Delta^{(3)} \vddash f = e' : \sf{void} \triangleright \Lambda^S\{o.f \mapsto t\};\Delta'$.
    
    Since we have that $\Lambda'', o.f \mapsto t' \lambdaeq{o} \Lambda^S$ we know that $\Lambda^S = \Lambda^{S'}, o.f \mapsto t'$. We also know that $\neg\lin(t')$. We know that $\sf{agree}(C\<t'''\>, t)$, so we only need to argue for $\Lambda^{(3)}(o).\sf{class}=C\<t'''\>$. This must be the case however, since we know that the class binding cannot be updated in any field typing environment, and we know that since $\Lambda(o).\sf{class}=C\<t'''\>$, then $\Lambda'(o).\sf{class}=C\<t'''\>$. But since $\Lambda'\lambdaeq{o} \Lambda^S\{o.f \mapsto t\}$, we must also have that $\Lambda^{(3)}(o).\sf{class}=C\<t'''\>$ (since we cannot update the class binding in $e'$).
    
    We have shown all premises of \rt{TFld}, hence we can conclude that $\Lambda^{(3)};\Delta^{(3)} \vddash \<h', env_S', f = e'\> : \sf{void} \triangleright \Lambda^S\{o.f \mapsto t\};\Delta'$.
    
    \end{subproof}
    
    \item[\rt{CallFC}]
    
    \begin{subproof}
        $\<h, env_S, f.m(e)\> \trans \<h', env_S', f.m(e')\>$ and $\Lambda;\Delta \vddash \<h, env_S, f.m(e)\> : t' \triangleright \Lambda'; \Delta'$ where $\Delta = envT_O \cdot envT_S$. 
        
        From \rt{MthdC} we have $\<h, env_S, e\> \trans \<h', env_S', e'\>$
        
        From \rt{WTC}
        \begin{itemize}
            \item $\Lambda \vddash h$
            \item $envT_S \vhdash_{f.m(e)} env_S$
            \item $envT_O \vhdash f.m(e)$
            \item $\Lambda;\Delta \vddash f.m(e) : t' \triangleright \Lambda';\Delta'$
            \item $\Lambda \vhdash \vv{D}$
        \end{itemize}
        
        From \rt{TCallF} we have:
        \begin{itemize}
            \item $\Lambda;\Delta''\cdot(o, S) \vddash e : t \triangleright \Lambda''\{o.f \mapsto C\<t''\>[\U]\};\Delta'''\cdot(o, S')$
            \item $\U \trans[m] \W$
            \item $t'\ m(t\ x)\{e''\}\in C\<t''\>.\sf{methods}_{\vv{D}}$
            \item $\Delta' = \Delta'''\cdot(o, S')$
            \item $\Lambda' = \Lambda''\{o.f \mapsto C\<t''\>[\W]\}$
        \end{itemize}
        
        From Lemma \ref{lemma:subexpression_welltyped} we have $\Lambda;\Delta \vddash \<h, env_S, e\> : t \triangleright \Lambda''\{o.f \mapsto C\<t''\>[\U]\};\Delta'$
        
        From our induction hypothesis we have 
        \[
        \exists \Lambda^{(3)}, \Delta^{(3)}. \ \Lambda^{(3)};\Delta^{(3)} \vddash \<h', env_S', e'\> : t \triangleright \Lambda^S;\Delta'
        \]
        where $\Lambda''\{o.f \mapsto C\<t''\>[\U]\} \lambdaeq{o} \Lambda^S$.
        
        From \rt{WTC} we then have:
        \begin{itemize}
            \item $\Lambda^{(3)} \vddash h'$
            \item $envT_S^{(3)}\vhpdash_{e'} env_S'$
            \item $envT_O^{(3)} \vhpdash e'$
            \item $\Lambda^{(3)};\Delta^{(3)} \vddash e' : t \triangleright \Lambda^S;\Delta'$
            \item $\Lambda^{(3)} \vhpdash \vv{D}$
        \end{itemize}
        
        We can then show $\Lambda^{(3)};\Delta^{(3)} \vddash \<h', env_S', f.m(e')\> : t \triangleright \Lambda^{S'}\{o.f \mapsto C\<t''\>[\W]\};\Delta'$.
        
        \begin{itemize}
            \item $\Lambda^{(3)} \vddash h'$
            \item $envT_S^{(3)} \vhpdash_{f.m(e')} env_S'$ follows from \rt{WTP-Mthd}
            \item $envT_O^{(3)} \vhpdash f.m(e')$ is trivial since $\objects{f.m(e')}=\objects{e'}$.
            \item $\Lambda^{(3)} \vhpdash \vv{D}$
        \end{itemize}
        It only remains to show $f.m(e')$ is well typed. Since $\Lambda''\{o.f \mapsto C\<t''\>[\U]\} \leq \Lambda^S$, we know that $\Lambda^S = \Lambda^{S'}\{o.f \mapsto C\<t''\>[\U]\}$. 
        
        Then we know $\Lambda^{(3)};\Delta^{(3)}\vddash e' : t \triangleright \Lambda^{S'}\{o.f \mapsto C\<t''\>[\U]\};\Delta'$. From our assumptions we have $\U \trans[m] \W$ and $t'\ m(t\ x)\{e''\}\in C\<t''\>.\sf{methods}$. Hence all premises of \rt{TCallF} are satisfied, and we have $\Lambda^{(3)};\Delta^{(3)} \vddash f.m(e') : t' \triangleright \Lambda^{S'}\{o.f \mapsto C\<t''\>[\W]\};\Delta'$.
        
        We have that $\Lambda' \lambdaeq{o} \Lambda^{S'}\{o.f \mapsto C\<t''\>[\W]\}$ since we have that $\Lambda''\lambdaeq{o} \Lambda^S$. The only change between $\Lambda''$ and $\Lambda'$ is the binding of $o.f$, which is reflected in $\Lambda^{S'}\{o.f \mapsto C\<t''\>[\W]\}$.
        
        Hence we can conclude $\Lambda^{(3)};\Delta^{(3)}\vddash \<h', env_S', f.m(e')\> : t \triangleright  \Lambda^{S'}\{o.f \mapsto C\<t''\>[\W]\};\Delta'$.
        
    \end{subproof}
    
    \item[\rt{CallPC}]
    \begin{subproof}
        Assume $\<h, env_S, x'.m(e)\> \trans \<h', env_S', x'.m(e')\>$ and $\Lambda; \Delta \vddash \<h, env_S, x'.m(e)\> : t' \triangleright \Lambda'; \Delta'$ where $\Delta = envT_O \cdot envT_S \cdot (o,S)$.
        
        From \rt{MthdC} we have:
        \begin{itemize}
            \item $\<h, env_S, e\> \trans \<h', env_S', e'\>$
        \end{itemize}
        
        From \rt{WTC} we have:
        \begin{itemize}
            \item $\Lambda \vddash h$
            \item $envT_S \vhdash_{x'.m(e)} env_S$
            \item $envT_O \vhdash x'.m(e)$
            \item $\Lambda; \Delta \vddash x'.m(e) : t' \triangleright \Lambda'; \Delta'$
            \item $\Lambda \vhdash \vv{D}$
        \end{itemize}
        
        From \rt{TCallP} we have:
        \begin{itemize}
            \item $\Lambda;envT_O \cdot envT_S \cdot (o, S) \vdash_{\vv{D}} e:t\triangleright \Lambda';\Delta'' \cdot (o, [x\mapsto C\<t''\>[\U]])$
            \item $\mathcal{U}\trans[m] \W$
            \item $t'\ m(t\ x)\{ e'' \}\in C\<t''\>.\sf{methods}_{\vv{D}}$
            \item $\Delta' = \Delta'' \cdot (o, [x \mapsto C\<t''\>[\W]])$
        \end{itemize}
        
        From Lemma \ref{lemma:subexpression_welltyped} we know that $\Lambda; \Delta \vddash \<h, env_S, e\> : t \triangleright \Lambda'; \Delta'' \cdot (o, [x \mapsto C\<t''\>[\U]])$.
        
        From our induction hypothesis we have $$\exists \Lambda^{(3)}, \Delta^{(3)} \ . \ \Lambda^{(3)}; \Delta^{(3)} \vddash \<h', env_S', e'\> : t \triangleright \Lambda^{S}; \Delta'' \cdot (o, [x \mapsto C\<t''\>[\U]])$$
        
        where $\Lambda' \lambdaeq{o} \Lambda^{S}$.
        
        From \rt{WTC} we then have:
        \begin{itemize}
            \item $\Lambda^{(3)} \vddash h'$
            \item $envT_S^{(3)} \vhpdash_{e'} env_S'$
            \item $envT_O^{(3)} \vhpdash e'$
            \item $\Lambda^{(3)}; \Delta^{(3)} \vddash e' : t \triangleright \Lambda^S; \Delta'' \cdot (o, [x \mapsto C\<t''\>[\U]])$
            \item $\Lambda^{(3)} \vhpdash \vv{D}$
        \end{itemize}
        
        We now have to show $\Lambda^{(3)}; \Delta^{(3)} \vddash \<h', env_S', x'.m(e')\> : t' \triangleright \Lambda^{S}; \Delta'$ by showing that the premises of \rt{WTC} are satisfied:
        \begin{itemize}
            \item $\Lambda^{(3)} \vddash h'$ from assumption
            \item $envT_S^{(3)} \vhpdash_{x'.m(e')} env_S$ from assumption and \rt{WTP-Mthd}
            \item $envT_O^{(3)} \vhpdash x'.m(e')$ from assumption since $\sf{objects}(x'.m(e')) = \sf{objects}(e')$
            \item $\Lambda^{(3)} \vhpdash \vv{D}$ from assumption
        \end{itemize}
        
        We now have to show that $x'.m(e')$ is well-typed. We know the following: $\Lambda^{(3)}; \Delta^{(3)} \vddash e' : t \triangleright \Lambda^S; \Delta'' \cdot (o, [x \mapsto C\<t''\>[\U]])$, $\U \trans[m] \W$, and $t' \ m(t \ x)\{e''\} \in C\<t''\>.\sf{methods}_{\vv{D}}$. Hence all premises of \rt{TCallP} are satisfied. We know $\Lambda' \lambdaeq{o} \Lambda^{S}$ and we can now conclude $\Lambda^{(3)}; \Delta^{(3)} \vddash \<h', env_S', x'.m(e')\> : t' \triangleright \Lambda^{S}; \Delta'$.
        
    \end{subproof}
    
    \item[\rt{SeqC}]
    
    \begin{subproof}
        Assume $\<h, env_S, e;e''\> \trans \<h', env_S', e';e''\>$ and $\Lambda;\Delta \vddash \<h', env_S', e;e''\> : t' \triangleright \Lambda'; \Delta'$. We know from \rt{SeqC} we know $\<h, env_S, e\> \trans \<h', env_S', e'\>$.
        
        From \rt{WTC} we have:
        \begin{itemize}
            \item $\Lambda;\Delta \vddash e;e'' : t' \triangleright \Lambda';\Delta'$
        \end{itemize}
        
        From \rt{TSeq} we know:
        \begin{itemize}
            \item $\Lambda';\Delta \vddash e : t \triangleright \Lambda'';\Delta''$
            \item $\Lambda'';\Delta'' \vddash e'' : t' \triangleright \Lambda';\Delta'$
        \end{itemize}
        
        From Lemma \ref{lemma:subexpression_welltyped} we know $\Lambda;\Delta\vddash \<h, env_S, e\> : t \triangleright \Lambda'';\Delta''$.
        
        From our induction hypothesis we have
        
        \[
            \exists \Lambda^{(3)},\Delta^{(3)} \vddash \<h', env_S', e'\> : t \triangleright \Lambda^S;\Delta''
        \]
        where $\Lambda''\lambdaeq{o} \Lambda^S$. 
        
        From \rt{WTC}
        \begin{itemize}
            \item $\Lambda^{(3)}\vddash h'$
            \item $envT_S^{(3)} \vdash^{h'}_{e'} env_S'$
            \item $envT_O^{(3)} \vhpdash e'$
            \item $\Lambda^{(3)};\Delta^{(3)} \vddash e' : t \triangleright \Lambda^S;\Delta''$
            \item $\Lambda^{(3)}\vhpdash \vv{D}$
        \end{itemize}
        
        We now show $\Lambda^{(3)};\Delta^{(3)}\vddash \<h', env_S', e';e''\>  : t' \triangleright \Lambda^{S'}; \Delta''$ by showing the premises of \rt{WTC}.
        
        \begin{itemize}
            \item $\Lambda^{(3)}\vddash h'$ follows from assumption
            \item $envT_S^{(3)} \vdash^{h'}_{e';e''} env_S'$ Follows from \rt{WTP-Seq}
            \item $envT_O^{(3)} \vhpdash e'$ is trivial since $\objects{e';e''}=\objects{e'}$ 
            \item $\Lambda^{(3)}\vhpdash \vv{D}$
        \end{itemize}
        
        We need to show $\Lambda^{(3)};\Delta^{(3)}\vddash e';e'' : t' \triangleright \Lambda^{S'};\Delta'$. We do this using \rt{TSeq}. 
        
        We know $\Lambda^{(3)};\Delta^{(3)} \vddash e' : t \triangleright \Lambda^S;\Delta''$ and $\neg\lin(t)$. Since $e;e''$ is well formed, then we know that $\sf{returns}(e'')=0$, we can use Lemma \ref{lemma:consistency} to conclude $\Lambda^{S};\Delta''\vddash e'' : t' \triangleright \Lambda^{S'};\Delta'$ where $\Lambda^{S'} \lambdaeq{o}\Lambda'$.
        
        In total we can conclude $\Lambda^{(3)};\Delta^{(3)} \vddash \<h', env_S', e';e''\> : t' \triangleright \Lambda^{S'}; \Delta'$.
    \end{subproof}
    
    \item[\rt{SwFC}]
    
    \begin{subproof}
    Assume $\<h, env_S, \switchf{e}\> \trans \<h', env_S', \switchf{e'}\>$ and $\Lambda; \Delta \vddash \<h, env_S, \switchf{e}\> : t \triangleright \Lambda'; \Delta'$ where $\Delta = envT_O \cdot envT_S$.
    
    From \rt{SwC} we have:
    \begin{itemize}
        \item $\<h, env_S, e\> \trans \<h', env_S', e'\>$
    \end{itemize}
    
    From \rt{WTC} we have:
    \begin{itemize}
        \item $\Lambda; \Delta \vddash \switchf{e} : t \triangleright \Lambda'; \Delta'$
    \end{itemize}
    
    From \rt{TSwF} we have:
    \begin{itemize}
        \item $\Lambda;\Delta \vddash e : L \triangleright \Lambda'',o.f \mapsto C\<t'\>[(\<l_i : u_i\>_{l_i \in L})^{\vv{E}}]];\Delta''$
        \item $\forall l_i \in L . \ \Lambda'', o.f \mapsto C\<t'\>[u_i^{\vv{E}}];\Delta'' \vddash e_i : t \triangleright \Lambda';\Delta'$
    \end{itemize}
    
    From Lemma \ref{lemma:subexpression_welltyped} we know that $\Lambda; \Delta \vddash \<h, env_S, e\> : L \triangleright \Lambda'',o.f \mapsto C\<t'\>[(\<l_i : u_i\>_{l_i \in L})^{\vv{E}}]]; \Delta''$.
    
    Then from our induction hypothesis we have $$\exists \Lambda^{(3)}; \Delta^{(3)} \ . \ \Lambda^{(3)}; \Delta^{(3)} \vddash \<h', env_S', e'\> : L \triangleright \Lambda^S; \Delta''$$ 
    
    where $\Delta = envT_O^{(3)} \cdot envT_S^{(3)}$ and $\Lambda'',o.f \mapsto C\<t'\>[(\<l_i : u_i\>_{l_i \in L})^{\vv{E}}]] \lambdaeq{o} \Lambda^S$.
    
    From \rt{WTC} we have:
    \begin{itemize}
        \item $\Lambda^{(3)} \vddash h'$
        \item $envT_O^{(3)} \vhpdash e'$
        \item $envT_S^{(3)} \vhpdash_{e'} env_S'$
        \item $\Lambda^{(3)} \vhpdash \vv{D}$
        \item $\Lambda^{(3)}; \Delta^{(3)} \vddash e' : L \triangleright \Lambda^S; \Delta''$
    \end{itemize}
    
    We can now show $\Lambda^{(3)}; \Delta^{(3)} \vddash \<h', env_S', \switchf{e'}\> : t \triangleright \Lambda^{S'}; \Delta'$ by showing the premises of \rt{WTC}.
    \begin{itemize}
        \item $\Lambda^{(3)} \vddash h'$ from assumption
        \item $envT_S^{(3)} \vhpdash_{\switchf{e'}} env_S'$ from \rt{WTP-Sw}
        \item $envT_O^{(3)} \vhpdash \switchf{e'}$ from assumption since $\sf{objects}(\switchf{e'}) = \sf{objects}(e')$
        \item $\Lambda^{(3)} \vhpdash \vv{D}$ from assumption
    \end{itemize}
    
    We now have to show that $\switchf{e'}$ is well formed, we have that $\sf{returns}(e_i) = 0$. From Lemma \ref{lemma:consistency} we know that $\Lambda^{S'} \lambdaeq{o} \Lambda'$
    
    We now have that all premises of \rt{TSwF} are satisfied and we can conclude $\Lambda^{(3)}; \Delta^{(3)} \vddash \<h', env_S', \switchf{e'}\> : t \triangleright \Lambda^{S'}; \Delta'$.

    \end{subproof}
    
    \item[\rt{SwPC}]
    
    \begin{subproof}
    Assume $\<h, env_S, \switchp{e}\> \trans \<h', env_S', \switchp{e'}\>$ and $\Lambda;\Delta \vddash \<h, env_S, \switchp{e}\> : t \triangleright \Lambda'; \Delta'$ where $\Delta = envT_O \cdot envT_S$.
    
    From \rt{SwC} we have $\<h, env_S, e\>\trans \<h', env_S', e'\>$.
    
    From \rt{WTC} we know that $\Lambda;\Delta\vddash \switchp{e} : t \triangleright \Lambda';\Delta'$
    
    From \rt{TSwP} we know
    \begin{itemize}
        \item $\Lambda;\Delta \vddash e : L \triangleright \Lambda''; \Delta''' \cdot (o, [x \mapsto C\<t'\>[(\<l_i : u_i\>_{l_i \in L})^{\vv{E}}]])$
        \item $\forall l_i \in L . \ \Lambda'';\Delta'''\cdot (o, [x \mapsto C\<t'\>[u_i^{\vv{E}}]]) \vddash e_i : t \triangleright \Lambda';\Delta'$
    \end{itemize}
    
    From Lemma \ref{lemma:subexpression_welltyped} we know that $\Lambda;\Delta\vddash \<h, env_S, e\> : L \triangleright \Lambda''; \Delta''' \cdot (o, [x \mapsto C\<t'\>[(\<l_i : u_i\>_{l_i \in L})^{\vv{E}}]]) $.
    
    Then from our induction hypothesis we have
    
    \[
    \exists \Lambda^{(3)};\Delta^{(3)} . \ \Lambda^{(3)};\Delta^{(3)} \vddash \<h', env_S', e'\> : L \triangleright \Lambda^S;\Delta''' \cdot (o, [x \mapsto C\<t'\>[(\<l_i : u_i\>_{l_i \in L})^{\vv{E}}]])
    \]
    
    where $\Lambda'' \lambdaeq{o} \Lambda^S$
    
    From \rt{WTC} we have:
    \begin{itemize}
        \item $\Lambda^{(3)} \vddash h'$
        \item $envT_S^{(3)} \vdash^{h'}_{e'} envT_S'$
        \item $envT_O^{(3)} \vhpdash e'$
        \item $\Lambda^{(3)};\Delta^{(3)} \vddash e' : L \triangleright \Lambda^S;\Delta''' \cdot (o, [x \mapsto C\<t'\>[(\<l_i : u_i\>_{l_i \in L})^{\vv{E}}]])$
        \item $\Lambda^{(3)} \vhpdash \vv{D}$
    \end{itemize}
    
    We can now show $\Lambda^{(3)};\Delta^{(3)} \vddash \<h', env_S', \switchp{e'}\> : t \triangleright \Lambda^{S'};\Delta'$ by showing the premises of \rt{WTC}.
    
    \begin{itemize}
        \item $\Lambda^{(3)} \vddash h'$ follows from assumption
        \item $envT_S^{(3)} \vhpdash_{\switchp{e'}} env_S'$ follows from \rt{WTP-Sw}
        \item $envT_O^{(3)} \vhpdash \switchp{e'}$ is trivial because $\objects{\switchp{e'}}=\objects{e'}$
        \item $\Lambda^{(3)}\vhpdash \vv{D}$ follows from assumption
    \end{itemize}
    
    We only need to show $\Lambda^{(3)};\Delta^{(3)} \vddash \switchp{e'} : t \triangleright \Lambda^{S'};\Delta'$.
    
    We now have that $\switchp{e'}$ is a well formed, we have that $\sf{returns}(e_i) = 0$. Hence From Lemma \ref{lemma:consistency} we have that $\Lambda^{S'} \lambdaeq{o} \Lambda'$
     
    Hence all premises of \rt{TSwP} are satisfied, and we can conclude 
    
    $\Lambda^{(3)};\Delta^{(3)} \vddash \<h', env_S', \switchp{e'}\> : t \triangleright \Lambda^{S'};\Delta'$.
    \end{subproof}
    
\end{enumerate}
\end{proof}
\subsection{Error freedom proof}
\label{app:errorfreedom}
\tagthm{37}

The second part of the safety theorem is error freedom. It tells us that a well-typed configuration does not have any of the runtime-errors we want to catch, that is, protocol deviation and null-dereferencing errors. The errors are described in the $\err$ relation, which shows how configurations exhibiting the unwanted behaviour must look. We then show that these configurations cannot be well-typed.

\begin{lemma}{(Error freedom).}
\label{thm:errorfreedom_app}
If $\exists \Lambda,\Delta.\ \Lambda;\Delta\vddash \<h, env_S, e\> : t \triangleright \Lambda';\Delta'$ then $\<h, env_S, e\>\noterr$
\end{lemma}
\begin{proof}
Induction on the structure of $\err$. For each error predicate, show that the configuration cannot be well typed. 

\begin{enumerate}[ncases]
\item[\rt{NullCall-1}]

\begin{subproof}
Assume that the configuration is well typed, that is: $\Lambda;\Delta \vddash \<h, (o, s) \cdot env_S, f.m(v)\> : t \triangleright \Lambda';\Delta'$ and that $\<h, (o, s) \cdot env_S, f.m(v)\>\err$. From \rt{WTC} we know that $\Lambda;\Delta \vddash f.m(v) : t \triangleright \Lambda';\Delta'$. Since no typing rules for values changes $\Lambda$, we know from \rt{TCallF} that  $\Lambda=\Lambda\{o.f \mapsto C\<t'\>[\U]\}$. But from \rt{WTC} we must have $\Lambda \vddash h$. From \rt{WTH} we have that $\Lambda(o).envT_F(f)=\sf{getType}(h(o).f, h)$. This however, is a contradiction since $\Lambda(o).envT_F(f)=C\<t\>[\U]$ but $\sf{getType}(h(o).f, h) = \bot$. Hence the configuration cannot be well typed.
\end{subproof}

\item[\rt{NullCall-2}]

\begin{subproof}
Assume that the configuration is well typed, that is: $\Lambda;envT_O \cdot envT_S \cdot (o, S) \vddash \<h, (o, [x \mapsto \sf{null}]) \cdot env_S, x.m(v)\> : t \triangleright \Lambda';\Delta'$ and that $\<h, (o, [x \mapsto \sf{null}]) \cdot env_S, f.m(v)\>\err$. From \rt{WTC} we know that $\Lambda;envT_O \cdot envT_S \cdot (o, S) \vddash x.m(v) : t \triangleright \Lambda';\Delta'$. Since no typing rules for values changes the parameter typing stack, we have that $S = [x \mapsto C\<t'\>[\U]]$. But from \rt{WTC} we know that $envT_S \vdash^{h}_{x.m(v)} env_S$, hence we know that $|env_S| = 1$ and that $C\<t'\>[\U] = \sf{getType}(\sf{null}, h)$. This however is a contradiction, since we know that $\sf{getType}(\sf{null}, h) = \bot$, hence the configuration cannot be well typed.
\end{subproof}

\item[\rt{MthdNotAv-1}]

\begin{subproof}
Assume that the configuration is well typed, that is: $\Lambda;\Delta \vddash \<h, (o, s) \cdot env_S, f.m(v)\> : t \triangleright \Lambda';\Delta'$ and that $\<h, (o, s) \cdot env_S, f.m(v)\>\err$. From \rt{WTC} we know that $\Lambda;\Delta \vddash f.m(v) : t \triangleright \Lambda';\Delta'$. Since no typing rules for values changes $\Lambda$, we know from \rt{TCallF} that  $\Lambda=\Lambda\{o.f \mapsto C\<t'\>[\U]\}$ and that $\U \trans[m] \W$. But since $\Lambda\vddash h$, we must have that $\sf{getType}(h(o).f, h) = \Lambda(o).envT_F(f)$. But then we must have that $h(h(o).f).usage \trans[m] \W$, hence we have a contradiction, and the configuration cannot be well typed.
\end{subproof}

\item[\rt{MthdNotAv-2}]

\begin{subproof}
Assume that the configuration is well typed, that is: $\Lambda;envT_O \cdot envT_S \cdot (o, S) \vddash \<h, (o, [x \mapsto o']) \cdot env_S, x.m(v)\> : t \triangleright \Lambda';\Delta'$ and that $\<h, (o, [x \mapsto o']) \cdot env_S, f.m(v)\>\err$. From \rt{WTC} we know that $\Lambda;envT_O \cdot envT_S \cdot (o, S) \vddash x.m(v) : t \triangleright \Lambda';\Delta'$. Since no typing rules for values changes $envT_S$, we know from \rt{TCallP} that $S = [x \mapsto C\<t'\>[\U]]$ and that $\U \trans[m] \W$. But from \rt{WTP} we know that $|env_S|=1$ and that $\sf{getType}(o', h) = C\<t'\>[\U]$. This however is a contradiction, since we know that $\U \trans[m] \W$ but $h(o').\sf{usage}\nottrans[m]$, hence the configuration cannot be well typed.
\end{subproof}

\item[\rt{FldErr}]

\begin{subproof}
Assume that the configuration is well typed, that is: $\Lambda;\Delta \cdot (o, S) \vddash \<h, (o, s) \cdot env_S, f\> : t \triangleright \Lambda';\Delta'$ and that $\<h, (o, s) \cdot env_S, f\>\err$. From \rt{WTC} we know that $\Lambda;\Delta \vddash f : t \triangleright \Lambda';\Delta'$. This could have been concluded using either \rt{TLinFld} or \rt{TNoLFld}. Common for both of these cases is that $\Lambda(o).f$ is defined i.e. that the field $f$ exists on the object $o$. From \rt{WTH} we have that $h(o).\sf{fields} = \Lambda(o).\sf{fields}$. But we know that $f \in \Lambda(o).\sf{fields}$ but $f \not\in \dom{env_F}$, hence we have a contradiction and the configuration cannot be well typed. 
\end{subproof}

\item[\rt{IfCErr}, \rt{FldCErr}, \rt{SeqCErr}, \rt{CallCErr}, \rt{SwCErr}]

\begin{subproof}
Since we know from Lemma \ref{lemma:subexpression_welltyped} that $\Lambda;\Delta \vddash \<h, env_S, e\> : t \triangleright \Lambda';\Delta'$ , by our induction hypothesis we have $\<h, env_S, e\>\noterr$, which is a contradiction with the premise of rule used to conclude $c\err$, hence the configuration $c$ cannot be well typed.

\end{subproof}

\item[\rt{RetCErr}]

\begin{subproof}
Assume that $\Lambda;\Delta \vddash \<h, env_S \cdot (o, s), \return{e}\> : t \triangleright \Lambda';\Delta'$, then from \rt{WTC} and \rt{TRet} we have that $\Delta = envT_O \cdot envT_S \cdot (o', S)$. But then from Lemma 36 we know that $\Lambda;envT_O \cdot envT_S\vddash \<h, env_S, e\> : t' \triangleright \Lambda'';\Delta''$. But again this contradicts our induction hypothesis that $\<h, env_S, e\>\noterr$, since we have from the premise of \rt{RetCErr} that $\<h, env_S, e\>\err$. Hence our original configuration cannot be well typed.
\vspace{-0.5cm}
\end{subproof}
\end{enumerate}
\end{proof}



\end{document}